\documentclass{article}
\usepackage{amsmath}
\usepackage{amsfonts,amsmath,amsthm,amssymb}
\usepackage{mathtools}

\usepackage[sort, numbers]{natbib}
\usepackage{geometry}
\usepackage{pgf,pgfplots}
\usepackage{tikz}
\usepackage{caption}
\usepackage{subcaption}
\usepackage[]{enumitem}
\usepackage{bigints}
\usepackage{float}

\usepackage[colorlinks,citecolor=blue]{hyperref}

\newtheorem{lem}{Lemma}[section]
\newtheorem{cor}[lem]{Corollary}
\newtheorem{prop}[lem]{Proposition}
\newtheorem{remark}[lem]{Remark}
\newtheorem{thm}[lem]{Theorem}
\newtheorem{defn}[lem]{Definition}
\newtheorem{eg}[lem]{Example}

\newtheorem{assum}[lem]{Assumption}

\newtheorem*{prop*}{Proposition}
\newtheorem*{thm*}{Theorem}
\newtheorem*{def*}{Definition}
\newtheorem*{lem*}{Lemma}

\newcommand{\R}{\mathbb{R}}
\renewcommand{\P}{\mathbb{P}}
\newcommand{\I}{\mathbb{I}}
\newcommand{\E}{\mathbb{E}}
\newcommand{\F}{\mathcal{F}}
\newcommand{\N}{\mathbb{N}}

\makeatletter
\renewenvironment{proof}[1][\proofname] {\par\pushQED{\qed}\normalfont\topsep6\p@\@plus6\p@\relax\trivlist\item[\hskip\labelsep\bfseries#1\@addpunct{.}]\ignorespaces}{\popQED\endtrivlist\@endpefalse}
\makeatother

\renewcommand{\( }{\left(}
\renewcommand{\)}{\right)}

 \DeclareMathOperator*{\diver}{div}
 \allowdisplaybreaks

\begin{document}

\title{Robust Asymptotic Growth in Stochastic Portfolio Theory under Long-Only Constraints\footnote{Acknowledgments: We would like to thank Scott Robertson and Johannes Ruf for helpful discussions. The first author also acknowledges the support of the Natural Sciences and Engineering Research Council of Canada (NSERC).}}

\author{David Itkin\thanks{Department of Mathematical Sciences, Carnegie Mellon University, Wean Hall, 5000 Forbes Ave, Pittsburgh, Pennsylvania 15213, USA, \url{ditkin@andrew.cmu.edu}. Corresponding author.} \and Martin Larsson\thanks{Department of Mathematical Sciences, Carnegie Mellon University, Wean Hall, 5000 Forbes Ave, Pittsburgh, Pennsylvania 15213, USA, \url{martinl@andrew.cmu.edu}.}}
\maketitle

\begin{abstract}
We consider the problem of maximizing the asymptotic growth rate of an investor under drift uncertainty in the setting of stochastic portfolio theory (SPT). As in the work of Kardaras and Robertson \cite{kardaras2018ergodic} we take as inputs (i) a Markovian volatility matrix $c(x)$ and (ii) an invariant density $p(x)$ for the market weights, but we additionally impose long-only constraints on the investor. Our principal contribution is proving a uniqueness and existence result for the class of \textit{concave functionally generated portfolios} and developing a finite dimensional approximation, which can be used to numerically find the optimum. In addition to the general results outlined above, we propose the use of a broad class of models for the volatility matrix $c(x)$, which can be calibrated to data and, under which, we obtain explicit formulas of the optimal unconstrained portfolio for any invariant density. 
\end{abstract}

\bigskip



\bigskip 
\noindent \textbf{Data Availability Statement:} Data sharing not applicable to this article as no datasets were generated or analysed during the current study.
\section{Introduction}
In this paper we tackle the problem of maximizing an investor's long-term growth in the setting of Stochastic Portfolio Theory (SPT) subject to model uncertainty and long-only constraints. SPT was introduced by Fernholz in \cite{fernholz1999diversity,fernholz2002stochastic} as a descriptive theory with the goal of explaining observable market phenomena. An important observation in SPT, and the main motivation of our paper, is that the ranked relative market capitalizations have remained remarkably stable over time across different US equity markets (see Chapter 5 in \cite{fernholz1999diversity}). Many probabilistic models have been developed to capture this phenomenon and, within the financial literature, can mostly be grouped into two types; (i) rank-based models such as the Atlas models developed in \cite{fernholz2002stochastic,banner2005atlas,ichiba2011hybrid} and (ii) volatility-stabilized models discussed in \cite{fernholz2005relative,banner2008short,cuchiero2019polynomial,pickova2014generalized}.

	 In this work, we allow for Knightian uncertainty regarding the dynamics of the relative market capitalizations and study a robust growth optimization problem in this setting. Portfolio optimization under model uncertainty has previously been studied in the SPT literature, such as in \cite{fernholz2011optimal} where the authors study optimal arbitrages and in \cite{cuchiero2019cover} where Cover's theorem on universal portfolios established in \cite{cover2011universal} is extended to the SPT framework. In \cite{kardaras2012robust}, the authors study a robust long-term growth maximization problem under model uncertainty in a very general setting, of which SPT is a special case. The authors of \cite{kardaras2012robust} solve for the robust optimal growth rate and trading strategy, which they find to be the principal eigenvalue and eigenfunction, respectively, of a certain operator. They carry out this analysis by fixing the volatility structure of the market weights, but allow for drift uncertainty. The paper \cite{bayraktar2013robust} studies a similar problem, but additionally allows for uncertainty with regards to the volatility specification.
	 
	  Recently, the authors of \cite{kardaras2012robust}, in \cite{kardaras2018ergodic}, studied an ergodic version of the problem from \cite{kardaras2012robust} where in addition to fixing the volatility structure they took as an input an invariant measure. In this paper, we borrow the setup from \cite{kardaras2018ergodic} by considering a financial market with $d$ market weights and we take as inputs two estimable quantities; (1) a Markovian instantaneous volatility matrix $c(x)$ for the relative market capitalizations and (2) an invariant density $p(x)$ capturing the aforementioned stability of the ranked market weights. The authors of \cite{kardaras2018ergodic} maximize an investor's asymptotic growth rate over an admissible class of probability measures; that is they study the quantity
$$ \lambda := \sup\limits_{V \in \mathcal{V}} \inf\limits_{\P \in \Pi} g(V;\P)$$
where $V \in \mathcal{V}$ is an investor's wealth process, $\P \in \Pi$ is an admissible probability measure and $g(V;\P)$ is the asymptotic growth rate. All of these objects are precisely defined in Section \ref{setup}. It was found in \cite{kardaras2018ergodic} that the optimal strategy is functionally generated in the sense of \cite{fernholz2002stochastic} and can be characterized as the solution of a variational problem. It turns out that in certain important examples, such as the case of volatility-stabilized markets that we consider in Section~\ref{a_motivating_example}, the optimal strategy requires heavy short-selling. As such, it may not be implementable by money mangers due to various institutional and risk considerations. 

The main contribution  of the paper is in Section \ref{long_only} where we solve an analogous robust asymptotic growth problem under \textit{long-only constraints}. A priori, one may consider several different long-only constrained robust optimization problems; in descending order of generality they are to admit (1) arbitrary long-only portfolios, (2) long-only portfolios in feedback form, (3) functionally generated long-only portfolios, and (4) portfolios generated by concave generating functions. In this paper we consider problem (4) and believe it to be the natural problem to consider in this setting, both due to well-posedness of the mathematical problem under the level of generality considered in this setup and due to practical considerations regarding the implementation of the optimal strategy. We defer the discussion of this subtle point until Section \ref{long_only_discussion}, where we provide a detailed explanation for this choice. 

Similarly to \cite{kardaras2018ergodic} we obtain an existence and uniqueness result and are able to characterize the optimal strategy via a constrained variational problem. The constrained variational problem is susceptible to a finite dimensional approximation and numerical implementation. We outline the procedure to exploit this approximation scheme in Section \ref{fin_dim} and provide numerical simulations of the constrained and unconstrained optimal strategies in Section \ref{numerical}.

A second contribution concerns the introduction of a large class of tractable models, under which we can obtain explicit solutions to the robust growth-optimization problem studied in \cite{kardaras2018ergodic} for any dimension $d$. We propose the use of a very general form for the volatility matrix $c(x)$ given by 
$$ c_{ij}(x) = \begin{cases}
- f_{ij}( x^{-ij})f_i(x^i)f_j(x^j)g(x) & i \ne j\\
\sum\limits_{k\ne i} f_{ik}( x^{-ik})f_i(x^i)f_k(x^k)g(x) & i = j
\end{cases} \quad 1 \leq i,j \leq d,$$
where $x^{-ij}$ is a $d-2$ dimensional vector obtained from $x$ by removing the $i^{\text{th}}$ and $j^{\text{th}}$ components. The mild conditions required for the functions $g,f_i$ and $f_{ij}$ are precisely stated in Assumption \ref{function_assumption}. This specification generalizes the volatility structure seen in the aforementioned volatility-stabilized models (which can be obtained here by letting $f_i$ be the identity function and $f_{ij} \equiv 1$ for every $i,j$) and allows us to obtain \textit{explicit formulas} for both the optimal portfolio and $\lambda$. We additionally believe that the properties of this volatility structure may be of interest to other researchers who are studying diffusion processes on the simplex.

The paper is organized as follows. Section~\ref{setup} introduces the problem and summarizes the results of \cite{kardaras2018ergodic}. In Section~\ref{a_motivating_example} we apply these results to the example of volatility-stabilized markets and observe that the optimal strategy requires heavy short-selling from the investor. Section~\ref{long_only_main_section} is dedicated to studying problem (4) mentioned above and contains our main results. In Section~\ref{rank_based_section} we connect the results of the previous sections to rank-based models. In Section~\ref{long_only_discussion} we discuss in detail the various differences between the problems (1)-(4). For the two dimensional case we are able to explicitly solve problem (2) in Section~\ref{2d_long_only}. Section~\ref{model} introduces the general specification of the instantaneous volatility matrix $c(x)$, from which we obtain explicit formulas for the optimal strategy and growth rate in the unconstrained case. Section~\ref{examples} contains several examples. In particular, in Section~\ref{gen_vol_stab} we consider generalized volatility-stabilized markets introduced in \cite{pickova2014generalized} and obtain explicit formulas for the invariant density and growth-optimal trading strategy in such models -- to the best our knowledge these quantities were previously unknown for this class of models. Section~\ref{applications} is devoted to applications.  In Section~\ref{fin_dim}, we develop a finite dimensional approximation to the constrained variational problem in Section~\ref{long_only} using novel results regarding exponentially concave functions, which we believe to be of independent interest. The performance of the various optimal strategies are then illustrated using numerical simulations in Section~\ref{numerical}. Finally, in Appendix~\ref{mart_prob_app} we rigorously examine Assumption~\ref{finite_assumption}, regarding the inputs $c$ and $p$, and develop sufficient conditions for it to hold, while Appendix~\ref{proof_appendix} contains the proof of Theorem~\ref{concave_invariance_thm}.

\section{Setup and Preliminaries} \label{setup}
 We consider a financial market in the context of SPT, where for $d \geq 2$ the \textit{market weights} $X = (X^1,\dots,X^d)^\top$ are the assets. Under the assumption that no market weights vanish they take values in the open simplex
 \begin{equation}\label{simplex}
 \Delta^{d-1}_+ := \left\{x \in \R^d: x^i > 0 \text{ for every } i \in\{1,\dots,d\}, \sum\limits_{i=1}^d x^i = 1\right\}.
 \end{equation} As such, we work on the path space $\Omega = C([0,\infty);\Delta^{d-1}_+)$ with Borel $\sigma$-algebra $\F$ induced by the topology of locally uniform convergence. We denote the coordinate process by $X$ and consider the filtration $\mathbb{F}$ as the right-continuous enlargement of the natural filtration generated by $X$.
  One way to interpret the market weights is to view them as the relative capitalizations of a collection of stocks; that is, if $S = (S^1,\dots,S^d)^\top$ represent the capitalizations of $d$ stocks, then $X^i = S^i/(S^1+\dots+S^d)$. Note that the coordinate process $X$ takes values in $\Delta^{d-1}_+$ by definition.

As defined in \eqref{simplex}, $\Delta^{d-1}_+$ is a relatively open $d-1$ dimensional subset of $\R^d$ and can be identified with an open set $E \subset \R^{d-1}$ via the transformation $$(x^1,\dots,x^{d-1}) \mapsto (x^1,\dots,x^{d-1},1-\sum_{i=1}^{d-1}x^i).$$
 Any function $\psi$ on $E$ can be associated with a function $\phi$ on  $\Delta^{d-1}_+$ via $\phi(x^1,\dots,x^d) := \psi(x^1,\dots,x^{d-1})$. Conversely, given a function $\phi$ on $\Delta^{d-1}_+$ one can define the associated function $\psi$ on $E$ via $\psi(x^1,\dots,x^{d-1}) := \phi(x^1,\dots,x^{d-1},1-\sum_{i=1}^{d-1}x^{i})$. The analysis in this paper requires us to compute derivatives of functions defined on the simplex and it is convenient to use this identification to make this precise. This motivates the following definition.

\begin{defn}
	Let $k \geq 1$, $\gamma \in (0,1]$ and $U \subseteq \R$  be given. We denote by $C^{k,\gamma}(\Delta^{d-1}_+;U)$ the set of all functions $\phi: \Delta^{d-1}_+ \to U$ such that the associated function $\psi:E \to U$ given by $\psi(x^1,\dots,x^{d-1}) := \phi(x^1,\dots,x^{d-1},1-\sum_{i=1}^{d-1}x^i)$ satisfies $\psi \in C^{k,\gamma}(E;U)$.
\end{defn}
 Next note that every function $\phi \in C^{1,\gamma}(\Delta^{d-1}_+; U)$ can be extended to a $C^{1,\gamma}$ function on an open set in $\mathbb R^d$ containing $\Delta^{d-1}_+$. The chain rule then yields
$$
\partial_i \psi(x^1,\ldots,x^{d-1}) = \partial_i \phi(x) - \partial_d \phi(x), \quad x \in \Delta^{d-1}_+, \quad i=1,\ldots,d-1.
$$
This can be viewed as differentiability on $\Delta^{d-1}_+$ regarded as a differentiable manifold, because the vector fields $\partial_i - \partial_d$, $i=1,\dots,d-1$ span the tangent space at each point of the simplex. As such, the above formula does not depend on the chosen extension of $\phi$ and it is easy to verify that all expressions in the paper involving derivatives of functions on the simplex can be unambiguously computed using an arbitrary such extension. Moreover, all integrals over $\Delta^{d-1}_+$ in the sequel should be understood  with respect to the pushforward of the Lebesgue measure on $\R^{d-1}$ under the map  $(x^1,\dots,x^{d-1}) \mapsto (x^1,\dots,x^{d-1},1-\sum_{i=1}^{d-1}x^i)$.

   Following the setup of \cite{kardaras2018ergodic} we 
  consider a family of probability measures, under which $X$ is a continuous semimartingale. We allow for drift uncertainty and as inputs take an instantaneous volatility matrix $c$ and invariant density $p$ which satisfy the following standing assumption:
  
  \begin{assum} \label{input_assumption}
  	For a fixed constant $\gamma \in (0,1]$ we have
  	\begin{enumerate}[label = ({\roman*})]
  		\item $c \in C^{2,\gamma}(\Delta_+^{d-1};\mathbb{S}^d_{+})$ is such that $c(x)$ has rank $d-1$ and $\boldsymbol{1} \in \text{Ker}(c(x))$ for every $x \in \Delta_+^{d-1}$,
  		\item $p \in C^{2,\gamma}(\Delta_+^{d-1};(0,\infty))$ is such that $\int_{\Delta_+^{d-1}} p = 1$.
  	\end{enumerate} 
  \end{assum}
  Here, and in what follows, $\boldsymbol{1}$ is the vector of all ones in $\R^d$. Our goal is to find the optimal growth rate in the worst-case model and the portfolio generating it under a class of admissible probability measures: 
  
  \begin{defn} \label{Pi_def}
  	Given inputs $(c,p)$  satisfying Assumption \ref{input_assumption} we define the set $\Pi$ consisting of all probability measures $\P$ on $(\Omega,\F)$ for which the following hold:
  	\begin{enumerate}[label = ({\roman*})]
  		\item $X$ is a $\P$-semimartingale with covariation process $\langle X\rangle = \int_0^\cdot c(X_t)dt$; $\P$-a.s.
  		\item For all Borel measurable functions $h$ on $\Delta_+^{d-1}$ with $\int_{\Delta_+^{d-1}}h^+p < \infty$  we have that 
  		$$ \lim\limits_{T \to \infty} \frac{1}{T}\int_0^T h(X_t)dt = \int_{\Delta_+^{d-1}}hp; \quad \P\text{-a.s.}$$
  		\item The laws of $\{X_t\}_{t \geq 0}$ under $\P$ are tight on $\Delta^{d-1}_+$. That is for every $\epsilon > 0$ there is a compact set $K \subset \Delta^{d-1}_+$ such that $\P(X_t \in K) \geq 1-\epsilon$ for every $t \geq 0$.
  	\end{enumerate}
  \end{defn}
  
  With regards to investment strategies, we will consider portfolios $\pi = (\pi^1,\dots,\pi^d)^\top$ which are processes that are $X$-integrable with respect to each $\P \in \Pi$ and satisfy
  \begin{equation} \label{portfolio_defn}
  \pi^1_t + \dots + \pi^d_t = 1, \quad t \geq 0.
  \end{equation}
  If $\pi_t^i \geq 0$ for every $t \geq 0$ and $i \in \{1,\dots,d\}$ we say that $\pi$ is \textit{long-only}.
  The \textit{wealth process} induced by a portfolio $\pi$ is given by
  \begin{equation} \label{wealth}
  \frac{dV^\pi_t}{V^\pi_t} := \sum\limits_{i=1}^d \pi^i_t\frac{dX^i_t}{X^i_t}; \quad V_0^\pi = 1
  \end{equation} and we let $\mathcal{V}$ be the set of all such wealth processes. 
 
  We stress that we only consider strategies that are fully invested in the market and do not incorporate a bank account in the model. This is captured by \eqref{portfolio_defn} and is standard in the SPT literature. Note the condition \eqref{portfolio_defn} is not a restriction when solving the robust growth-optimization problem. This is because $X^\top \boldsymbol{1} = 1$, so if $\rho$ is a candidate trading strategy which does not satisfy \eqref{portfolio_defn}, then by setting $\pi_t = \rho_t + (1-\rho_t^\top \boldsymbol{1})X_t$ we have $V^\pi \equiv V^\rho$ and $\pi_t^\top\boldsymbol{1} \equiv 1$. Removing condition \eqref{portfolio_defn} would therefore not enlarge the set of achievable wealth processes.
  
   The process $V^\pi$ is interpreted as a relative wealth process with respect to the market portfolio. Indeed, if one takes a portfolio $\pi$ invested in capitalizations $S$ with corresponding wealth process $W^\pi$ and sets $\Sigma = S^\top \boldsymbol{1}$ then $V^\pi$, as defined in \eqref{wealth}, is given by $W^\pi/\Sigma$.  
  As such, the component $\pi^i$ represents the proportion of wealth invested in asset $i$.  We refer the reader to \cite[Chapter 1]{fernholz2002stochastic} for a more detailed discussion of these conventions.
  
  Now we define the \textit{asymptotic growth rate} (in probability) corresponding to a wealth process $V \in \mathcal{V}$ and a measure $\P \in \Pi$ as 
  $$g(V;\P) := \sup\left\{\gamma \in \R: \lim\limits_{T\uparrow \infty}\P\(T^{-1}\log V_T \geq \gamma\) = 1\right\}.$$
  Our main goal is to identify the \textit{optimal robust asymptotic growth rate} 
   \begin{equation} \label{lambda}
  \lambda := \sup\limits_{V \in \mathcal{V}} \inf\limits_{\P \in \Pi} g(V;\P)
  \end{equation}
  together with the growth-optimal portfolio that achieves this growth rate. In \cite{kardaras2018ergodic}, the authors were able to identify $\lambda$ under additional assumptions on the inputs $c$ and $p$. To formulate these assumptions we define the operator $L$ to be the generator of $X$ under driftless dynamics; that is \begin{equation}\label{L_def}
  Lf(x) := \frac{1}{2}\sum\limits_{i,j=1}^d c_{ij}(x)\partial_{ij}f(x) 
  \end{equation} 
  for all $f \in C^2(\R^d)$. Next, denote by $c^{-1}\diver c(x)$ any vector $y$ that satisfies the linear equation 
  \begin{equation} \label{divc}
  c(x)y= \diver c(x)
  \end{equation} where $\diver c_i (x) = \sum_{j}\partial_j c_{ij}(x)$ for every $i = 1,\dots,d$ and $x \in \Delta^{d-1}_+$. Note that since $\boldsymbol{1} \in \text{Ker}(c(x))$ we equivalently have that $\diver c_i(x) = \sum_{j \ne i} (\partial_j - \partial_i) c_{ij}(x)$, which is unambiguous. By Assumption \ref{input_assumption}(i) we have that $\diver c(x)$ belongs to the range of $c(x)$ for every $x \in \Delta^{d-1}_+$ so \eqref{divc} has a solution which is unique up to multiples of $\bf 1$. The equations in the sequel will not depend on the particular choice of solution to \eqref{divc}. However, occasionally $c^{-1}\diver c$ can be represented as the gradient of a function, in which case we use this representative.
  
  \begin{assum} \label{finite_assumption}
  	In what follows write $\ell := \frac{1}{2}\(\nabla \log p + c^{-1}\diver c\)$.  We then assume
  	\begin{enumerate}[label = ({\roman*})]
  		\item $\int_{\Delta_{+}^{d-1}} \ell^\top c \ell p < \infty$,
  		\item $\int_{\Delta_{+}^{d-1}} \(\diver{pc\ell}\)^+ < \infty$,
  		\item There exists a non-explosive solution $\tilde \P$ to the generalized martingale problem associated to the operator $ L^R$ given by
  		$$  L^R f :=  Lf + \ell^\top c \nabla f, \quad  f \in C^2(\R^d).$$
  	\end{enumerate}
  \end{assum}
The martingale problem in Assumption~\ref{finite_assumption}(iii) is called a \textit{generalized} martingale problem as the coordinate process may, in general, explode in finite time. As such, a rigorous formulation of the generalized martingale problem requires one to define it locally on an increasing sequence of compact sets. We state the precise formulation of (iii) and develop sufficient conditions on $(c,p)$ which ensure the solution $\tilde \P$ is non-explosive in Appendix \ref{mart_prob_app}.

  Assumption~\ref{finite_assumption}(iii) is important to enforce to guarantee the non-degeneracy of \eqref{lambda}. If this assumption fails to hold then, as discussed in \cite{kardaras2018ergodic},  \cite[Theorem 6.6.2 (ii)]{pinsky1995positive} implies that there are no time-homogeneous diffusions with laws in $\Pi$. Moreover, \cite[Proposition 1.9]{kardaras2018ergodic} shows that when $d=2$ if Assumption~\ref{finite_assumption}(iii) fails then either $\Pi = \emptyset$ or $\lambda = \infty$. As such, if the coordinate process explodes under the measure $\tilde \P$ then the robust growth-optimization problem may not be well-posed. Conversely, if Assumption~\ref{finite_assumption}(iii) does hold then $\tilde \P$ (which will be a candidate for the worst-case measure as will become clear in Theorem~\ref{KR_main}) is in $\Pi$. Moreover, under this law the dynamics of $X_t$ are given by
\begin{equation} \label{worst_case_dynamics}
dX_t = c\ell(X_t) dt + \sigma(X_t)dW_t
\end{equation}
for some Brownian motion $W$ and where $\sigma(x)$ is the positive definite square root of $c(x)$. Assumption~\ref{finite_assumption}(i) ensures that the growth-optimal portfolio under $\tilde \P$ has finite asymptotic growth rate.

To state the main results of \cite{kardaras2018ergodic}, which identify both $\lambda$ and the growth-optimal portfolio, we need the notion of \textit{functionally generated portfolios} and the \textit{master formula} developed in \cite{fernholz2002stochastic}.

\begin{defn} \label{def_func_gen}
 Let $\pi$ be a portfolio and let $G:\R^d \to (0,\infty)$ be a function that is continuous on $\Delta^{d-1}_+$ such that $\log G(X)$ is a semimartingale. If we have the representation 
	\begin{equation} \label{func_gen_def}
	\log V_T^\pi = \log G (X_T) + \Gamma_T 
	\end{equation}for some finite variation process $\Gamma$ with $\Gamma_0 = 0$ then we call $\pi$ a functionally generated portfolio with generating function $G$ and drift process $\Gamma$ and we write $\pi_G$ for $\pi$. 
	
	 If a portfolio is of the form $\pi_t = \pi(X_t)$ for some deterministic function $\pi(\cdot)$, we say that it is in feedback form. 
\end{defn}

\begin{thm}[Master Formula] \label{master_formula}
	Let $G:\R^d \to (0,\infty)$ be a function that is continuous on $\Delta^{d-1}_+$ such that $G(X)$ is a semimartingale. Assume there exist locally bounded measurable functions $g_i: \Delta_+^{d-1} \to \R$ for $i = 1,\dots,d$ and a finite variation process $Q$ with $Q(0) = 0$ such that 
	\[
	d\log G(X_t) = \sum\limits_{i=1}^d g_i(X_t)dX^i_t + dQ_t.\]
	Then $G$ functionally generates the portfolio $\pi_G$ with weights given by
	\[
	\frac{\pi_G^i(x)}{x^i} = g_i(x) + 1 - \sum\limits_{j=1}^d x^jg_j(x); \quad i = 1,\dots,d.
	\]
\end{thm}
\begin{remark} \label{func_gen_remark}
	If the function $G$ from the previous theorem is $C^2$ then it follows that the assumptions of the theorem are satisfied with $g_i(x) = \partial_i \log G(x)$ and $dQ_t = L\log G(X_t)dt$ where $L$ is given by \eqref{L_def}. Thus the corresponding portfolio weights become
\[
	\frac{\pi^i_G(x)}{x^i} = \partial_i \log G(x) + 1 - \sum\limits_{j=1}^d x^j\partial_j \log G(x); \quad i = 1,\dots,d
\]
	and we have the representation $\log V_T^\pi = \log G(X_T) + \Gamma_T$ where $\Gamma_T = \int_0^T\frac{-LG}{G}(X_t)dt$.
\end{remark}
We are now ready to state the main results from \cite{kardaras2018ergodic}.
\begin{thm}\label{KR_main}
	Under Assumptions \ref{input_assumption} and \ref{finite_assumption} we have that 
	\[
	\lambda = \frac{1}{2}\int_{\Delta_{+}^{d-1}} \nabla \hat u^\top c\nabla \hat u p
	\]
	where $\hat u: \Delta^{d-1}_+ \to \R$ is the unique (up to an additive constant) minimizer of the variational problem
	\begin{equation} \label{variational}
\min\limits_{u \in C^2(\Delta^{d-1}_+)} \int_{\Delta_{+}^{d-1}} (\ell-\nabla u)^\top c(\ell - \nabla u)p.
	\end{equation}
	Moreover, the optimal portfolio is functionally generated by $\hat G:= \exp (\hat u)$ and it is the growth-optimal portfolio under the unique measure (up to initial distribution) $\P_{\hat u}$ characterized by $X$ having the semimartingale decomposition
	\begin{equation} \label{hat_P_dynamics}
	dX_t = c  \nabla \hat u (X_t)dt + \sigma(X_t)dW_t
	\end{equation}
	 under $ {\P_{\hat u}}$. \end{thm}
The measure ${\P_{\hat u}}$ can be interpreted as the worst-case measure since the growth-optimality of $\pi_{\hat G}$ under ${\P_{\hat u}}$ yields
$$\lambda = \sup\limits_{V \in \mathcal{V}}\inf\limits_{\P \in \Pi} g(V;\P) = g(V^{\pi_{\hat G}};{\P_{\hat u}}) = \sup\limits_{V \in \mathcal{V}} g(V;{\P_{\hat u}}).$$

In the case that $\ell = \nabla \log  R$ is a gradient the solution to \eqref{variational} is clearly given (up to an additive constant) by $\hat u = 
\log R$ and the measure ${\P_{\hat u}}$ is equal to the measure $\tilde{\P}$; that is it solves the generalized martingale problem for $L^R$ in Assumption \ref{finite_assumption}(iii). If $\ell$ is not a gradient then the authors in \cite{kardaras2018ergodic} proved that $\hat u$ solves the Euler--Lagrange equation
\begin{equation} \label{euler_lagrange}
\diver(pc\nabla \hat u) = \diver(pc \ell).
\end{equation}
The measures $\tilde \P$ and ${\P_{\hat u}}$ may differ in this case.
\section{Long-Only Constraints and Concave Generating Functions} \label{long_only}
In this section, we will study the robust growth-optimization problem introduced in the previous section under long-only constraints. Section~\ref{a_motivating_example} contains an illustrative example, which motivates the study of the long-only problem. In Section~\ref{long_only_main_section} we formulate and solve a long-only robust growth-optimization problem. Section \ref{rank_based_section} applies our results to models of ranked market weights. In Section~\ref{long_only_discussion} we discuss in detail our choice of long-only constraints. 

\subsection{A Motivating Example}\label{a_motivating_example}
Consider for $x \in \Delta^{d-1}_+$ the specification
\begin{align*}
c_{ij}(x) & = \begin{cases}
-x^ix^j & i \ne j\\
x^i(1-x^i) & i = j
\end{cases},&  
p(x) = \frac{\Gamma(ad)}{\Gamma(a)^d}\prod_{i=1}^d (x^i)^{a-1}
\end{align*}
where $a > 1$ is a fixed parameter. In Example~\ref{dirichlet_example}, where we examine a generalized version of this example, it is shown that this choice of ($c,p)$ satisfies Assumptions~\ref{input_assumption} and \ref{finite_assumption}, so that the results of Theorem~\ref{KR_main} are valid. Moreover, this choice of $c$ matrix is a special case of the specification \eqref{c_def} considered later in Section~\ref{model} so it follows from Theorem~\ref{divc_lemma} that $c^{-1}\diver c(x) = \sum_{i=1}^d\nabla \log (x^i)$, from which we obtain that $\ell = \nabla \log R$ where 
$$R(x) = \prod_{i=1}^d (x^i)^{a/2}.$$
The equation \eqref{worst_case_dynamics} for the dynamics of $X$ under the worst case measure $\tilde \P$ become
\begin{equation} \label{vol_stabilized}
dX^i_t  = \frac{ad}{2}\(\frac{1}{d}- X^i_t\)dt + \sqrt{X^i_t}(1-X^i_t)dW^i_t - \sum_{j \ne i} X^i_t\sqrt{X^j_t}dW^j_t, \quad  i = 1,\dots,d.
\end{equation}
Equation \eqref{vol_stabilized} corresponds to the dynamics of the market weights in the volatility-stabilized market introduced in \cite{fernholz2005relative}. This is a pertinent example in SPT as it exhibits an important feature empirically observed in real-life equity markets; namely that assets with smaller capitalizations tend to have bigger volatilites and growth rates than the larger assets. We refer the reader to \cite{fernholz2005relative,pickova2014generalized} for a more in depth discussion of volatility-stabilized markets. 

Using the explicit expression for $\ell$ we obtain from Theorems~\ref{master_formula} and~\ref{KR_main} that the robust growth-optimal portfolio is given by 
\begin{align} \label{vol_stabilized_portfolio}
\pi^i_{\hat G}(x) & = \frac{1}{2}\(a + x^i(2-da)\), & i =1,\dots,d,
\end{align}
 where $\hat G = R$. This portfolio is an affine combination of the market portfolio and a constant-weighted portfolio. It can be viewed as a leveraged position in a constant-weighted portfolio, funded by a short position in the market portfolio (indeed, $ 2-da < 0$ because $a > 1$). As such it is rather aggressive, especially when $da$ is large.
  Though this strategy does not require the investor to short-sell the smallest assets (which can be infeasible for even the most sophisticated investors), it does prescribe the short-selling of large market weights. Indeed, whenever $x^i > \frac{a}{da-2}$, the portfolio takes a short position in asset $i$. Roughly speaking, essentially independent of $a$, one would have to short approximately the 100 largest stocks in the $S\&P \ 500$ (under current capitalizations) to implement this strategy. Since many investors have restrictions on short-selling it is desirable to tackle the robust growth-optimization problem under \textit{long-only constraints}. 

\subsection{The Long-Only Robust Growth-Optimization Problem} \label{long_only_main_section}
Motivated by the substantial short-selling observed in Section~\ref{a_motivating_example}, we will introduce long-only constraints. We will do so by restricting to a large and well-behaved class of long-only functionally generated portfolios -- those generated by concave functions. It is well-known that concave functions generate long-only portfolios (see Proposition~3.1.15 in \cite{fernholz2002stochastic} for the $C^2$ case and Theorem~3.7 in \cite{karatzas2017trading} for the general case).

Another, different, attempt would have been to optimize over \emph{all} long-only portfolios; however without imposing any additional structure it is unclear how to proceed. Indeed, the techniques employed in \cite{kardaras2018ergodic} heavily rely on the functionally generated structure of the candidate optimal strategy. In Section \ref{2d_long_only} we are able to explicitly solve the problem among all portfolios in feedback form in the $d=2$ case. We observe that the resulting portfolio is functionally generated by a function which is (typically) \textit{not} $C^2$. The absence of this feature along with the lack of an explicit candidate solution causes difficulty in generalizing the results to higher dimension; a further discussion of these points is conducted in Section~\ref{long_only_discussion}.
 
In our analysis of portfolios generated by concave functions it will be convenient to emphasize the logarithm of the generating function. This is visible already in Theorem \ref{KR_main}, but now becomes essential. As such, we introduce the concept of \textit{exponentially concave functions} \cite{pal2016geometry}:
 \begin{defn}
 	We say that a function $\phi: \Delta^{d-1}_+ \to \R$ is exponentially concave if $e^\phi$ is concave and denote the set of all such functions by $\mathcal{E}$.
 \end{defn}
 In this section we restrict our attention to the class of wealth processes produced by concave portfolio generating functions,
 $$\mathcal{V}^{\mathcal{E}} := \{V^{\pi_G}\in \mathcal{V}: \log G \in \mathcal{E}\}.$$ We consider the problem of characterizing
 \begin{equation} \label{exp_concave_problem}
 \lambda_{\mathcal{E}} :=\sup\limits_{V \in \mathcal{V}^\mathcal{E}} \inf\limits_{ \P \in \Pi} g(V;\P).
 \end{equation} 
To state our main result we will need to introduce a subset of the set $\Pi$.

\begin{defn} \label{tilde_pi_def}
	Given $\P \in \Pi$ let 
	$$dX_t = b_t^\P dt + \sigma(X_t)dW_t$$
	be the semimartingale decomposition of $X$ under $\P$. Then we define the class
	$\tilde \Pi$ to be the measures $\P \in \Pi$ such that \begin{equation} \label{drift_assumption}
		\limsup\limits_{T \to \infty}\frac{1}{T}\int_{0}^T \E[|(b_t^{\P})^i|^q\I_K(X_t)]dt< \infty
	\end{equation} 
	for some $q > 1$, every $i \in \{1,\dots,d\}$ and every compact set $K \subset \Delta^{d-1}_+$.
\end{defn}

Condition \eqref{drift_assumption} is a weak asymptotic integrability condition on the drift of $X$. In particular if $b^{\P}_t$ is dominated by a locally bounded function of $X_t$ then \eqref{drift_assumption} holds. Hence both the measure $\tilde \P$, under which $X$ has dynamics given by \eqref{worst_case_dynamics}, and the worst case measure $\P_{\hat u}$, under which $X$ has dynamics given by \eqref{hat_P_dynamics}, are members of $\tilde \Pi$.

Before proceeding we will make the following standing convention to simplify the notation and arguments to come. Every exponentially concave function $\phi$ is concave and is therefore continuous and almost everywhere differentiable. Though there may exist a nonempty Lebesgue null set $N$ such that the superdifferential $\partial \phi(x)$ is larger than a singleton for $x \in N$, we will denote by $\nabla \phi$ any version of $\partial \phi$. With this notation established we are now ready to state our main result.
 
 \begin{thm} [Main Theorem] \label{main_thm}
 Under Assumptions \ref{input_assumption} and \ref{finite_assumption} we have that 
 	\begin{equation} \label{lambda_EC}
 	\lambda_{\mathcal{E}} = \frac{1}{2}\int_{\Delta^{d-1}_+}  \(\ell^\top c \ell - (\ell - \nabla \hat \phi)^\top c (\ell - \nabla \hat \phi)\)p
 	\end{equation} where $\hat \phi$ is the unique (up to an additive constant) solution of 
 	\begin{equation} \label{min_exp}
 	\inf_{\phi \in \mathcal{E}} \int_{\Delta^{d-1}_+} (\ell - \nabla \phi)^\top c (\ell -\nabla \phi)p.
 	\end{equation} Moreover, the robust growth-optimal strategy is functionally generated by the concave function $\hat G = \exp \hat \phi$. It satisfies $g(V^{\pi_{\hat G}};\P) \geq \lambda_{\mathcal{E}}$ for every $\P \in \tilde \Pi$ and  $g(V^{\pi_{\hat G}};\tilde \P) = \lambda_{\mathcal{E}}$.
 \end{thm}
 
This result establishes an equation for the optimal robust asymptotic growth rate in terms of a concave function $\hat G$. Using Theorem~\ref{partialE_thm} below we are able to find $C^2$ concave functions $\hat G_n$ such that $\lim_{n \to \infty} \inf_{\P \in \Pi} g(V^{\pi_{\hat G_n}};\P) = \lambda_{\mathcal{E}}$.  However, we are not able to show that the limiting portfolio $\pi_{\hat G}$ achieves growth rate $\lambda_{\mathcal{E}}$ for the class of measures $\Pi$; we are only able to show that $\pi_{\hat G}$ achieves at least growth rate $\lambda_{\mathcal{E}}$ for measures in the smaller class $\tilde \Pi$. Our proof of this lower bound relies on the drift assumption \eqref{drift_assumption} and it is unclear to us if it can be relaxed. We view this as a mild technical restriction and still refer to the portfolio $\pi_{\hat G}$ as robust growth-optimal.
 
The remainder of this section will be dedicated to proving this theorem. To accomplish this we will find the growth-optimal portfolio $\pi_{\hat G}$ under the measure $\tilde \P$ using the tractable properties of this measure. $\pi_{\hat G}$ is the candidate robust optimal portfolio and the growth rate it achieves under $\tilde \P$ will serve as an upper bound for $\lambda_{\mathcal{E}}$. To establish the lower bound we will use an approximation argument. This crucially relies on the following remarkable fact:\ a portfolio
generated by a $C^2$ function has the same growth rate for every $\P \in \Pi$. This ``built-in" robustness property of portfolios generated by $C^2$ functions was already identified in \cite{kardaras2018ergodic}, but nevertheless we present a self-contained proof in Lemma~\ref{C2_lem} below. By approximating arbitrary concave functions with $C^2$ concave functions we will be able to show, using Lemma~\ref{C2_lem}, that the robust asymptotic growth rate $\lambda_{\mathcal{E}}$ cannot be greater than the maximal growth rate achievable under $\tilde \P$. This will allow us to establish \eqref{lambda_EC} and \eqref{min_exp}. Using the aforementioned growth rate invariance property for portfolios generated by $C^2$ functions we establish in Theorem~\ref{concave_invariance_thm} below that a portfolio generated by a concave function achieves at least the same growth rate under any measure $\P \in \tilde \Pi$ as it does under $\tilde \P$. This will establish the second part of the theorem and complete the proof.

Recall the operator $L$ defined in \eqref{L_def}. A key property of concave functions that we use is that $-LG$ is a nonegative function for any $C^2$ concave function $G$. This allows us to establish Theorem~\ref{partialE_thm} below. The absence of this property for non-concave functions is one of the mathematical difficulties we face in solving problem \eqref{lambda_EC} for a more general class of long-only portfolios. On the other hand, from a practical standpoint, this is a desirable property for a portfolio to possess. Indeed, \cite[Theorem~3.7]{karatzas2017trading} shows that the $\Gamma$ process in \eqref{func_gen_def} is increasing for a concave generating function $G$. When $G$ is $C^2$ this is already visible by Remark~\ref{func_gen_remark} since $-LG$ is non-negative. 

Now we turn our attention to proving Theorem~\ref{main_thm}. To start, we establish an explicit form for the growth rate of portfolios under $\tilde \P$.
\begin{lem} \label{worst_case_growth} 
	Let $\pi = \pi(x)$ be a portfolio in feedback form. Then we have that 
	$$g(V^\pi;\tilde \P) = \frac{1}{2}\int_{\Delta^{d-1}_+} \(\ell^\top c \ell(x)- \(\ell(x)-\frac{\pi(x)}{x} \)^\top c(x) \(\ell(x)-\frac{\pi(x)}{x} \)\)p(x)dx$$
	where $\frac{\pi(x)}{x} =(\frac{\pi^1(x)}{x^1},\dots,\frac{\pi^d(x)}{x^d})$.
\end{lem}
\begin{proof}
	Let $h(x) = \frac{\pi(x)}{x}$. Under the measure $ \tilde \P$, the dynamics of $X_t$ are given by \eqref{worst_case_dynamics}. Applying It\^o's lemma to \eqref{wealth} we see that the normalized $\log$ wealth is given by
	\begin{align}
	\frac{1}{T}\log V_T^\pi & = \frac{1}{T}\int_0^T \(h^\top c \ell - \frac{1}{2} h^\top c h\)(X_t)dt + \frac{1}{T}\int_0^Th^\top\sigma(X_t) dW_t \nonumber\\
	& = \frac{1}{T}\int_0^T \frac{1}{2}\ell^\top c \ell(X_t)dt + \frac{1}{T}\int_{0}^T \ell^\top\sigma(X_t)dW_t -\frac{1}{T}\int_0^T \frac{1}{2}(\ell-h)^\top c (\ell-h)(X_t)dt \label{logv}\\
	& \hspace{6.6 cm} -\frac{1}{T}\int_0^T(\ell-h)^\top \sigma(X_t)dW_t \nonumber
	\end{align}
	where $\sigma(x)$ is the unique positive-definite square root of $c(x)$.
	We claim that 
	$$ \lim\limits_{T \to \infty}\(\frac{1}{T}\int_0^T \frac{1}{2}\ell^\top c \ell(X_t)dt + \frac{1}{T}\int_{0}^T \ell^\top\sigma(X_t)dW_t\) = \frac{1}{2}\int_{\Delta^{d-1}_+} \ell^\top c \ell p; \quad \tilde \P\text{-a.s.}$$
	Indeed, the first term converges to the required limit by Assumption \ref{finite_assumption}(i) and the ergodic property, while the second term $N_T := \int_0^T\ell^\top \sigma (X_t)dW_t$ is a local martingale with quadratic variation  $\langle N \rangle_T = \int_0^T \ell^\top c \ell (X_t)dt$. By the ergodic property $\lim_{T \to \infty} T^{-\alpha} \langle N \rangle_T = 0$ for any $\alpha > 1$, so it follows that $\lim_{T\to \infty} T^{-1}N_T = 0$; $ \tilde{\P}$-a.s. (see for example Lemma 1.3.2 in \cite{fernholz2002stochastic}). 
	
	Next let $M_T = \int_0^T (\ell - h)^\top\sigma(X_t)dW_t$. First suppose that $\int_{\Delta^{d-1}_+} (\ell-h)^\top c (\ell-h)p < \infty$. Then by the ergodic property $$\lim_{T\to \infty} T^{-1} \langle M \rangle_T = \int_{\Delta^{d-1}_+}(\ell-h)^\top c (\ell-h)p; \quad \tilde \P\text{-a.s.}$$ so by the same argument as for $N_T$ it follows that $\lim_{T \to \infty} T^{-1}M_T = 0$. Hence from \eqref{logv} we see that $$\lim_{T\to \infty}T^{-1}\log V_T^\pi = \frac{1}{2}\int_{\Delta^{d-1}_+} ( \ell^\top c \ell -  (\ell-h)^\top c (\ell-h))p; \quad \tilde \P\text{-a.s.}$$ which establishes the required growth rate.  
	
	Now suppose that $\int_{\Delta^{d-1}_+} (\ell-h)^\top c (\ell-h)p = \infty$, so that in particular $\langle M\rangle_T \to \infty$; $\tilde \P$-a.s.\ as $T \to \infty$. By the Dambis--Dubins--Schwarz Theorem there exists a standard Brownian motion $B$, such that $M_T = B_{\langle M \rangle_T}$.
	We then have
	\begin{align}
	\limsup\limits_{T \to \infty}\frac{1}{T}\log V_T^\pi & = \limsup\limits_{T \to \infty} \(\frac{1}{T}\int_0^T \frac{1}{2}\ell^\top c \ell(X_t)dt + \frac{1}{T}\int_{0}^T \ell^\top \sigma(X_t)dW_t - \frac{1}{T}\(B_{\langle M \rangle_T} + \frac{1}{2}\langle M \rangle_T\)\)\nonumber \\
	& = \frac{1}{2} \int_{\Delta^{d-1}_+} \ell^\top c \ell p(x)dx - \liminf\limits_{T \to \infty}\frac{\langle M \rangle_T}{T}\(\frac{B_{\langle M \rangle_T}}{\langle M \rangle_T} + \frac{1}{2}\)\nonumber\\
	& = \label{step4_liminf} \frac{1}{2}\int_{\Delta^{d-1}_+} \ell^\top c \ell p(x)dx -\frac{1}{2}\liminf\limits_{T\to \infty}\frac{\langle M \rangle_T}{T}
	\end{align}
	where we used the strong law of Brownian motion in the last step.
	Setting $$K_N = \{x \in \Delta^{d-1}_+: (\ell-h)^\top c(\ell-h)(x) \leq N \} $$ we have
	$$ \frac{\langle M \rangle_T}{T} = \frac{1}{T}\int_0^T (\ell-h)^\top c(\ell-h)(X_t)dt \geq \frac{1}{T}\int_0^T (\ell-h)^\top c(\ell-h)\I_{K_N}(X_t)dt.$$
	Here $(\ell-h)^\top c(\ell-h)\I_{K_N}$ is a bounded function, so the ergodic property yields
	$$\liminf\limits_{T \to \infty} \frac{\langle M \rangle_T}{T} \geq \int_{\Delta^{d-1}_+} (\ell-h)^\top c(\ell-h)p\I_{K_N}.$$  Now taking $N \to \infty$ we see from \eqref{step4_liminf} that
	$$\limsup\limits_{T \to \infty} \frac{1}{T}\log V_T^\pi \leq \frac{1}{2}\int_{\Delta^{d-1}_+} (\ell^\top c \ell - (\ell-h)^\top c(\ell-h))p = - \infty.$$
	It follows that $g(V^\pi;\tilde \P) = -\infty$ in this case, completing the proof.
\end{proof}
Next, using this lemma we are able to establish a useful characterization of the growth rate for the portfolio $\pi_G$ whenever $G$ is twice continuously differentiable. Note that the following lemma does not require $G$ to be a concave function.
\begin{lem} \label{C2_lem}
	Let $G \in C^2(\Delta^{d-1}_+; (0,\infty))$ be given such that $\int_{\Delta^{d-1}_+} \(\frac{LG}{G}\)^+p  < \infty$. Then we have for every $\P \in \Pi$ that 
	\begin{equation} \label{IBP_growth} g(V^{\pi_G};\P) = \int_{\Delta^{d-1}_+} \frac{-LG}{G}p = \frac{1}{2} \int_{\Delta^{d-1}_+} \(\ell^\top c \ell - (\ell - \nabla \phi)^\top c (\ell - \nabla \phi)\)p
	\end{equation}
	where $\phi = \log G$ and the operator $L$ is given by \eqref{L_def}.
\end{lem}
\begin{proof}
	We fix $\P \in \Pi$ and note that by Remark \ref{func_gen_remark} we have the representation
	$$ \log V_T^{\pi_G} = \log G(X_T) + \int_0^T \frac{-LG}{G}(X_t)dt$$
	since $G$ is $C^2$.
	The ergodic property yields
	$$\lim_{T\to \infty}\frac{1}{T}\int_0^T \frac{-LG}{G}(X_t)dt = \int_{\Delta^{d-1}_+} \frac{-LG}{G}p; \quad \P\text{-a.s.}$$
	so it just remains to examine the term $\log G(X_T)$. We claim that $T^{-1}\log G(X_T) \to 0$ in probability as $T \to \infty$. Indeed for any $\delta > 0$, by Definition~\ref{Pi_def}(iii) we can find a compact set $K_\delta$ such that $\P(X_T \in K_\delta) \geq 1-\delta$ for every $T \geq 0$. We then see for every $\delta,\epsilon> 0$ that
	\begin{align*}
	\P\bigg(\frac{|\log G(X_T)|}{T} > \epsilon\bigg) & \leq  \P\bigg(\frac{|\log G(X_T)|}{T} > \epsilon ;X_T \in K_\delta\bigg) +  \delta.
	\end{align*}  
	By continuity $\log G$ is bounded on $K_\delta$, so the first term on the right hand side vanishes when we send $T \to \infty$. Thus, we obtain 
	$\lim_{T \to \infty} \P(T^{-1}|\log G(X_T)| > \epsilon) \leq \delta$ for every $\delta,\epsilon > 0$. Now sending $\delta \to 0$ yields the claimed convergence in probability and proves the first equality in \eqref{IBP_growth}
\begin{equation} \label{gr1}
g(V^{\pi_G};\P) = \int_{\Delta^{d-1}_+} \frac{-LG}{G}p
\end{equation} for every $\P \in \Pi$. To establish the second equality in \eqref{IBP_growth} we note that by Lemma \ref{worst_case_growth} together with the master formula Theorem \ref{master_formula} we have
	\begin{equation}\label{gr2} g(V^{\pi_G};\tilde \P) = \frac{1}{2} \int_{\Delta^{d-1}_+} \(\ell^\top c \ell - (\ell - \nabla \phi)^\top c (\ell - \nabla \phi)\)p.
	\end{equation}
	Comparing \eqref{gr1} when $\P = \tilde \P$  with \eqref{gr2} completes the proof.
\end{proof}
\begin{remark}
	\begin{enumerate}
		\item As previously mentioned, the assumptions of the previous lemma do not require $G$ to be a concave function. The assumptions, however, do hold when $G$ is concave, since under concavity we have $-LG \geq 0$.
		\item Note that this proof relies on the asymptotic growth rate being defined via limits in probability. Indeed if $\log G$ is unbounded near the boundary of the simplex it is not clear whether or not $\log G (X_T)/T \to 0$ as $T \to \infty$ almost surely for every $\P \in \Pi$. This consideration explains the chosen asymptotic growth rate definition rather than one involving almost sure limits.
		\item 
	For any $G \in C_c^\infty(\Delta^{d-1}_+;(0,\infty))$ it follows from integration by parts that 
	$$ \int_{\Delta^{d-1}_+} \frac{-LG}{G}p = \frac{1}{2} \int_{\Delta^{d-1}_+} \(\ell^\top c \ell - (\ell - \nabla\phi)^\top c (\ell - \nabla \phi)\)p.$$ However it is not immediately clear why this identity would extend to the class of $C^2$ functions. It is the probabilistic Assumption \ref{finite_assumption}(iii), which indirectly influences the boundary behaviour of $c$ and $p$ and allows us to conclude that this integration by parts formula holds for all $C^2$ functions satisfying $\int_{\Delta^{d-1}_+} \(\frac{LG}{G}\)^+p  < \infty$.
	\end{enumerate}
\end{remark}

Note that Lemma \ref{worst_case_growth} together with the master formula indicates that maximizing the growth rate under $\tilde \P$ over all portfolios generated by concave functions is equivalent to the minimization problem \eqref{min_exp}.
To study this problem we introduce the space 
\[\mathcal{H}^{c,p} = \left\{v: \Delta^{d-1}_+ \to \R^d \quad \Big| \quad \|v\|^2_{\mathcal{H}^{c,p}} := \int_{\Delta^{d-1}_+ }v^\top c v p < \infty\right\}\bigg / \sim
\]
where we say that $v \sim w$ if and only if there exists a measurable function $h: \Delta^{d-1}_+ \to \R$ such that $v(x) = w(x) + h(x)\boldsymbol{1}$ for almost every $x$. It is clear that this is a Hilbert space when equipped with the inner product
$$(v,w)_{\mathcal{H}^{c,p}} := \int_{\Delta^{d-1}_+} v^\top c wp$$ by the nondegeneracy of $cp$ guaranteed by Assumption \ref{input_assumption}. With this new notation the minimization problem \eqref{min_exp} becomes
\begin{equation}\label{EC} 
\inf_{ \nabla \phi \in \partial \mathcal{E}} \|\ell - \nabla \phi\|^2_{\mathcal{H}^{c,p}}
\end{equation}
where $\partial \mathcal{E} \subseteq \mathcal{H}^{c,p}$ consists of those maps that arise as supergradients of exponentially concave functions,
\begin{equation} \label{partialE} \partial \mathcal{E} := \{v \in \mathcal{H}^{c,p}: v = \nabla \phi \text{ for some } \phi \in \mathcal{E}\}.
\end{equation}
The use of the symbol ``$\partial$" is a mnemonic for differential. To prove existence and uniqueness for \eqref{EC} we first need a technical lemma regarding exponentially concave functions. We start with a definition:

 \begin{defn} \label{MCM_defn}
	Let $U \subseteq \R^d$ and let $T: U \rightrightarrows  \R^d$ be a multi-valued map taking non-empty values. We say that $T$ is multiplicatively cyclically monotone (MCM) if for every $m \in \N$ and for all $\{x_i\}_{i=0}^m \subseteq U$ with $ x_0 = x_m$ (called a cycle) we have for all values $y_i \in T(x_i)$ that 
	\begin{enumerate}[label = ({\roman*})]
		\item $\langle y_i, x_{i+1}-x_i \rangle \geq -1$ for all $i = 0,\dots,m-1$,
		\item $\prod_{i=0}^{m-1} \(1 + \langle y_i, x_{i+1}-x_i \rangle\)\geq 1$.
	\end{enumerate}
\end{defn}
It turns out that superdifferentials of exponentially concave functions satisfy MCM and, conversely, that multi-valued maps possessing the MCM property are subsets of superdifferentials of exponentially concave functions. This was proven in \cite{pal2016geometry}, however we will need a slight refinement of this result, which only assumes that the property holds outside a Lebesgue null set.

\begin{lem} {} \label{MCM} {}
	\begin{enumerate} [label = (\arabic*)]
		\item Let $\phi: \Delta^{d-1}_+ \to \R$ be an exponentially concave function. Then the superdifferential $\partial \phi$ is MCM.
		\item 	Let $T: \Delta^{d-1}_+ \to \R^d$ be a multi-valued map taking non-empty values. Suppose there exists a Lebesgue null set $N \subseteq \Delta^{d-1}_+$ such that $T$ is MCM on $\Delta^{d-1}_+ \setminus N$. Then there exists an exponentially concave function $\phi: \Delta^{d-1}_+ \to \R$ such that $T(x) \subseteq \partial \phi(x)$ for every $x \in \Delta^{d-1}_+ \setminus N$.
	\end{enumerate}
\end{lem}
\begin{proof}
	We refer the reader to Proposition 4 in \cite{pal2016geometry} for the proof of (1) and note that the proof of (2) presented here is a minor modification of that proof as well. 
 Fix $x_0 \in \Delta^{d-1}_+ \setminus N$ and 
	define for $z \in \Delta_+^{d-1}$
	$$\Phi(z) := \inf\left\{\prod_{j=0}^{m-1} (1 + \langle y_j,x_{j+1} - x_j\rangle\right\}
	$$
	where the infimum is taken over all $m \in \N$, $\{x_j\}_{j=1}^{m}$, and $\{y_j\}_{j=0}^{m-1}$ such that $x_j \in \Delta^{d-1}_+\setminus N$ for $j \in \{1,\dots,m-1\}$, $x_m = z$ and $y_j \in T(x_j)$ for $j \in \{0,\dots,m-1\}$.
	Since $\Phi$ is the pointwise infimum of a family of affine functions it is concave on $\Delta_+^{d-1}$. By condition (i) of MCM we see that $\Phi$ is nonnegative on $\Delta^{d-1}_+ \setminus N$ and by condition (ii) it is clear that $\Phi(x_0) =1$. Hence by concavity and continuity of $\Phi$ we must have that $\Phi$ is strictly positive everywhere on $\Delta_+^{d-1} $.
	
 Now let $z \in \Delta^{d-1}_+ \setminus N$ and $q \in \Delta_+^{d-1}$ be given and choose $\alpha$ such that $\alpha > \Phi(z)$. Then by definition of $\Phi$ we can find an $m \geq 1$, $\{x_j\}_{j=1}^m \subseteq \Delta^{d-1}_+ \setminus N$ with $x_m = z$ and $y_j \in T(x_j)$ for $j = 0,\dots,m-1$ such that
	$$\prod\limits_{j=0}^{m-1}\(1 + \langle y_j, x_{j+1} - x_j \rangle\) < \alpha.$$ Set  $x_{m+1} = q$ and let $y_m \in T(z)$ be arbitrary. It then follows that
	\begin{equation} \label{Phi_ineq}
	\Phi(q)  \leq \prod\limits_{j=0}^{m} \(1 + \langle y_j, x_{j+1} - x_j \rangle\) < \alpha(1 + \langle y_m,q-z\rangle).
	\end{equation}
	To obtain \eqref{Phi_ineq} we used the definition of $\Phi$, the fact that $z \not \in N$ and that $1+\langle y_m,q - z\rangle \geq 0$. Note that for $q \not \in N$ this last property follows directly from property (i) of Definition~\ref{MCM_defn}, and so it extends to all $q \in \Delta^{d-1}_+$ by approximation since $N$ is a Lebesgue null-set.
	Sending $\alpha \downarrow \Phi(z)$ in \eqref{Phi_ineq} shows that
	$$\Phi(q)- \Phi(z) \leq  \langle \Phi(z)y_m,q-z \rangle.$$
	Since this holds for every $q \in \Delta^{d-1}_+$ and we have that $\Phi > 0$ on $\Delta_+^{d-1}$ we see by definition of superdifferential that $ \Phi(z)y_m \in \partial \Phi(z)$ and hence $y_m \in \partial\log \Phi(z)$ for every $z \in \Delta^{d-1}_+ \setminus N$. Since $y_m \in T(z)$ was arbitrary it follows that $T(z) \subseteq \partial\log \Phi(z)$ for every $z \in \Delta^{d-1}_+ \setminus N$.  Setting $\phi = \log \Phi$ completes the proof.
\end{proof}

With these preparations in hand we prove some properties regarding the set $\partial \mathcal{E}$.

\begin{thm} \label{partialE_thm}
	The set $\partial \mathcal{E}$ defined by \eqref{partialE} is a closed, convex and bounded set in $\mathcal{H}^{c,p}$. Moreover the set $$\partial \mathcal{E} \cap C^1 := \{v \in \mathcal{H}^{c,p} :v = \nabla \phi \text{ for some } C^2 \text{ exponentially concave function } \phi\}$$ is dense in $\partial \mathcal{E}$.
	
\end{thm}

\begin{proof}
	It is established in Proposition 5 of \cite{alirezaei2018exponentially} that $\mathcal{E}$, the set of exponentially concave functions, forms a convex set. From this it follows that $\partial \mathcal{E}$ is convex.
	
	To prove closedness, suppose that $\{\nabla \phi_n\}_{n \in \N}$ is a sequence in $\partial \mathcal{E}$ converging to some $v$ in $\mathcal{H}^{c,p}$. Then
	$$ \lim\limits_{n \to \infty}\int_{\Delta^{d-1}_+} (\nabla \phi_n - v)^\top c(\nabla \phi_n - v)p = 0$$ so that there is a subsequence $(n_k)_{k\in \N}$ such that $\nabla \phi_{n_k}$  converges to $v$ almost everywhere as $k \to \infty$.  We will verify that $v$ possesses the MCM property and so must be the supergradient of an exponentially concave function by Lemma \ref{MCM}.
	
	Let $N$ be the Lebesgue null set where the convergence does not take place. For any cycle $\{x_{j}\}_{j=0}^{m}\subseteq \Delta_+^{d-1}$ we have that $\langle \nabla \phi_{n_k}(x_j),x_{j+1} - x_{j}\rangle \geq -1$ and $\prod_{j=0}^{m-1}(1 + \langle \nabla \phi_{n_k}(x_j), x_{j+1} - x_j\rangle) \geq 1$ by part (1) of Lemma \ref{MCM}. Hence we see by passing to the limit in the above expressions that $v$ satisfies MCM on $\Delta_+^{d-1} \setminus N$. 
	By Lemma \ref{MCM} (2) we have that $v$ is a version of the superdifferential of an exponentially concave function; that is there exists a $\phi \in \mathcal{E}$ such that $v = \nabla \phi$, which establishes that $v \in \partial\mathcal{E}$. 
	
	To handle the boundedness claim, we will first show that $\partial \mathcal{E} \cap C^1$ is bounded and obtain that $\partial \mathcal{E}$ is bounded by proving the density assertion. Fix $\nabla \phi \in \partial \mathcal{E} \cap C^1$ and set $G = \exp \phi$. By concavity of $G$ we have that $-LG/G$ is a nonnegative function. Using Lemma \ref{C2_lem} we get 
	\begin{align*} 0 \leq \int_{\Delta^{d-1}_+}  \frac{-LG}{G}p & = \frac{1}{2} \int_{\Delta^{d-1}_+} \(\ell^\top c \ell - (\ell - \nabla \phi)^\top c (\ell - \nabla \phi)\)p \\
	& = (\nabla \phi, \ell)_{\mathcal{H}^{c,p}} - \frac{1}{2}\|\nabla \phi\|^2_{\mathcal{H}^{c,p}}.
	\end{align*}
	From this inequality and Cauchy--Schwarz we deduce
	$$ \frac{1}{2}\|\nabla \phi\|^2_{\mathcal{H}^{c,p}} \leq (\nabla \phi, \ell)_{\mathcal{H}^{c,p}} \leq \|\nabla \phi\|_{\mathcal{H}^{c,p}}\|\ell\|_{\mathcal{H}^{c,p}}.$$
	This yields the bound $\|\nabla \phi\|_{\mathcal{H}^{c,p}} \leq 2 \|\ell\|_{\mathcal{H}^{c,p}}$, which is finite by Assumption \ref{finite_assumption}(i). Thus $\partial \mathcal{E} \cap C^1$ is bounded.
	
	Now given a general $\nabla \phi \in \partial \mathcal{E}$ we will approximate it by members of $\partial \mathcal{E} \cap C^1$ in the following way.  Let $G = \exp \phi$, which is a positive concave function. Viewing $\Delta^{d-1}_+$ as an open set in $\R^{d-1}$, let
	$\psi \in C_c^\infty(\R^{d-1})$ be such that $\psi \geq 0$, $\text{supp}(\psi) \subseteq \Delta^{d-1}_+$, $\text{diam(supp(}\psi)) = 1/2$ and  $\int_{\R^{d-1}} \psi = 1$. We define $\psi_n(x) := n^{d-1} \psi(nx)$. For each $t \in (0,1)$ define a concave function $G^t$ on a $t$-neighbourhood of $\Delta^{d-1}_+$  by 
	$$G^t(x) := G((1-t)x + t \bar x);\quad \text{dist}(x,\Delta^{d-1}_+) < t,$$  where $\bar x = \frac{1}{d}\boldsymbol{1}$. We can then define the functions $G_n:= G^{1/n}*\psi_n$ on $\Delta^{d-1}_+$. They are all positive, concave and smooth. It follows that $G_n \to G$ pointwise on $\Delta^{d-1}_+$. Indeed, fix $x \in \Delta^{d-1}_+$ and choose a compact set $K$ such that $x \in K$ and supp($\psi_N(x-\cdot)) \subseteq K$ for some $N$ large enough. Using the fact that every concave function is locally Lipschitz continuous there is a constant $L > 0$ such that $|G(y) - G(z)| \leq L|y-z|$ for all $y,z \in K$. Then for all $n \geq N$ we estimate that 
	\begin{align*}|G(x) - G_n(x)| & \leq \int_{B(0,1)} |G(x) - G((1-n^{-1}y) + n^{-1}\bar x)|\psi_n(x-y) dy \\
	& \leq L \int_{B(0,1)} \((1-n^{-1})|x-y| + n^{-1}|x-\bar x| \psi_n(x-y)\)dy \\
	&\leq \frac{L}{2}\(\frac{1+n^{-1}}{n} + \frac{1}{n}\) \xrightarrow{n \to \infty} 0.
	\end{align*}  By Theorem 25.7 in \cite{rockafellar1970convex} it follows that $\nabla G_n \to \nabla G$ at every point where $G$ is differentiable. As a consequence we obtain that $\phi_n \to \phi$ and $\nabla \phi_n \to \nabla \phi$ almost everywhere where $\phi_n = \log G_n$. 
	
	By construction we have that $\nabla \phi_n \in \partial \mathcal{E} \cap C^1$ for every $n$. We have aleady shown that $\partial \mathcal{E} \cap C^1$ is norm bounded so we can find a subsequence $(n_k)_{k \in \N}$ such that $\nabla \phi_{n_k}$ converges weakly in $\mathcal{H}^{c,p}$. However, since $\nabla \phi_n \to \nabla \phi$ pointwise it follows that this weak limit must be $\nabla \phi$. From here we conclude that $\nabla \phi$ is in the weak closure of $\partial \mathcal{E} \cap C^1$. However, since $\partial \mathcal{E} \cap C^1$ is a convex set, its weak and strong closures coincide (see e.g.\ Corollary 2 in Chapter II.9 of \cite{schaefer1971locally}) so that $\nabla \phi \in \overline{\partial \mathcal{E} \cap C^1}$, where $\overline{\partial \mathcal{E} \cap C^1}$ denotes the strong closure of $\partial \mathcal{E} \cap C^1$ in $\mathcal{H}^{c,p}$. This establishes that $\partial \mathcal{E} = \overline{\partial \mathcal{E} \cap C^1}$ which proves the density claim. The boundedness of $\partial \mathcal{E}$ follows from the boundedness of $\partial \mathcal{E} \cap C^1$.
\end{proof}
\begin{cor} \label{long_thm}
	There exists a unique solution $\nabla \hat \phi$ to \eqref{EC}.
\end{cor}
\begin{proof}
By the previous theorem we have that $\partial \mathcal{E}$ is a closed convex set in $\mathcal{H}^{c,p}$ so the result follows from the Hilbert space projection theorem.
\end{proof}
Lastly, before proving Theorem~\ref{main_thm}, we establish that a portfolio generated by a concave function achieves at least the same growth rate under any measure $\P \in \tilde \Pi$ as it does under $\tilde \P$.
	\begin{thm} \label{concave_invariance_thm}
		Let $G: \Delta^{d-1}_+ \to (0,\infty)$ be a concave function. Then 
		$g(V^{\pi_G};\P) \geq g(V^{\pi_G};\tilde \P)$ for every $\P \in \tilde \Pi$.
	\end{thm}

The proof of Theorem~\ref{concave_invariance_thm} is contained in Appendix~\ref{proof_appendix}. 
We are now ready to prove our main theorem.
\begin{proof} [Proof of Theorem \ref{main_thm}]
From Lemma \ref{worst_case_growth} together with the master formula we know that $$g(V^{\pi_G};\tilde \P) = \frac{1}{2}\int_{\Delta^{d-1}_+} (\ell^\top c \ell - (\ell-\nabla \phi)^\top c (\ell-\nabla \phi))p$$ for any function $G = \exp \phi$ where $\phi \in \mathcal{E}$. By Corollary \ref{long_thm} we have that there is a unique solution $\nabla \hat \phi$ to \eqref{EC} so combining these results yields the upper bound
$$\lambda_{\mathcal{E}} \leq \sup\limits_{V \in \mathcal{V}^{\mathcal{E}}} g(V;\tilde \P) = \frac{1}{2}\int_{\Delta^{d-1}_+} (\ell^\top c \ell - (\ell-\nabla \hat \phi)^\top c (\ell-\nabla \hat \phi))p.$$

For the lower bound we use an approximation argument. By Theorem \ref{partialE_thm} we can find $C^2$ exponentially concave functions $\hat \phi_n$ such that $\nabla \hat \phi_n \to \nabla \hat \phi$ in $\mathcal{H}^{c,p}$ as $n \to \infty$. But then by Lemma~\ref{C2_lem} we see by setting $\hat G_n:= \exp \hat \phi_n$ that 
\begin{align*}
 \lambda_{\mathcal{E}} \geq \inf_{\P \in \Pi} g(V^{\pi_{\hat G_n}};\P) & = \frac{1}{2}\int_{\Delta^{d-1}_+} (\ell^\top c \ell - (\ell-\nabla \hat \phi_n )^\top c ( \ell-\nabla \hat \phi_n ))p
\end{align*}
Sending $n \to \infty$ yields the lower bound and proves \eqref{lambda_EC}. The final claim in the statement of the theorem follows from Theorem~\ref{concave_invariance_thm} and the fact that $g(V^{\pi_{\hat G}};\tilde {\P}) = \lambda_{\mathcal{E}}$ by Lemma~\ref{worst_case_growth}.
\end{proof}
\begin{remark}
We can establish an inequality for $\lambda_{\mathcal{E}}$ in terms of the norm of the optimum $\nabla \hat \phi$. To derive this bound, we define the function $g: [0,1] \to \R$  by $g(t) = (t\nabla \hat \phi,\ell)_{\mathcal{H}^{c,p}}  - \frac{1}{2}\|t\nabla \hat \phi\|^2_{\mathcal{H}^{c,p}}$, which is maximized at $t = 1$ by the optimality of $\hat \phi$. It follows that 
$$0 \leq g'(1) = (\nabla \hat \phi,\ell)_{\mathcal{H}^{c,p}} - \|\nabla \hat \phi\|^2_{\mathcal{H}^{c,p}}$$ from which obtain that $(\nabla \hat \phi,\ell)_{\mathcal{H}^{c,p}} \geq \|\nabla \hat \phi\|^2_{\mathcal{H}^{c,p}}$. Applying this to \eqref{lambda_EC} yields $\lambda_{\mathcal{E}} \geq \frac{1}{2}\|\nabla \hat \phi\|^2_{\mathcal{H}^{c,p}}$.
\end{remark}

\subsection{Connection to Rank-Based Models} \label{rank_based_section}

We are interested in extending the analysis from the previous section to a rank-based market model. Rank-based models play an important role in stochastic portfolio theory due to the observed empirical stability of the capital distribution curve of the ranked market weights. The asymptotic growth rate in the setting of rank-based models was studied in \cite{kardaras2018ergodic} and we obtain analogous results for the long-only problem. 

\begin{defn}
	\begin{enumerate} [label=(\roman*)]
	\item For $x \in \Delta_+^{d-1}$, we denote its rank vector by $x^{()} = (x^{(1)},\dots,x^{(d)})$ where $x^{(1)} \geq x^{(2)} \geq \dots \geq x^{(d)}$ and $\{x^{(k)}:k = 1,\dots,d\} = \{x^i:i = 1,\dots,d\}$.
	\item 	We define the ordered simplex $$\Delta^{d-1}_{+,\geq} := \left\{x \in \Delta_+^{d-1}: x^1 \geq x^2 \geq \dots \geq x^d\right\}$$ so that $x^{()} \in \Delta^{d-1}_{+,\geq}$ for every $x \in \Delta^{d-1}_+$. 
	\item For each $k \in \{1,\dots,d\}$ we define the function $r_k:\Delta^{d-1}_+ \to \{1,\dots,d\}$ via $r_k(x) = i$ if $x^{(k)} = x^{i}$ with ties broken by lexicographical ordering.
	\item We say a portfolio $\pi(\cdot)$ in feedback form is rank-based if there exists a function $\rho:\Delta^{d-1}_{+,\geq} \to \R^d$ such that $\pi^{i}(x) = \sum_{k=1}^d \rho^k(x^{()})1_{\{r_k(x) = i\}}$ for every $x \in \Delta^{d-1}_+$ and $i=1,\dots,d$.
	\end{enumerate}
\end{defn}

Note that the portfolio \eqref{vol_stabilized_portfolio} from Section~\ref{a_motivating_example} is rank-based so studying the long-only problem in the rank-based case is well motivated. We follow the setup and notational conventions of \cite{kardaras2018ergodic} to define the robust asymptotic growth problem in a rank-based setting. As inputs we take a pair $(\kappa,q)$ where $\kappa: \Delta^{d-1}_{+,\geq} \to \mathbb{S}^d_+$ and $q:\Delta^{d-1}_{+,\geq} \to (0,\infty)$ with $\int_{\Delta^{d-1}_+} q = 1$. In an analogous way to Definition~\ref{Pi_def} we define the following set of admissible measures.
\begin{defn} \label{Pi_geq}
	 Let $\Pi_{\geq}$ consist of all probability measures $\P$ on $(\Omega,\F)$ for which the following hold:
	\begin{enumerate}[label = ({\roman*})]
		\item $X$ is a $\P$-semimartingale and satisfies  $$\langle X^i,X^j\rangle = \int_0^\cdot \kappa_{kl}(X_t^{()})dt;\quad \P\text{-a.s}, \qquad  r_k(X_t) = i, r_l(X_t) = j, \quad i,j=1,\dots,d, $$
		\item For all Borel measurable functions $h$ on $\Delta_{+,\geq}^{d-1}$ with $\int_{\Delta_{+,\geq}^{d-1}}h^+q < \infty$  we have that 
		$$ \lim\limits_{T \to \infty} \frac{1}{T}\int_0^T h(X_t^{()})dt = \int_{\Delta_{+,\geq}^{d-1}}hq; \quad \P\text{-a.s.}$$
		\item The laws of $\{X_t\}_{t \geq 0}$ under $\P$ are tight.
	\end{enumerate}
\end{defn}

The problem we tackle in this section is to solve for
\begin{equation} \label{rank_robust_problem}
\lambda_{\mathcal{E}}^\geq:= \sup\limits_{V \in \mathcal{V}^\mathcal{E}} \inf\limits_{ \P \in \Pi_{\geq}} g(V;\P).
\end{equation}

First, we extend $\kappa$ and $q$ to the (unordered) simplex by setting
\begin{align} \label{c_sym}
c_{ij}(x) & := \kappa_{kl}(x^{()}),  \qquad  r_k(x) = i, r_l(x) = j, \quad i,j=1,\dots,d; & \quad x \in \Delta^{d-1}_+, \\
p(x) & := \frac{1}{d!}q(x^{()}); & x \in \Delta^{d-1}_+. \label{p_sym}
\end{align}
We now assume
\begin{assum} \label{sym_assumption}
	The pair $(\kappa,q)$ is such that the pair $(c,p)$ defined in \eqref{c_sym} and \eqref{p_sym} satisfy Assumptions~\ref{input_assumption} and \ref{finite_assumption}.
\end{assum}
It is not immediately clear when Assumption~\ref{sym_assumption} is satisfied given an input pair $(\kappa,q)$. In Section~\ref{model} we develop a large class of models for which Assumption~\ref{sym_assumption} holds. Moreover, the content of \cite[Proposition 3.8]{kardaras2018ergodic} gives a procedure for modifying the inputs $\kappa$ and $q$ in an arbitrarily small neighbourhood of the boundary so that Assumption~\ref{sym_assumption} is satisfied for the modified inputs. 

To state the main result of this section we define, analogously to Definition~\ref{tilde_pi_def}, $\tilde \Pi_{\geq}$ to be the set of all measures $\P \in \Pi_{\geq}$ that satisfy \eqref{drift_assumption}. Note that Theorem~\ref{concave_invariance_thm} continues to hold for all $\P \in \tilde \Pi_{\geq}$ as long as the concave generating function in the statement of theorem is assumed to be permutation invariant; this is discussed in Remark~\ref{rank_appendix_remark} following the proof of Theorem~\ref{concave_invariance_thm}. 
\begin{thm} \label{rank_optimal}
		Let Assumption~\ref{sym_assumption} hold for the pair $(\kappa,q)$. For the associated pair $(c,p)$ constructed by equations \eqref{c_sym} and \eqref{p_sym} let $\lambda_{\mathcal{E}}$ be as in \eqref{exp_concave_problem} and $\hat \phi,\hat G$ be as in the statement of Theorem~\ref{main_thm}. Then we have
		\begin{enumerate} [label = ({\roman*})]
				\item \label{rank_optimal_ii} The function $\hat G$ is permutation invariant. In particular, the portfolio $\pi_{\hat G}$ is rank-based.
			\item $\lambda_{\mathcal{E}}^{\geq}=\lambda_{\mathcal{E}}$. The robust growth-optimal portfolio for the problem \eqref{rank_robust_problem} is functionally generated by the concave function $\hat G = \exp \hat \phi$. It satisfies $g(V^{\pi_{\hat G}};\P) \geq \lambda_{\mathcal{E}}^\geq$ for every $\P \in \tilde \Pi_{\geq}$ and $g(V^{\pi_{\hat G}};\tilde \P) = \lambda_{\mathcal{E}}^{\geq}$.
		\end{enumerate} 
	\end{thm}
	\begin{proof}  
		The proof of the first item is very similar to the proof of Proposition~3.6 in \cite{kardaras2018ergodic}. Given a permutation $\sigma$ of $\{1,\dots,d\}$ we define $ \hat \phi_\sigma: \Delta^{d-1}_+ \to \R$ via $\hat \phi_\sigma(x) = \hat\phi(x_\sigma)$ where $x_\sigma^i = x^{\sigma(i)}$. It follows that $\hat \phi_\sigma$ is exponentially concave for every $\sigma$ and by convexity of $\mathcal{E}$ we see that $\frac{1}{d!} \sum_{\sigma} \hat \phi_\sigma$ is exponentially concave as well. By the permutation invariance of $c$ and $p$ it follows from \cite[Equation~(B.1)]{kardaras2018ergodic} that
		\begin{equation} \label{perm_invariance}
		\int_{\Delta^{d-1}_+} (\ell-\nabla \phi)^\top c (\ell-\nabla \phi)p = \int_{\Delta^{d-1}_+} (\ell-\nabla \phi_\sigma)^\top c (\ell-\nabla \phi_\sigma)p
		\end{equation} for every $\sigma$. Hence we see that
		$$ \|\ell - \frac{1}{d!}\sum\limits_\sigma \nabla \hat  \phi_\sigma\|^2_{\mathcal{H}^{c,p}} \leq \frac{1}{d!}\sum\limits_{\sigma} \|\ell-\nabla \hat \phi_\sigma\|^2_{\mathcal{H}^{c,p}} = \|\ell -\nabla \hat \phi\|^2_{\mathcal{H}^{c,p}} = \inf\limits_{\nabla \phi \in \partial \mathcal{E}} \|\ell -\nabla \phi\|^2_{\mathcal{H}^{c,p}}.$$
		Thus, $\frac{1}{d!} \sum_{\sigma} \hat \phi_\sigma$ is a minimizer for \eqref{EC}. By uniqueness it follows that (up to an additive constant) $\hat \phi = \frac{1}{d!} \sum_{\sigma} \hat \phi_\sigma$ so that $\hat \phi$, and hence $\hat G$, is permutation invariant. By Theorem~\ref{master_formula} this, in turn, implies that $\pi_{\hat G}$ is a rank-based portfolio.
		
		To prove the second item, we note that we clearly have $\Pi \subseteq \Pi_{\geq}$, from which it follows that $\lambda_{\mathcal{E}}^{\geq} \leq \lambda_{\mathcal{E}}$. To prove the reverse inequality we argue in a similar way to the proof of Theorem \ref{main_thm}. First note that by \eqref{perm_invariance}, Lemma~\ref{C2_lem} continues to hold for all $\P \in \Pi_{\geq}$ as long as the function $G$ in the statement of the lemma is permutation invariant. By Theorem~\ref{partialE_thm} we can find $C^2$ exponentially concave functions $\hat \phi_n$ such that $\nabla \hat \phi_n \to \nabla \hat \phi$ in $\mathcal{H}^{c,p}$ as $n \to \infty$. Moreover, since $\hat \phi$ is permutation invariant it is clear that the functions $\hat \phi_n$ can be chosen to be permutation invariant as well.
	
		 Then by Lemma~\ref{C2_lem} we see by setting $G_n:= \exp \hat \phi_n$ that 
		\begin{align*}
		\lambda_{\mathcal{E}}^{\geq} \geq \inf_{\P \in \Pi_{\geq}} g(V^{\pi_{\hat G_n}};\P) & = \frac{1}{2}\int_{\Delta^{d-1}_+} (\ell^\top c \ell - (\ell -\nabla \hat \phi_n)^\top c (\ell - \nabla \hat \phi_n))p
		\end{align*}
		Sending $n \to \infty$ yields $\lambda^{\geq}_{\mathcal{E}} \geq \lambda_{\mathcal{E}}$. The claims regarding the asymptotic growth rate of the portfolio $\pi_{\hat G}$ readily follow from Theorem~\ref{concave_invariance_thm}, Remark~\ref{rank_appendix_remark} and the definition of $\lambda_{\mathcal{E}}^{\geq}$. 
	\end{proof}

\begin{remark}
	By the previous theorem, \eqref{lambda_EC} and \eqref{perm_invariance} it follows that
	$$\lambda_{\mathcal{E}}^{\geq} = \lambda_{\mathcal{E}} =  \frac{1}{2}\int_{\Delta^{d-1}_+} (\ell^\top c \ell - (\ell - \nabla \hat \phi)^\top c (\ell - \nabla \hat\phi)) p = \frac{1}{2}\int_{\Delta^{d-1}_{+,\geq}}(\varrho^\top \kappa \varrho - (\varrho - \nabla \hat \phi)^\top \kappa (\varrho - \nabla  \hat\phi)) q$$
	where $\varrho = \frac{1}{2}\kappa^{-1}\diver \kappa + \frac{1}{2}\nabla \log q$.
	Moreover by the permutation invariance of $\hat \phi$ and optimality of $\hat \phi$ in \eqref{min_exp} it is easy to see $\hat \phi$ can additionally be characterized as the unique (up to additive constant) minimizer of
	$$\inf_{\phi}\int_{\Delta^{d-1}_{+,\geq}} (\varrho - \nabla \phi)^\top \kappa (\varrho - \nabla \phi) q$$ 
	where the infimum if taken over all exponentially concave functions $\phi$ on $\Delta^{d-1}_{+,\geq}$.
\end{remark}

\subsection{Discussion of Long-Only Constraints} \label{long_only_discussion}
As previously mentioned, there are (at least) four possible long-only problems that could have been considered. Namely, optimizing over
\begin{enumerate}[label = ({\arabic*}),noitemsep]
	\item arbitrary long-only portfolios,
	\item long-only portfolios in feedback form,
	\item functionally generated long-only portfolios,
	\item long-only portfolios generated by concave functions.
\end{enumerate}

The only procedure we are aware of to determine the optimum is to find a suitable class of portfolios such that the following two properties are satisfied:
\begin{itemize}[noitemsep]
	\item (Growth Rate Invariance) Each portfolio in the class has the same asymptotic growth rate under every measure $\P \in \Pi$.
	\item (Approximation) The growth rate of the optimal portfolio can be approximated by the growth rates of portfolios in the chosen class.
\end{itemize}
In \cite{kardaras2018ergodic}, the authors used the class of portfolios generated by $C^2$ functions to tackle the unconstrained problem. In this paper we were able to use the class $\partial \mathcal{E} \cap C^1$(defined in the statement of Theorem \ref{partialE_thm}) to carry out a tractable analysis for problem (4).

A seemingly natural choice would be to consider problem (3). At this point we would like to establish that not all long-only functionally generated portfolios are generated by concave functions, so that problems (3) and (4) are truly different. Indeed, even \textit{convex} functions can generate long-only portfolios as the following example shows.
\begin{eg}
	Set $G(x)= \exp(\frac{1}{2}\|x\|^2)$. An application of the master formula Theorem \ref{master_formula} yields
	$$ \frac{\pi^i_{G}(x)}{x^i} = 1 + (x^i - \|x\|^2); \quad i =1,\dots,d,$$ so that $\pi_{G}$ is long-only.
\end{eg}
For problem (3) one would expect the candidate optimum to solve the variational problem \eqref{variational} over all weakly differentiable functions $\phi$ under the constraint $\pi_{e^{\phi}}^i(x) \geq 0$ for all $x \in \Delta^{d-1}_+$ and $i \in \{1,\dots,d\}$. One can establish existence and uniqueness of the solution $\hat \phi$ to this mathematical problem and the corresponding objective function value would serve as an upper bound for the optimal robust asymptotic growth rate in this setting. However, the method for proving the lower bound in Theorem \ref{main_thm} relies on approximating $\nabla \hat \phi$ in the $\mathcal{H}^{c,p}$ sense by gradients of $C^2$ functions that generate long-only portfolios. In the general context of functionally generated portfolios, it is not clear that this is possible.  Additionally, without establishing regularity properties of the optimizer $\hat \phi$, it is not clear if $\hat \phi(X_t)$ is even a semimartingale, in which case the candidate optimal portfolio would not be functionally generated in the sense of Definition \ref{def_func_gen}.  

In regards to problem (2), portfolios in feedback form do \textit{not} in general have the growth rate invariance property. Indeed, consider the case when $\ell$ is not a gradient so that the measures $\tilde \P$ and $\P_{\hat u}$, defined in Assumption \ref{finite_assumption}(iii) and Theorem \ref{KR_main} respectively, are distinct and both belong to $\Pi$. It is established in Theorem \ref{KR_main} that $\pi_{\hat G}$ is growth-optimal under ${\P_{\hat u}}$ where $ \hat G = \exp \hat u$ and it follows from Lemma \ref{worst_case_growth} that the portfolio  $\pi_\ell$ given by
$$\pi_{\ell}^i(x) := x^i \ell^i(x) + x^i - x^i \ell(x)^\top x; \quad i = 1,\dots,d$$ is the growth-optimal portfolio under $\tilde \P$. Moreover, since the objective function value for the variational problem \eqref{variational} is strictly positive at the optimum $\hat u$, we see that $g(V^{\pi_\ell}; \tilde \P) > g(V^{\pi_{\hat G}}; \tilde \P)$. Thus, if the portfolio $\pi_\ell$ achieved the same growth rate in all admissible models it would outperform $\pi_{\hat G}$ contradicting Theorem \ref{KR_main}. It follows that the class of feedback portfolios is too large to establish growth-rate invariance under the admissible measures. The special structure of functionally generated portfolios is needed. We note that this analysis only holds for $d > 2$ as in the two dimensional case $\ell$ is always a gradient. Moreover, as we will see in Section \ref{2d_long_only}, problems (2) and (3) are equivalent when $d=2$ since every portfolio $\pi$ in feedback form is functionally generated by the function $\varphi^\pi$ given by \eqref{2d_func_gen} below. As a result, the $d=2$ case is explicitly solvable and we were able to get around the aforementioned difficulties. However, without explicit formulae for the solution to the constrained variational problem in higher dimensions, problems (2) and (3) remain open and are beyond the scope of this paper.

In view of the obstructions for problems (2) and (3), problem (1) seems out of reach. As such, we believe problem (4) is the natural one to consider in this context. Moreover, in Section \ref{fin_dim} we will see that the variational problem \eqref{EC} is susceptible to a finite dimensional approximation due to properties of exponentially concave functions and the fact that $\partial \mathcal{E}$ is a bounded set in $\mathcal{H}^{c,p}$. Thus, problem (4) can be numerically solved and implemented in practice.

\section{Long-Only Feedback Portfolios and the \texorpdfstring{$\boldsymbol{d=2}$}{d=2} Case} \label{2d_long_only}
In the $d =2 $ case more can be said about long-only portfolios. In this case we can phrase everything in terms of a one dimensional problem since $X^2 = 1-X^1$ and $\pi^2 = 1-\pi^1$ for every portfolio $\pi$. Moreover if $\pi = \pi(x^1,x^2)$ is a portfolio in feedback form then it is always functionally generated. Indeed, by defining the function
\begin{equation}\label{2d_func_gen}
\varphi^\pi(x) := \frac{\pi^1(x,1-x)}{x(1-x)} -\frac{1}{1-x}
\end{equation} 
we observe that
\begin{align*}
\frac{\pi^1(x, 1-x)}{x} & = 1 + (1-x)\varphi^\pi(x)\\
\frac{\pi^2(x,1-x)}{1-x} & = 1 - x \varphi^\pi(x)
\end{align*}
for every $x \in (0,1)$.
From these equations it follows that $\pi$ is functionally generated by
\[
	G(x,y) = G(x) = \exp\(\int_\theta^x \varphi^\pi(z)dz\)
\] where $\theta \in (0,1)$ is arbitrary. Additionally from the above representation of the portfolio weights we see that $\pi$ is long-only if and only if 
\begin{equation}  \label{2d_long_only_constraint}-\frac{1}{1-x} \leq  \varphi^\pi(x) \leq \frac{1}{x}
\end{equation}
for every $x \in (0,1)$. 

In the two dimensional case the matrix $c$ is of the form
$$c(x,1-x) = \begin{bmatrix}
c_{11}(x,1-x) & - c_{11}(x,1-x) \\
-c_{11}(x,1-x) & c_{11}(x,1-x)
\end{bmatrix}$$ for some nonnegative function $c_{11}$. Moreover it is easy to check that $c^{-1}\diver c = \nabla \log c_{11}$ and so in particular $\ell(x)= \frac{1}{2}\nabla \log(p(x)c_{11}(x))$. 

To state the next theorem we define
$\tilde p(x) = p(x,1-x)$, $ \tilde c(x) = c_{11}(x,1-x)$ and $\tilde \ell (x) = \frac{1}{2}(\log \tilde p \tilde c)'(x)$.
\begin{thm} \label{opt_long_only}
	Let $\Xi$ be the set of all long-only portfolios in feedback form and set 
	$$\lambda_{\text{long}} := \sup_{\pi \in \Xi}\inf_{\P \in \Pi} g(V^{\pi};\P).$$
	Then we have that 
\[
	\lambda_{\text{long}} = \frac{1}{2}\int_0^1 \(\tilde \ell^2 \tilde c - \(\tilde \ell -\varphi^{\hat \pi}\)^2\tilde c\)\tilde p(x)dx
\]
	where 
		\begin{align*}\hat \pi^1 (x) & = \begin{cases}
	1 & \text{if } \tilde \ell(x) > 1/x \\
	0 & \text{if } \tilde \ell(x) < -1/(1-x)\\
	x + x(1-x)\tilde \ell(x) & \text{otherwise,}
	\end{cases}\\
	\hat \pi^2(x) & = 1 - \hat \pi^1(x)
	\end{align*}  and  $\varphi^{\hat \pi}$ is given by \eqref{2d_func_gen}.
\end{thm}
	Here the optimal portfolio $\hat \pi$ itself has a simple interpretation:\ wherever the unconstrained optimal solution was long-only one performs that strategy. As soon as we enter the region where we were to short $x^1$ we do not invest in it and hold all of our wealth in $x^2$. Similarly, when we hit the region where we were to short $x^2$ we do not invest in it and hold all of our wealth in $x^1$. 
	
To prove this theorem we will need the following lemma, which does not assume that $d=2$ and is also used in the proof of Proposition \ref{extreme_points} and Lemma \ref{cutoff_lemma}.

\begin{lem} \label{dom_conv}
	Let Assumptions \ref{input_assumption} and \ref{finite_assumption} hold. Then

$$ \int_{\Delta^{d-1}_+} \frac{1}{x^ix^j}|c_{ij}(x)|p(x)dx < \infty$$ for all $i,j \in \{1,\dots,d\}$.
\end{lem}
\begin{proof}
	Consider the portfolio $ \pi_i$ given by $\pi^i_i(x) = 1$, $\pi^j_i(x) = 0$ for all $j \ne i$. For such a portfolio we have that 
	 $\log V_T^{\pi_i} = \log X_T^i$. By Lemma \ref{worst_case_growth} we have that $\tilde \P$-a.s.
	\begin{align*}\lim_{T\to \infty} T^{-1} \log X_T^i = \lim_{T\to \infty} T^{-1} \log V_T^{\pi_i} & =  \frac{1}{2}\|\ell\|^2_{\mathcal{H}^{c,p}} - \frac{1}{2}\|\ell - h_i\|^2_{\mathcal{H}^{c,p}}
	\end{align*}
	where $h_i(x) = \pi_i(x)/x$.
	
Now if $\|\ell - h_i\|^2_{\mathcal{H}^{c,p}} = \infty$ then we would have that $\lim_{T\to \infty} T^{-1}\log X_T^i = - \infty$; $\tilde \P$-a.s., which would in turn imply that $\lim_{T \to \infty} X_T^i = 0; \tilde \P$-a.s. This contradicts the ergodicity of the process $X_T$ so the norm must be finite. But then by the triangle inequality we see that 
$$\|h_i\|^2_{\mathcal{H}^{c,p}} \leq 2\|\ell - h_i\|^2_{\mathcal{H}^{c,p}} + 2\|\ell\|^2_{\mathcal{H}^{c,p}} < \infty.$$ 
Expanding out we have that 
$$ \|h_i\|^2_{\mathcal{H}^{c,p}}  = \int_{\Delta^{d-1}_+} \frac{1}{(x^i)^2} c_{ii}(x)p(x)dx$$ which proves the claim when $i = j$. To handle the general case let $a(x) = \text{diag}(1/x)c(x) \text{diag}(1/x)$ which is a positive semidefinite matrix for every $x \in \Delta^{d-1}_+$. By properties of positive semidefinite matrices we have for all $i,j$ that 
$$\frac{|c_{ij}(x)|}{x^ix^j} = |a_{ij}(x)| \leq \frac{1}{2}(a_{ii}(x) + a_{jj}(x)) =\frac{1}{2}\( \frac{c_{ii}(x)}{(x^i)^2} + \frac{c_{jj}(x)}{(x^j)^2}\).$$
It follows from this bound that $$\int_{\Delta^{d-1}_+} \frac{1}{x^ix^j}|c_{ij}(x)|p(x)dx \leq  \int_{\Delta^{d-1}_+} \frac{1}{2}\( \frac{c_{ii}(x)}{(x^i)^2} + \frac{c_{jj}(x)}{(x^j)^2}\)p(x)dx<\infty$$ completing the proof. 
\end{proof}

\begin{proof}[Proof of Theorem \ref{opt_long_only}]
	By \eqref{2d_func_gen}, we have that $\frac{\hat \pi^1(x)}{x} = 1 + (1-x)\varphi^{\hat \pi}(x)$
where
\begin{equation} \label{varphi_opt}\varphi^{\hat \pi}(x) = \begin{cases}
1/x & \text{if } \tilde \ell(x) > 1/x \\
-1/(1-x) & \text{if }  \tilde \ell(x) < -1/(1-x)\\
\tilde \ell(x) & \text{otherwise.} 
\end{cases}
\end{equation}
We obtain (in a similar way to the proof of Theorem \ref{main_thm}) the upper bound
\begin{align}  
	\lambda_{\text{long}} \leq \sup_{\pi \in \Xi} g(V^{\pi};\tilde \P) & = \frac{1}{2}\int_{\Delta^{d-1}_+} \ell^\top c \ell p - \frac{1}{2}\inf_{\pi \in \Xi} \int_{\Delta^{d-1}_+} \(\ell(x) - \frac{\pi(x)}{x}\)^\top c(x) \(\ell(x) - \frac{\pi(x)}{x}\)p(x)dx  \nonumber\\
	& = \frac{1}{2} \int_0^1 \tilde \ell^2 \tilde c \tilde p(x)dx - \frac{1}{2}\inf_{\pi \in \Xi} \int_0^1(\tilde \ell -\varphi^\pi(x) )^2\tilde c\tilde p(x)dx \label{inf_xi}
\end{align}
where we performed the change of variables $x_2 = 1-x_1$. By virtue of \eqref{2d_long_only_constraint}, which is a pointwise constraint, the infimum in \eqref{inf_xi} can be computed by pointwise minimization of the integrand. 
Because $\tilde c\tilde p(x) > 0$ the minimization problem is equivalent to
$$\min_{-1/(1-x) \leq y \leq 1/x} ( \tilde \ell(x)-y)^2$$ for every $x \in (0,1)$. Since the objective function is a one dimensional quadratic it follows that for each fixed $x$ the minimum is achieved at $\hat y = \varphi^{\hat \pi}(x)$. This establishes the upper bound
$$ \lambda_{\text{long}} \leq \frac{1}{2}\int_0^1 \(\tilde \ell^2 \tilde c - \(\tilde \ell -\varphi^{\hat \pi}\)^2\tilde c\)\tilde p(x)dx.$$

To obtain the lower bound we will again approximate the generating function of $\hat \pi$ by $C^2$ generating functions. Extend $\varphi^{\hat \pi}$ to $\R$ by setting it to be identically zero outside of $(0,1)$ and let $\psi \in C_c^\infty(\R)$ be such that $\psi \geq 0$, supp($\psi$) = $[0,1]$ and $\int_\R \psi = 1$. Define the mollifiers $\psi_n(x) = n^2\psi(n^2x)$ and set $\varphi_n(x) := \frac{n}{n+1}(\varphi^{\hat \pi}\I{(1/n,1-1/n)} * \psi_n)(x)$. We have that $\varphi_n \in C_c^\infty(\R)$ and claim that $\varphi_n \to \varphi^{\hat \pi}$ pointwise. To show this we fix $x \in (0,1)$ and let $\epsilon > 0$ be arbitrary. By continuity of $\varphi^{\hat \pi}$ we can find an $N$ large enough so that
\begin{enumerate}[label = ({\roman*}),noitemsep]
	\item  if $|x-y| < 2/N^2$ then $|\varphi^{\hat \pi}(x) - \varphi^{\hat \pi}(y)| \leq \epsilon / 2$,
	\item $|\varphi^{\hat \pi}(x)|/(N+1) \leq \epsilon/2$,
	\item $x -1/N^2 \geq 1/N$
 \end{enumerate} Then for all $n \geq N$ we see that 
\begin{align*}
|\varphi_n(x) - \varphi^{\hat \pi}(x)| & \leq \frac{n}{n+1} \left|\int_{\R} (\varphi^{\hat \pi}(y)\I_{(1/n,1-1/n)}(y) - \varphi^{\hat \pi}(x))\psi_n(x-y)dy\right| + \frac{|\varphi^{\hat \pi}(x)|}{n+1}\\
& \leq \int_{x-1/n^2}^{x+1/n^2} |\varphi^{\hat \pi}(y) - \varphi^{\hat \pi}(x)|n^2\psi(n^2(x-y))dy +  \frac{|\varphi^{\hat \pi}(x)|}{n+1}&  \text{(iii)} \\
& \leq \frac{\epsilon}{2} + \frac{\epsilon}{2} = \epsilon & \text{(i) and (ii)}
\end{align*}
which proves the pointwise convergence.
  Next we claim that 
\begin{equation} \label{phin_long_only} -\frac{1}{1-x} \leq \varphi_n(x) \leq \frac{1}{x}
\end{equation} for every $x \in (0,1)$ and $n \in \N$. We fix $n$ and estimate that 
\begin{align}\varphi_n(x) & = \frac{n}{n+1}\int_{(x - 1/n^2) \lor 1/n}^{(x+1/n^2) \land (1-1/n)}  \varphi^{\hat \pi}(y)n^2\psi(n^2(x-y))dy \label{2d_mollify} \\
& \leq \frac{n}{n+1}\int_{(x - 1/n^2) \lor 1/n}^{(x+1/n^2) \land (1-1/n)}  \frac{1}{y}n^2\psi(n^2(x-y))dy \nonumber\\
& \leq \frac{n}{n+1}\frac{1}{(x - 1/n^2) \lor 1/n} \nonumber
\end{align}
If $x  - 1/n^2 \geq 1/n$ then
$$\frac{n}{n+1}\frac{1}{(x - 1/n^2) \lor 1/n} = \frac{n}{(n+1)(x-1/n^2)} = \frac{n}{nx + x - 1/n^2 - 1/n} \leq \frac{1}{x}.$$ Conversely if $x < 1/n + 1/n^2$ then 
 $$\frac{n}{n+1}\frac{1}{(x - 1/n^2) \lor 1/n} = \frac{n^2}{n+1} \leq \frac{1}{x}$$ so that in any case we have $\varphi_n(x) \leq 1/x$. For the lower bound we proceed from \eqref{2d_mollify} to obtain
\begin{align*}\varphi_n(x) & \geq - \frac{n}{n+1}\int_{(x - 1/n^2) \lor 1/n}^{(x+1/n^2) \land (1-1/n)}  \frac{1}{1-y}n^2\psi(n^2(x-y))dy \\
& \geq  -\frac{n}{n+1}\frac{1}{(x+1/n^2) \land (1-1/n)}.
\end{align*}
A similar case analysis shows that $\varphi_n(x) \geq - 1/(1-x)$ for every $x \in (0,1)$. Since $\varphi_n$  satisfies \eqref{phin_long_only} and is $C^2$ we see by Lemma \ref{C2_lem} that
\begin{equation} \label{2d_lower_bound}\lambda_{\text{long}} \geq \inf_{\P \in \Pi} g(V^{\pi_{G_n}};\P) = g(V^{\pi_{G_n}};\tilde \P) =\frac{1}{2}\int_0^1 (\tilde \ell^2 \tilde c  - (\tilde \ell - \varphi_n)^2 \tilde c)\tilde p(x)dx
\end{equation}
where $G_n(x) = \exp(\int_\theta^x \varphi_n(y)dy)$ for arbitrary $\theta \in (0,1).$ By \eqref{phin_long_only} we deduce the bound 
\begin{align*}(\tilde \ell - \varphi_n)^2\tilde c \tilde p (x) & \leq \max\big\{(\tilde \ell(x) - 1/x)^2 \tilde c \tilde p(x), (\tilde \ell(x) + 1/(1-x))^2\tilde c \tilde p(x)\big\}\\
& \leq   4 \tilde \ell^2 \tilde c \tilde p(x) + \frac{2}{x^2}\tilde c \tilde p(x) +  \frac{2}{(1-x)^2} \tilde c \tilde p(x).
\end{align*}
Assumption \ref{finite_assumption}(i) and Lemma \ref{dom_conv} guarantee that the function on the right hand side is integrable so by the Lebesgue dominated convergence theorem we can send $n \to \infty$ in \eqref{2d_lower_bound} to obtain
$$ \lambda_{\text{long}} \geq  \frac{1}{2}\int_0^1 \(\tilde \ell^2 \tilde c - \(\tilde \ell -\varphi^{\hat \pi}\)^2\tilde c\)\tilde p(x)dx.$$
This completes the proof.
\end{proof}
We have been able to explicitly solve the robust optimal growth problem over all portfolios in feedback form in the two dimensional case. Next we investigate the relationship between this problem and the concave problem studied in Section \ref{long_only}. The following example shows that the optimal portfolio $\hat \pi$ from Theorem \ref{opt_long_only} need not be generated by a concave function.
\begin{eg}
	Let $\zeta: (0,1) \to (0,1)$ be given by $\zeta(x) = \exp(-\cos(x))$. Choose a smooth, nonnegative cutoff function $\psi:(0,1) \to [0,1]$ such that $\psi(x) = 1$ for $x \in (1/3,2/3)$ and $\psi(x) = 0$  for $x \in (0,1/4) \cup (3/4,1)$. Then define
	$$\tilde c(x) = \zeta(x)\psi(x) + x(1-x)(1-\psi(x));\quad \quad \tilde p(x) = \frac{1}{Z}\(\psi(x) + x(1-x)(1-\psi(x))\)$$
	where $Z = \int_0^1   (\psi(x) + x(1-x)(1-\psi(x)))dx$.  The inputs $c$ and $p$ from Section~\ref{a_motivating_example} with $a=2$ satisfy 
	$$ \tilde c(x) = c(x,1-x);\quad \quad \tilde p(x) = K p(x,1-x)$$ for $x \in (0,1/4) \cup (3/4,1)$ and for some constant $K > 0$. Thus, we have that $(c,p)$ satisfy Assumption~\ref{finite_assumption} and since this assumption concerns the behaviour of the inputs near the boundary of the simplex it follows that ($\tilde c, \tilde p$) satisfy Assumption~\ref{finite_assumption} as well. 
	
	By Theorem \ref{opt_long_only} the optimal portfolio is generated by $\hat G(x) = \exp(\int_\theta^x \varphi^{\hat \pi}(y)dy)$ where $\theta \in (0,1)$ is arbitrary and $\varphi^{\hat \pi}$ is given by \eqref{varphi_opt}. Note that for $x \in (1/3,2/3)$ we have $\tilde \ell(x) = \frac{1}{2}\sin(x)$ so that in this interval the inequalities $-1/(1-x) \leq \tilde \ell(x) \leq 1/x$ are satisfied. Thus, it follows that $\varphi^{\hat \pi}(x) = \frac{1}{2}\sin(x)$ for $x \in (1/3,2/3)$. For such $x$ we now compute that 
	$$\hat G''(x) = \hat G(x) (\varphi^{\hat \pi}(x)^2 + (\varphi^{\hat \pi})'(x)) = \hat G(x) \(\frac{1}{4} \sin^2(x) + \frac{1}{2}\cos(x)\) > 0.$$
	This shows that $\hat G$ is not concave on $(1/3,2/3)$.
\end{eg} 
The next proposition gives a necessary and sufficient condition for the optimal portfolio in feedback form to be generated by a concave function, so that it is also the solution to the less general robust growth-optimization problem considered is Section \ref{long_only}.
\begin{prop} \label{concave_suff}
Let $A = \{x \in (0,1): -\frac{1}{1-x} < \tilde \ell(x) < \frac{1}{x}\}$. The portfolio $\hat \pi$ from Theorem \ref{opt_long_only} is functionally generated by a concave function if and only if $\tilde \ell(x)^2 + \tilde \ell'(x) \leq 0$ for every $x \in A$.
\end{prop}
\begin{remark} \label{sqrtcp}
	Since $A$ is an open set we can write $A = \cup_{n \in \N} I_n$ for some pairwise disjoint open intervals $\{I_n\}_{n \in \N}$. The condition on $\tilde \ell$ in the statement of the previous proposition is then equivalent to $\sqrt{\tilde c\tilde p}$ being a concave function on each $I_n$.
\end{remark}
\begin{proof}
	First suppose that $\tilde \ell(x)^2 + \tilde \ell'(x) \leq 0$ for every $x \in A$. Note that $\hat \pi$ is functionally generated by 
	$\hat G(x) := \exp(\int_\theta^x \varphi^{\hat \pi}(y)dy)$ for arbitrary $\theta \in (0,1)$ where $\varphi^{\hat \pi}$ is given by \eqref{varphi_opt}. Now fix $$x \in B:=  A \cup \{y: \tilde \ell (y) > 1/y\} \cup \{y: \tilde \ell(y) < -1/(1-y)\}.$$ Since $B$ is an open set, it is clear that $\varphi^{\hat \pi}$ is differentiable at $x$ and we have that 
	\begin{equation} \label{concave_criterion} \varphi^{\hat \pi}(x)^2 + (\varphi^{\hat \pi})'(x) \leq 0
	\end{equation}
	by our assumption on $\tilde \ell$ together with the fact that the functions $1/x$ and $-1/(1-x)$ satisfy the above relationship (with equality).
		Next we claim that 
	\begin{equation} \label{G'_decreasing}
	\limsup_{h \downarrow 0} \frac{\hat G'(x+h) - \hat G'(x)}{h} \leq 0
	\end{equation} 
	for every $x \in (0,1)$. Indeed, an application of the chain rule together with \eqref{concave_criterion} yields $\hat G''(x) \leq 0$ for all $x \in B$, which in turn implies \eqref{G'_decreasing} for $x \in B$. Next fix $x \in \{y: \tilde \ell(y) =1/y\}$. Let $\{h_n\}_{n \in \N}$ be a sequence converging to $0$ which achieves the limsup in \eqref{G'_decreasing} for our choice of $x$. Since $\hat G'(x) = \hat G(x)\varphi^{\hat \pi}(x) > 0$ we have that $\hat G'(x+h_n) > 0$ for $h_n$ small enough. For such $h_n$, $\hat G'(x+h_n) =\hat G(x+h_n)\min\{1/(x+h_n),\tilde \ell(x+h_n)\}$. Assume first that $\min\{1/(x+h_n),\tilde \ell(x+h_n)\} = 1/(x+h_n)$ for infinitely many $n$. Denoting the subsequence where this occurs by $h_{n_k}$ we see that 
	\begin{align*}
	\lim_{k \to \infty} \frac{\hat G'(x+h_{n_k})- \hat G'(x)}{h_{n_k}}&  = \lim_{k \to \infty} \frac{\hat G(x+h_{n_k})\frac{1}{x+h_{n_k}} - \hat G(x)\frac{1}{x}}{h_{n_k}} \\
	& = \(\hat G(y)\frac{1}{y}\)'\bigg|_{y = x} = \hat G(x)\(\frac{\varphi^{\hat \pi}(x)}{x} - \frac{1}{x^2}\) = 0 
	\end{align*}
	where we used the fact that $\varphi^{\hat \pi}(x) = 1/x$. Since the original sequence converged to the limsup we see that \eqref{G'_decreasing} holds in this case. If instead we have that $\min\{1/(x+h_n),\tilde \ell(x+h_n)\} = \tilde \ell(x+h_n)$ for infinitely many $n$ then by considering the subsequence $h_{n_m}$ where this occurs we again see that 
	\begin{align*}	\lim_{m \to \infty} \frac{\hat G'(x+h_{n_m})- \hat G'(x)}{h_{n_m}}&  = \lim_{m \to \infty} \frac{\hat G(x+h_{n_m})\tilde \ell(x+h_{n_m}) - \hat G(x)\tilde \ell(x)}{h_{n_m}} \\
	& = (\hat G\tilde \ell)'(x) = \hat G(x)\(\varphi^{\hat \pi}(x)\tilde \ell(x) + \tilde \ell'(x)\)\\
	& = \hat G(x)(\tilde \ell (x)^2 + \tilde \ell'(x)) \leq 0
	\end{align*} 
where we used the fact that $\varphi^{\hat \pi}(x) = \tilde \ell(x)$ in the last equality and the assumption on $\tilde \ell$ together with continuity of $\tilde \ell^2 + \tilde \ell'$ in the inequality. This establishes that \eqref{G'_decreasing} holds for $x \in \{y: \tilde \ell(y) =1/y\}$. A similar argument yields that \eqref{G'_decreasing} holds for $x \in \{y: \tilde \ell(y) = -1/(1-y)\}$ which proves that \eqref{G'_decreasing} is satisfied for all $x \in (0,1)$. This implies that $\hat G'$ is nonincreasing, which is equivalent to $\hat G$ being concave and completes the proof of the forward direction.

Conversely, assume that there exists an $x \in A$ such that $\tilde \ell(x)^2 +\tilde \ell'(x) > 0$. Then a similar calculation to the one above yields $\hat G''(x) > 0$. Thus, we can find an $\epsilon > 0$ such that $\hat G'$ is increasing on $(x-\epsilon,x+\epsilon)$. It follows that $\hat G$ is not concave.
\end{proof}

 \section{A Tractable Class of Models} \label{model}
 The theory developed by Kardaras and Robertson provides a theoretical way to perform asymptotic growth maximization, but when the number of assets $d$ is large, solving the variational problem \eqref{variational} or equivalently the PDE \eqref{euler_lagrange} can be intractable. Although the main focus of this paper is the study of the long-only robust growth-optimization problem as defined in Section~\ref{long_only}, it is useful to develop large classes of models in the context of SPT where tractable analyses can be performed.  Indeed, without the explicit formula \eqref{vol_stabilized_portfolio} for the growth optimal portfolio in the volatility-stabilized market example it was not clear from the variational problem \eqref{variational} that the optimal strategy may exhibit undesirable features, such as the short-selling we observed. 
 
  In Section~\ref{model_def} we propose the use of a specific, but fairly general, volatility structure which generalizes the volatility structures previously considered in the SPT literature, while still providing closed form solutions to the unconstrained problem. Section~\ref{verify_assumptions} is dedicated to developing conditions on the inputs which ensure Assumptions~\ref{input_assumption} and \ref{finite_assumption} hold for the class of models introduced in Section~\ref{model_def}. Lastly, in Section~\ref{rank_based_model} we discuss when the aforementioned tractable class of models can be understood to come from a rank-based model in the sense of Section~\ref{rank_based_section}.
 
 \subsection{The Model} \label{model_def}
Take functions $g:\Delta^{d-1}_+ \to (0,\infty)$, $f_{i}: (0,1) \to (0,\infty)$,  $f_{ij}:\tilde E \to [0,\infty)$ and set
\begin{align} \label{c_def}
c_{ij}(x) &:= \begin{cases}
- f_{ij}( x^{-ij})f_i(x^i)f_j(x^j)g(x) & i \ne j\\
\sum\limits_{k\ne i} f_{ik}( x^{-ik})f_i(x^i)f_k(x^k) g(x) & i = j
\end{cases} & 1 \leq i,j \leq d.
\end{align}
Here $\tilde E = \{y \in (0,\infty)^{d-2}:\sum_{i=1}^{d-2} y^i < 1\}$ and $ x^{-ij}$ stands for the $d-2$ dimensional vector obtained from $x \in \Delta^{d-1}_+$ by removing the $i^{\text{th}}$ and $j^\text{th}$ coordinate. We impose the following assumption on the functions:

\begin{assum} \label{function_assumption} {} \
	\begin{enumerate}[label = (\roman*)]
		\itemsep-0.2em
		\item  The functions $g,\{f_i\}_{i=1}^d$ and $\{f_{ij}\}_{i \ne j}$ are all $C^{2,\gamma}$ for some $\gamma \in (0,1]$,
		\item  For every $i$, $f_i$ satisfies $\lim_{x \downarrow 0} f_i(x) = 0$ and $\lim_{x \uparrow 1} f_i(x) < \infty$,
		\item  For all $i\ne j$, $f_{ij}$ is bounded and $f_{ij} = f_{ji}$.
	\end{enumerate} 
\end{assum}

\begin{thm} \label{divc_lemma}With these specifications it follows that
\[
	c^{-1}\diver c(x) = \nabla \log \(g(x)\prod\limits_{i=1}^d f_i(x^i)\).
\]
\end{thm}
\begin{proof}
	It suffices to show that $c(x) \nabla \log(g(x)\prod_{i=1}^d f_i(x^i)) = \diver c (x)$. Writing $c_i(x)$ for the $i^{\text{th}}$ row of the matrix $c(x)$, we compute for every $i \in \{1,\dots,d\}$ that
	\begin{align*} c_i(x)& \nabla\log\( g(x)\prod_{i=1}^d f_i(x^i)\) \\
	& = \sum\limits_{\substack{j=1 \\ j \ne i}}^d f_{ij}(x^{-ij})f_i(x^i)f_j(x^j)g(x)\(\frac{\partial_i f_i(x^i)}{f_i(x^i)}  - \frac{\partial_j f_j(x^j)}{f_j(x^j)}+ \frac{\partial_ig(x)}{g(x)}  - \frac{\partial_j g(x)}{g(x)}\) \\
	&= \sum\limits_{\substack{j=1 \\ j \ne i}}^d (\partial_i - \partial_j)(f_{ij}(x^{-ij}) f_i(x^i)f_j(x^j)g(x))\\
	& =  \diver c_i(x).
	\end{align*}
	This gives the result.
\end{proof}

Now applying Theorem \ref{KR_main} we get the following result for this tractable class of models.
\begin{cor} \label{stationary_cor}
	Let $c$ be given by \eqref{c_def}. Under Assumptions \ref{input_assumption} and \ref{finite_assumption} on the inputs $(c,p)$ the solution to \eqref{variational} is given by $\hat \phi(x)  = \frac{1}{2}(\log p(x) + \log g(x) + \sum_{i=1}^d \log f_i(x^i))$ and the corresponding growth rate is
	$$ \lambda = \frac{1}{8}\int_{\Delta_+^{d-1}}\nabla \log \(p(x)g(x)\prod_{i=1}^df_i(x^i)\)^\top c(x) \nabla \log \(p(x)g(x)\prod_{i=1}^df_i(x^i)\)p(x)dx.$$
	Moreover the optimal portfolio is functionally generated with generating function $\hat G = \exp({\hat \phi})$ and corresponding portfolio weights
	\begin{equation}\label{stationary_weights}
	\frac{\pi_{\hat G}^i(x)}{x^i} = \frac{1}{2}\partial_i \log p(x) + \frac{1}{2}\partial_i \log g(x) + \frac{1}{2}\partial_i \log f_i(x^i)  +1 - \frac{1}{2}\sum\limits_{j=1}^d x^j\partial_j\log \(p(x)g(x)\prod_{i=1}^df_i(x^i)\).
	\end{equation}
\end{cor}
\begin{proof}
	This follows immediately from Theorems  \ref{KR_main} and \ref{divc_lemma}.
\end{proof}
Note that the case $f_i(x^i) = x^i$, $f_{ij} = g \equiv 1$ and $p(x) =\frac{\Gamma(ad)}{\Gamma(a)^d}\prod_{i=1}^d (x^i)^{a-1}$ for some $a > 1$ corresponds to the volatility-stabilized market discussed in Section~\ref{a_motivating_example} and \eqref{stationary_weights} reduces to \eqref{vol_stabilized_portfolio}. In general Corollary~\ref{stationary_cor} shows that the specification \eqref{c_def} leads to analytical solutions for the robust growth-optimization problem for arbitrary dimension $d$. We now interpret each of the functions that make up the $c$ matrix to provide intuition regarding how these ingredients affect the market weight dynamics under the worst-case measure $\tilde \P$ (which is equal to $\P_{\hat u}$ under the specification \eqref{c_def} thanks to Theorem \ref{divc_lemma}).

The function $g(x)$ is common to each entry of the matrix $c(x)$ and can be interpreted as a state-dependent time change. Indeed if we let $\tilde X$ be a process with covariation matrix given by \eqref{c_def} without the function $g$ and we set $A_t = \int_0^t (1/g(\tilde X_s))ds$ then $\hat X_t := \tilde X_{A_t^{-1}}$ has covariation matrix $c(\cdot)$ given by \eqref{c_def}. The drift of $\hat X_t$ under this time-change, however, will not match the dynamics of $X_t$ given by \eqref{worst_case_dynamics}. This is because such a time change affects the invariant density; indeed if $p$ was the invariant density for $\tilde X$ then (up to the normalizing constant) $p/g$ will be the invariant density for $\hat X$ as defined above. Since we require $p$ to be the invariant density for the market weight process $X$, the drift must be adjusted accordingly -- in \eqref{worst_case_dynamics} this is handled by the $\diver c$ term appearing in the drift and, as such, the function $g$ also appears in the optimal strategy \eqref{stationary_weights}.

To simplify the discussion regarding the other terms assume now, for simplicity, that $g \equiv 1$. Note that
$$d\langle X^i\rangle_t = f_i(X^i_t)\sum_{j \ne i} f_{ij}(X^{-ij}_t)f_j(X^j_t)dt.$$ From this expression we see that the instantaneous volatility of $X^i$ can be decomposed as a product of a term that depends only on the size of $X^i$ and a term that depends on the configuration of all the other market weights. As such, the function $f_i$ represents the effect on the volatility of $X^i$ driven by the value of $X^i$ itself. To better understand the $f_{ij}$ terms define $F_i: (0,1) \to \R$ via $F_i(x) = \int_\theta^x \frac{1}{f_i(y)}dy$ for every $i =1,\dots,d$ where $\theta \in (0,1)$ is arbitrary. Then we have
\begin{equation} \label{scaled_variable}
d\langle F_i(X^i),F_j(X^j)\rangle_t = -f_{ij}(X^{-{ij}}_t)dt.
\end{equation}
Here the function $F_i$ is a coordinate transformation of $X^i$ and $-f_{ij}$ represents the covariation of the transformed variables. If we take $f_i(x^i) = x^i$ as in the volatility-stabilized example, then \eqref{scaled_variable} reads $d\langle \log X^i,\log X^j \rangle_t = -f_{ij}(X^{-ij})dt$. As such, it is surprising and puzzling from a financial perspective that the optimal strategy from Corollary \ref{stationary_cor} does not depend on the choice of the functions $f_{ij}$. However, the $f_{ij}$'s do affect the growth rate and we expect that they also determine the rate of convergence to the asymptotic growth rate for the wealth process $\log V^{\pi_{\hat G}}$. This indicates that under this framework, the covariations of the coordinate transformed variables $F_i(X^i)$ play a role in determining the achievable growth rate in the market, but under the worst case measure $\tilde \P$ cannot be explicitly exploited to obtain additional growth beyond what is given by the other inputs; namely the invariant density $p(x)$ and the functions $f_i(x^i)$. From a calibration perspective this is an attractive property. Indeed, one does not need to estimate the $f_{ij}$'s, which is a difficult task, to determine the optimal strategy. Nevertheless, this is an unexpected and interesting property of the solution for which we have not been able to obtain a deeper financial explanation.

\subsection{Verifying Assumptions \ref{input_assumption} and \ref{finite_assumption}} \label{verify_assumptions}
Next we establish necessary and sufficient conditions for $c$ to be nondegenerate in the sense of Assumption \ref{input_assumption}(i).
\begin{prop} \label{graph}
	For each $x \in \Delta^{d-1}_+$ let $G_x$ be a graph on $\{1,\dots,d\}$ where we create an edge between vertices $i$ and $j$ if $f_{ij}( x^{-ij}) > 0$. Then the matrix $c(x)$ satisfies the non-degeneracy condition in Assumption \ref{input_assumption}(i) if and only if the graph $G_x$ is connected. 
\end{prop}
\begin{proof}
	Fix an $x \in \Delta^{d-1}_+$.
	First we note that  \begin{equation} \label{c_calc}
	v^\top c(x)v = \sum\limits_{i >j} \zeta_{ij}, \quad \text{where } \zeta_{ij} := f_{ij}( x^{-ij})f_i(x^i)f_j(x^j)g(x)(v^i-v^j)^2 \quad \text{for every }v \in \R^d.
	\end{equation} 
	$(\Leftarrow)$
	We see that 
	$$ v^\top c(x)v = 0 \iff \zeta_{ij} = 0 \text{ for all } i \ne j \iff f_{ij}( x^{-ij})(v^i - v^j)^2 = 0 \text{ for all } i \ne j.$$
	From the last condition if follows that if $f_{ij}( x^{-ij}) > 0$ then $\zeta_{ij}= 0\iff v^i = v^j$.
	Now fix an arbitrary $i,j \in \{1,\dots,d\}$ with $i \ne j$. Since $G_x$ is connected we can find a path of the graph $G_x$
	$$ i \mapsto i_1 \mapsto \dots \mapsto i_m \mapsto j$$ connecting $i$ to $j$. Thus by the definition of the graph it must be that $v^i = v^{i_1} = \dots = v^{i_m}= v^j$ to have $v^\top c(x)v = 0$. Since $i,j$ was arbitrary it follows that 
	$$v^\top c(x)v = 0 \iff v^i = v^j \text{ for all } i,j \in \{1,\dots, d\}$$ so that $v \in \text{span}(\boldsymbol{1})$.
	
	$(\Rightarrow)$ Suppose by way of contradiction that $G_x$ is not connected. Then we can find a nonempty subset $A$ of vertices of $G_x$ such that $A^c$ is nonempty and there are no edges between $A$ and $A^c$. Setting
	$$ v := \sum\limits_{i \in A} e_i ,$$ where $\{e_i\}_{i=1}^d$ are the standard basis vectors in $\R^d$, we see from \eqref{c_calc} that $v^\top c(x)v = 0$ and $v \ne \boldsymbol{1}$ which is the required contradiction.
\end{proof}
The previous proposition shows that the non-degeneracy of $c(x)$ can be reduced to the study of a certain graph. We now consider an example to help interpret this condition with regards to the market weight process $X$. Set $d=4$ and pick constant  functions $f_{ij}$ given by $f_{12} = f_{34} =1$, $f_{13} = f_{14} = f_{23} = f_{24} = 0$. Then we see that $e_1+e_2$ and $e_3+e_4$ are in the kernel of $c(x)$ for every $x$. Here the graph condition fails with the corresponding graph having two distinct connected components: $\{1,2\}$ and $\{3,4\}$. As a result the dynamics of the market weight process $X$ given by \eqref{worst_case_dynamics} satisfy $d(X^1 + X^2) = d(X^3 + X^4) = 0$. Hence the sub-markets consisting of $\{X^1,X^2\}$ and $\{X^3,X^4\}$ each have a fixed size; that is for an initial configuration $x_0$ we have $X^1_t + X^2_t = x^1_0+x^2_0$ and $X^3_t + X^4_t = x^3_0 + x^4_0$ for all $t\geq 0, \tilde \P$-a.s.  Conversely, when the graph condition is satisfied no subset of the market weights can be seen to occupy a fixed proportion of the market. Indeed, consider even the most extreme case where we set $f_{i,i+1} = 1$ for every $i =1,\dots,d-1$ and $f_{ij} = 0$ if $|i-j| > 1$. Here the graph is connected, but only by the single path 
$$1 \leftrightarrow 2 \leftrightarrow \dots \leftrightarrow d.$$
However, we will still have that the fluctuations of $X^1$ indirectly affect the fluctuations of $X^d$ since $X^1$ has nontrivial covariation with $X^2$, which in turn has nontrivial covariation with $X^3$ and so on. This way of viewing the graph condition as a connectivity condition between the market weight covariation processes motivates the following way one can check if the graph condition of Proposition~\ref{graph} is satisfied. Let $A_{ij}(x)= \I\{f_{ij}(x) > 0\}$ for $i \ne j$ and set $A_{ii}(x)= 1$ for every $i$. It is easy to see that the graph condition for $G_x$ holds if and only if $A^{d-1}(x)$ has strictly positive entries. In particular if $A^k_{ij}(x) > 0$ for all $i,j = 1,\dots,d$ and some $k \leq d-1$ then the graph condition holds.

We now present sufficient conditions on the input functions in \eqref{c_def} so that Assumption~\ref{finite_assumption} holds.
\begin{prop} \label{sufficient}
	Let $c$ be given by \eqref{c_def} where the functions $g,f_i,f_{ij}$ satisfy Assumption \ref{function_assumption} and let
	\[
	R(x)=\sqrt{p(x)g(x)\prod_{i=1}^df_i(x^i)}.
	\]
	If $\lim_{x \to \partial \Delta^{d-1}_+}R(x) = 0$, $LR/R$ is bounded from below and we have that
	\[
	\int_{\Delta^{d-1}_+} \(|\frac{LR}{R}| + |L(\log R)|\)p < \infty
	\]
	then conditions (i)--(iii) in Assumption \ref{finite_assumption} hold.
\end{prop}
This proposition will be used in Section~\ref{examples} to verify that Assumption~\ref{finite_assumption} holds in the various examples considered. The proof of the result is located in Appendix~\ref{mart_prob_app}.

\subsection{The Rank-Based Case} \label{rank_based_model}
We finish off this section by connecting the tractable class introduced in equation \eqref{c_def} to the rank-based problem of Section~\ref{rank_based_section}. In that section we started with a pair $(\kappa,q)$ defined on the ordered simplex and extended it to a pair $(c,p)$ via the equations \eqref{c_sym} and \eqref{p_sym}. Theorem~\ref{rank_optimal}\ref{rank_optimal_ii} showed that the optimal long-only portfolio turned out to be rank-based. Here, we take the converse approach. Namely, starting with \eqref{c_def} defined on the entire simplex, we will develop conditions on the functions $g,f_i$ and $f_{ij}$ so that $c$ can be viewed as an extension of some $\kappa$ defined on the ordered simplex. Additionally, using the explicit formula for the optimal unconstrained portfolio we develop conditions on the inputs that ensure it is a rank-based portfolio. This is the content of the next proposition.

\begin{prop} \label{ranked_based_tractable}
	\begin{enumerate}[label = ({\arabic*})]		
		\item There exists a $\kappa \in C^{2,\gamma}(\Delta^{d-1}_{+,\geq};\mathbb{S}^d_+)$ such that \eqref{c_sym} holds for the matrix $c$ given by \eqref{c_def} if and only if
		\begin{enumerate} [label = ({\roman*})]
			\item g is permutation invariant,
			\item There exists a common function $f$ such that $f_i = f$ for every $i =1,\dots,d,$
			\item There exists a common function $h$ such that $f_{ij} = h$ for every $i \ne j$.
		\end{enumerate}  
		\item Let $(c,p)$ satisfy Assumptions~ \ref{input_assumption} and \ref{finite_assumption}. Further suppose that $p$ is permutation invariant and that $c$ is given by \eqref{c_def} and satisfies (i) and (ii) above. Then the optimal unconstrained strategy $\pi_{\hat G}$ from Corollary \ref{stationary_cor} is rank-based.
	\end{enumerate}
\end{prop}
\begin{proof}
	First assume that \eqref{c_sym} holds for an appropriate matrix valued function $\kappa$. It follows that for every $x \in \Delta^{d-1}_+$ we have $c_{ij}(x) = c_{\sigma^{-1}(i)\sigma^{-1}(j)}(x_\sigma)$ for every $i,j$ and any permutation $x_\sigma$ of $x$. This reads that we must have
	\begin{align*} 
	-f_{ij}(x^{-ij})f_i(x^i)f_j(x^j)g(x) & = - f_{\sigma^{-1}(i)\sigma^{-1}(j)}(x^{-ij})f_{\sigma^{-1}(i)}(x^i)f_{\sigma^{-1}(j)}(x^j)g(x_\sigma); & i \ne j.
	\end{align*}By comparing like terms we see that $(i),(ii)$ and $(iii)$ must hold. For the converse direction assume that $(i),(ii)$ and $(iii)$ hold. Define $\kappa_{kl}(y) = -h(y^{-kl})f(y_k)f(y_l)g(y)$ for every $y \in \Delta^{d-1}_{+,\geq}$ and $k \ne l$. Set $\kappa_{kk} = -\sum_{l \ne k} \kappa_{kl}$ for each $k=1,\dots,d$. It is then easy to verify that for this choice of $\kappa$ the $c$ matrix defined by \eqref{c_sym} and \eqref{c_def} are the same.
	
	To prove (2) we note that if a generating function is permutation invariant then the portfolio it generates is rank-based. By Corollary~\ref{stationary_cor} and condition $(ii)$ from the statement of the theorem we see that 
	$$\hat G(x) = \sqrt{p(x)g(x)\prod_{i=1}^d f(x^i)}.$$
	The permutation invariance of $p$ and $g$ ensures $\hat G$ is permutation invariant, which completes the proof.
\end{proof}

	Note that if we had additionally assumed that $(iii)$ from (1) held then part (2) of the theorem would follow from Proposition~3.6 in \cite{kardaras2018ergodic}. However, the remarkable structure of this class of models makes it so that the optimal unconstrained portfolio is rank-based if only $(i)$ and $(ii)$ hold, even if $(iii)$ does not. In this case $c$ has non rank-based dependencies, but the  optimal trading strategy is still rank-based. The same cannot as easily be said about the optimal long-only strategy. If $(i),(ii)$ and $(iii)$ hold from part (1) of the above theorem then Theorem~\ref{rank_optimal}\ref{rank_optimal_ii} guarantees the optimal long-only portfolio generated by a concave function is rank-based. However, without an analytically available solution to this problem it is unclear if the optimal portfolio remains rank-based if condition (iii) fails as is the case for the unconstrained optimal portfolio.
\section{Examples} \label{examples}
We now consider a few examples.
\subsection{The Dirichlet Case} \label{dirichlet_example}
Take $g(x) = 1$, $f_i(x^i) = (x^i)^{b^i}$ for some parameters $b^i \geq 1$ and let
\[
p(x) = \frac{1}{B(a)}\prod_{i=1}^d (x^i)^{a^i-1}
\]
be the Dirichlet density with parameters $a = (a^1,\dots,a^d)$ with $a^i > 0$ for every $i \in \{1,\dots,d\}$. Here
\[
B(a) = \frac{ \prod_{i=1}^d \Gamma(a^i) }{ \Gamma(\sum_{i=1}^d a^i) }
\]
is the generalized Beta function. The choice $b^i=1,f_{ij} \equiv 1$ and $a^i=a$ for every $i,j$ and some $a > 1$ reduces to the volatility-stabilized market example discussed in Section~\ref{a_motivating_example}.

Define $\gamma^i := a^i + b^i -1$. Under the condition that $\gamma^i > 1$ for every $i$ we have that Assumption~\ref{finite_assumption} holds. Indeed using the notation of Proposition~\ref{sufficient} we have that $R(x) = \prod_{i=1}^d (x^i)^{\gamma^i/2} $ so that $\lim_{x \to \partial \Delta^{d-1}_+} R(x) = 0$. A direct computation in this case yields
$$ LR(x) = \frac{1}{2}R(x) \sum\limits_{i > j} f_{ij}( x^{-ij})(x^i)^{b^i}(x^j)^{b^j}\( \frac{\gamma^i(\gamma^i-1)}{(x^i)^2} + \frac{\gamma^j(\gamma^j-1)}{(x^j)^2} - 2\frac{\gamma^i\gamma^j}{x^ix^j}\).$$
Since $\gamma^i,\gamma^j > 1$, the coefficients in the first two terms are strictly positive and the condition that $b^i,b^j \geq 1$ ensures that the third term does not blow up to $-\infty$. It follows that $LR/R$ is bounded from below on $\Delta^{d-1}_+$. Next we see that 
\begin{align*} \frac{LR(x)}{R(x)}p(x)
= & \frac{1}{2B({a})}\sum\limits_{i > j}\bigg( f_{ij}(x^{-ij})\prod\limits_{k \not \in \{i,j\}}(x^k)^{a^k-1}\\
&\times  \(\gamma^i(\gamma^i-1)(x^i)^{\gamma^i -2 }(x^j)^{\gamma^j} + \gamma^j(\gamma^j-1)(x^i) ^{\gamma^i}(x^j)^{\gamma^j -2} - 2\gamma^i\gamma^j(x^i)^{\gamma^i -1}(x^j)^{\gamma^j -1}\)\bigg)
\end{align*}	 
which is integrable since $\gamma^i - 2 > -1$ for every $i$. Lastly we compute that 
\begin{align*}L(\log R)(x) &= -\frac{1}{4}\sum\limits_{\substack{i,j=1 \\ i \ne j}}^d \gamma^i f_{ij}(x^{-ij}) (x^i)^{b^i-2}(x^j)^{b^i}	
\shortintertext{so we see that}
|p(x) L(\log R)(x)| &= \frac{1}{4B({a})} \sum\limits_{\substack{i,j=1 \\ i \ne j}}^d \gamma^i f_{ij}(x^{-ij})\gamma^i (x^i)^{\gamma^i -2}(x^j)^{\gamma^j}
\end{align*} which is again integrable. Hence, by Proposition~\ref{sufficient} we conclude that  Assumption~\ref{finite_assumption} holds in this case. 

Further assume that $f_{ij}(x^{-ij}) = \alpha_{ij}$ for some constants $\alpha_{ij}$ satisfying the graph condition of Proposition \ref{graph}. Then, thanks to Corollary~\ref{stationary_cor} the robust growth-optimal portfolio is characterized by $\hat \phi = \sum_{i=1}^d \frac{\gamma^i}{2}\log(x^i)$. By the master formula Theorem~\ref{master_formula}, the portfolio weights are given by 
\[ \pi_{\hat G}^i(x) = \frac{1}{2}\(\gamma^i + x^i\(2-\sum\limits_{j=1}^d \gamma^j\)\).
\]
Furthermore we have  	$$ \lambda = \frac{1}{8B(a)}\sum\limits_{\substack{i,j=1 \\ i \ne j}}^d \alpha_{ij}\bigg((\gamma^i)^2B(a +(b^i -2)e_i + b^je_j) - \gamma^i\gamma^jB(a+(b^i - 1)e_i + (b^j-1)e_j)\bigg)$$ where $\{e_i\}_{i=1}^d$ are the standard basis vectors in $\R^d$.
Indeed, a calculation gives
$$\ell^\top c\ell(x) = \frac{1}{4}\sum\limits_{\substack{i,j=1 \\ i \ne j}}^d\alpha_{ij}(x^i)^{b^i}(x^j)^{b^j} \(\frac{(\gamma^i)^2}{(x^i)^2} -  \frac{\gamma^i\gamma^j}{x^ix^j}\)$$ so that 
multiplying by $\frac{1}{2}p$, integrating and using the fact that $\int_{\Delta^{d-1}_+} \prod_{i=1}^d (x^i)^{r^i-1}dx = B(r)$ for any constants $r^i > 0$ yields the desired result.

As in the example of Section~\ref{a_motivating_example} we see that $\pi_{\hat G}$ is given by an affine combination of a constant-weighted portfolio and the market portfolio. Moreover, it will never be long-only for $d > 2$ due to the constraint $\gamma^i > 1$ for every $i$. Using the results of Section~\ref{2d_long_only} we can obtain an explicit expression for the optimal long-only portfolio in feedback form when $d=2$. In the notation of that section we have $$ \tilde \ell(x) = \frac{1}{2}\(\frac{\gamma^1}{x} - \frac{\gamma^2}{1-x}\).$$
Setting $\theta^1= \frac{\gamma^1-2}{\gamma^1 + \gamma^2 -2}$ and $\theta^2 = \frac{\gamma^1}{\gamma^1 + \gamma^2 -2}$ it follows from Theorem~\ref{opt_long_only} that the optimal long-only portfolio in feedback form is given by
$$ \hat \pi^1(x) = \begin{cases}
1 &  x^1 < \theta^1\\
\frac{1}{2}\gamma^1 + x^1(1 - \frac{1}{2}\gamma^1 - \frac{1}{2}\gamma^2) &  \theta^1 \leq  x^1 \leq \theta^2 \\
0 & x^1 > \theta^2.
\end{cases}
$$ 
and $\hat \pi^2(x) = 1-\hat \pi^1(x)$.
In particular we see that the unconstrained portfolio is long-only if and only if $1 < \gamma^i \leq 2$ for $i = 1,2$. Moreover, we can explicitly compute the growth rate achieved by this portfolio in terms of the incomplete Beta function. Additionally, it can be verified that 
$\sqrt{\tilde c \tilde p(x)} = (x)^{\gamma^1/2}(1-x)^{\gamma^2/2}$ is concave on $(0 \lor \theta_1,\theta_2 \land 1)$ so by Remark \ref{sqrtcp} together with Proposition~\ref{concave_suff} it follows that $\hat \pi$ is generated by a concave function.

\begin{remark}
	In the case when $d = 2$ the condition $\gamma^i \geq 1$ for $i = 1,2$ is necessary and sufficient for Assumption \ref{finite_assumption} to hold as can be checked by Feller's test for explosion (see e.g. Theorem 5.29 in \cite{karatzas1998brownian} for the statement of Feller's test).
\end{remark}
\begin{remark}
	It is worth noting that before studying the ergodic robust asymptotic growth problem in \cite{kardaras2018ergodic}, the authors in \cite{kardaras2012robust} had studied under a similar framework the asymptotic growth problem where the only input was the covariation matrix $c$. They found that the optimal growth rate is given by the negative of the principal eigenvalue $\lambda^*$ of the operator $L$ and the optimal generating function is $\log \eta^*$ where $\eta^*$ is the eigenfunction corresponding to $-\lambda^*$; that is $\eta^*$ and $\lambda^*$ satisfy $L\eta^* = -\lambda^*\eta^*$. Note that one has the following minimax representation for the principal eigenvalue of an operator (see Theorem 4.4.7 in \cite{pinsky1995positive} )
	
	\begin{equation} \label{minimax}-\lambda^* = \inf\limits_{\mu  \in \mathcal{P}(\Delta^{d-1}_+)}\sup\limits_{\substack {u \in C^2(\Delta^{d-1}_+)\\ u >0}} \int_{\Delta^{d-1}_+} \frac{-Lu}{u}(x)\mu(dx)
	\end{equation}
	where $\mathcal{P}(\Delta^{d-1}_+)$ is the space of probability measures on $\Delta^{d-1}_+$.
	From Lemma \ref{C2_lem} we have for $C^2$ functions $G$ and any measure $\P\in \Pi$ that the growth rate is given by
	$$ g(V;\P) = \int_{\Delta^{d-1}_+} \frac{-LG}{G}p.$$ By comparing this representation of the growth rate with \eqref{minimax}, we might expect that in many cases, by minimizing the robust optimal-growth rate over densities $p$ for which Assumption~\ref{finite_assumption} holds, we would recover the eigenvalue problem of \cite{kardaras2012robust}. Indeed, we can verify this in the following example. Taking $b = \boldsymbol{1}$ in the Dirichlet example above it can easily be verified that $$-\lambda^* = \frac{1}{2} \sum\limits_{\substack{i,j=1 \\ i \ne j}}^d \alpha_{ij}$$ with eigenfunction $\eta^*(x) = \prod_{i=1}^d x^i$. Thus, by taking $a= (2,2,\dots,2)$ we do indeed recover the worst-case model from \cite{kardaras2012robust}. 
\end{remark}

\subsection{Generalized Volatility-Stabilized Models} \label{gen_vol_stab}
In \cite{pickova2014generalized} the authors consider a generalized volatility-stabilized model, which specifies that the stock capitalizations satisfy the SDE

\begin{equation} \label{vol_stab_stock_dynam}
d\log S^i_t = \frac{\alpha^i}{2(X^i_t)^{2\beta}}K(S_t)^2dt + \frac{\sigma}{(X^i_t)^\beta}K(S_t)dW^i_t; \quad i=1,\dots,d,
\end{equation}
where $X^i = S^i/(S^1+\dots+S^d)$. Here  $K$ is a measurable real-valued function and $\alpha^i \geq 0$, $\sigma > 0$ and $\beta > 0$ are fixed parameters. In this section, we assume the following conditions on $K$:
\begin{enumerate}[label = ({\roman*})]
	\itemsep=-0.3ex
	\item $K \in C^{2,\gamma}((0,\infty)^d,(0,\infty))$ for some $\gamma \in (0,1]$.
	\item There exist constants $K_{\min},K_{\max} > 0$ such that $K_{\min} \leq K(s) \leq K_{\max}$ for every $s \in (0,\infty)^d$,
	\item $K$ depends on $s$ only through $x$; that is, there exists a function $\tilde K$ which satisfies $\tilde K(x) = K(s)$ for every $s \in (0,\infty)^d$ where $x = s/(s^1+\dots+s^d)$. 
\end{enumerate}  It is shown in \cite{pickova2014generalized} that there exists a weak solution to \eqref{vol_stab_stock_dynam} which does not explode in finite time. 

A computation yields
\begin{equation} \label{vol_stab_market_dynam}
\begin{split}
\frac{dX_t^i}{X_t^i} = \(\sigma^2\(- (X_t^i)^{1-2\beta}  +\sum_{j=1}^d(X_t^j)^{2(1-\beta)} \) + \frac{a^i}{2}(X^i_t)^{-2\beta} - \frac{1}{2}\sum_{j=1}^da^j(X^j_t)^{1-2\beta} \)\tilde K(X_t)^2dt \\
+  \sigma \tilde K(X_t)\((X_t^i)^{-\beta}(1-X^i_t)dW^i_t + \sum_{j \ne i} (X^j_t)^{1-\beta}dW_j\)
\end{split}
\end{equation}
where $a^i := \alpha^i - \sigma^2$ for every $i$. Note that if we choose $\beta = 1/2, \sigma^2=1, \tilde K(\cdot) \equiv 1$ and $a^i = a$ for some $a > 0$ and every $i=1,\dots,d$, then 
\eqref{vol_stab_market_dynam} reduces to \eqref{vol_stabilized} from Section~\ref{a_motivating_example}. From \eqref{vol_stab_market_dynam} we see that the instantaneous covariation matrix is given by
$$ c_{ij}(x) = -\sigma^2 \tilde K^2(x)x^ix^j\((x^i)^{1-2\beta} + (x^j)^{1-2\beta}-\sum_{k}(x^k)^{2(1-\beta)}\); \quad i\ne j,\ x \in \Delta^{d-1}_+,$$
and $c_{ii}(x) = -\sum_{j \ne i} c_{ij}(x).$ For $\beta \ne 1/2$ this specification for $c$ does not correspond to the tractable class from Section~\ref{model} defined by \eqref{c_def}. Nevertheless, we are still able to compute $c^{-1}\diver c$ and obtain an explicit formula for the invariant density. To the best of our knowledge the formulas below for the invariant density and growth-optimal strategies in the case $\beta \ne 1/2$ are new. To simplify the presentation  set $\gamma^i = \frac{a^i}{\sigma^2}$ and for a vector $z \in \R^d$ and $q > 0$ we set $|z|_q = (\sum_{i=1}^d |z^i|^q)^{1/q}$. It can be shown that
$$c^{-1} \diver c(x) = \nabla \log\(\tilde K(x)^2|x|_{2\beta}^{2(1+(d-1)\beta)-d}\prod_{i=1}^d(x^i)^{2(1-\beta)}\).$$

We define
$$p(x) \propto |x|_{2\beta}^{b}\prod_{i=1}^d (x^i)^{\gamma^i + 2(\beta-1)}\tilde K(x)^{-2}$$
where $b = d-\sum_{i=1}^d \gamma^i - 2(d-1)\beta$ and note that since $|x|_{2\beta}^b$ and $\tilde K$ are both bounded, $p$ is integrable if and only if $\gamma^i > 1-2\beta$ for every $i=1,\dots,d$.
Under these specifications it is clear that Assumption~\ref{input_assumption} holds and a computation shows that \eqref{vol_stab_market_dynam} becomes 
$$dX_t = c\ell(X_t)dt + \sigma(X_t)dW_t$$ where $\sigma(x)$ is a matrix square root of $c(x)$, and 
$$\ell(x) = \frac{1}{2}c^{-1}\diver c + \frac{1}{2}\nabla \log p = \nabla \log \( |x|_{2\beta}^{1-\frac{1}{2}\sum_{i=1}^d \gamma^i} \prod_{i=1}^d(x^i)^{\frac{\gamma^i}{2}}\).$$  A similar calculation to the one in Example~\ref{dirichlet_example} shows that Assumptions~\ref{finite_assumption}(i) and (ii) hold if and only if $\gamma^i > \max\{2(1-\beta),1\}$ for every $i$. As previously mentioned it was shown in \cite{pickova2014generalized} that $S$ does not explode, from which it follows that $X$ does not explode so that Assumption~\ref{finite_assumption}(iii) holds here as well.

Thus, under the assumption $\gamma^i > \max \{2(1-\beta),1\}$ for every $i$ it follows from Theorem~\ref{KR_main} that the robust growth-optimal portfolio is generated by the function 
$$G(x) = |x|_{2\beta}^{1-\frac{1}{2}|\gamma|_1} \prod_{i=1}^d (x^i)^{\frac{\gamma^i}{2}}$$ and the corresponding portfolio is given by
$$\pi_{\hat G}^i(x) = \frac{1}{2}\(\gamma^i + \(2-|\gamma|_1\)\frac{(x^i)^{2\beta}}{|x|_{2\beta}^{2\beta}}\), \quad i=1,\dots,d.$$ We see that $\pi_{\hat G}$ is rank-based if and only if $\gamma^i = \gamma$ for some $\gamma > \max\{2(1-\beta),1\}$ and for every $i=1,\dots,d$. Additionally, $\pi_{\hat G}$ is a long-only portfolio if and only if $\gamma^i \geq |\gamma|_1-2$ for every $i$. When $d > 2$, it is easily checked that this condition is incompatible with the requirement $\gamma^i > 1$. Hence, as in the standard volatility-stabilized market example of Section~\ref{a_motivating_example}, the unconstrained robust growth-optimal strategy is never long-only for $d > 2$.

Again, using the results of Section~\ref{2d_long_only} we can obtain an explicit expression for the optimal long-only portfolio in feedback form when $d=2$. Set  $\theta^1 = (\frac{\gamma^2-2}{\gamma^1})^{1/(2\beta)}\I\{\gamma^2 >2\}$ and $\theta^2 = (\frac{\gamma^2}{\gamma^1-2})^{1/(2\beta)}\I\{\gamma^1 >2\} + \infty \I\{\gamma^1 \leq 2\}.$ From Theorem~\ref{opt_long_only} we obtain that

$$ \hat \pi^1(x) = \begin{cases}
	1 &  x < \frac{1}{1+\theta^2}\\
	\frac{1}{2}\(\gamma^1  + (2-\gamma^1-\gamma^2)\frac{x^{2\beta}}{x^{2\beta} + (1-x)^{2\beta}}\) &  \frac{1}{1+\theta^2} \leq  x \leq \frac{1}{1+\theta^1} \\
	0 & x > \frac{1}{1+\theta^1}.
\end{cases}
$$  
and $\hat \pi^2(x) = 1-\hat \pi^1(x)$. It follows that the unconstrained optimal strategy for $d=2$ is long-only if and only if $\max\{2(1-\beta),1\} < \gamma_i \leq 2$ for $i=1,2$. Moreover, it can be verified that $$\sqrt{\tilde p \tilde c(x)} = (x^{2\beta} + (1-x)^{2\beta})^\frac{{1 - \frac{1}{2}\gamma^1 - \frac{1}{2}\gamma^2}}{2\beta}x^{\frac{\gamma_1}{2}}(1-x)^{\frac{\gamma^2}{2}}$$ is always concave on the region $(\frac{1}{1+\theta^2},\frac{1}{1+\theta^1})$ so that by Proposition~\ref{concave_suff} together with Remark~\ref{sqrtcp} we see that $\hat \pi$ is generated by a concave function.

\subsection{The Logit-Normal Distribution} \label{logit_normal}
Let $g(x) = 1$ and set $f_i(x^i) = (x^i)^{a^i}(1-x^i)^{b^i}$ for some constants $a^i, b^i > 2$. Let $p$ be the density of the logit-normal distribution with parameters $\mu \in \R^d$ and $\Sigma \in \mathbb{S}^d_{++}$; that is we have
$$p(x)  = \frac{1}{\sqrt{2\pi|\det \Sigma|}} \frac{1}{\prod_{i=1}^dx^i(1-x^i)}\exp\(-\frac{1}{2} \(\text{logit}(x) - \mu\)^\top \Sigma^{-1}\(\text{logit}(x)-\mu\)\)$$
where $$\text{logit}(x) = \(\log\(\frac{x^1}{1-x^1}\),\dots,\log\(\frac{x^d}{1-x^d}\)\).$$
The logit-normal distribution has the property that if $Y$ is logit-normal with parameters $\mu,\Sigma$ then $\text{logit}(Y)$ is $d$-variate normal with mean vector $\mu$ and covariance matrix $\Sigma$. 

We will now show that under these specifications Assumption~\ref{finite_assumption} is satisfied by virtue of Proposition~\ref{sufficient}. Using the notation from Proposition~\ref{sufficient} we have in this case that
$$R(x)  \propto \prod_{i=1}^d (x^i)^{\frac{1}{2}(a^i-1)}(1-x^i)^{\frac{1}{2}(b^i-1)}\exp\(-\frac{1}{4} \(\text{logit}(x) - \mu\)^\top \Sigma^{-1}\(\text{logit}(x)-\mu\)\)$$
so we see that $\lim_{x \to \partial \Delta^{d-1}_+}R(x) = 0$. A direct calculation shows that there is a constant $C>0$ such that
\begin{align*}
\left|\frac{LR}{R}(x)\right| &\leq C\sum\limits_{i=1}^d (x^i)^{a^i-2}(1-x^i)^{b^i-2} \\
|\log R(x)| &\leq C \sum\limits_{i=1}^d (x^i)^{a^i-2}(1-x^i)^{b^i-2}
\end{align*}
and that $LR/R$ is bounded from below since $a^i,b^i > 2$. It follows that both $|LR/R|p$ and $|\log R|p$ are integrable so by Proposition \ref{sufficient} we have that Assumption \ref{finite_assumption} is satisfied.

Hence, by Corollary \ref{stationary_cor} we conclude that 
\begin{align*}\hat \phi(x) =   -&\frac{1}{4}\(\text{logit}(x)-\mu\)^\top \Sigma^{-1}\(\text{logit}(x) - \mu\)\\
+&  \frac{1}{2}\sum\limits_{i=1}^d(a^i-1)\log x^i + (b^i-1) \log (1-x^i) 
\end{align*}
and by the master formula Theorem \ref{master_formula} that
\begin{align*}\pi_{\hat G}^i(x)=  &\frac{1}{2}\(a^i-1 - (b^i-1) \frac{x^i}{1-x^i} - \frac{x^i}{(1-x^i)^2}\Sigma^{-1}\(\text{logit}(x)-\mu\)^i\) \\& + x^i- \frac{x^i}{2}\sum\limits_{j=1}^d \(a^j-1 - (b^j-1) \frac{x^j}{1-x^j} -\frac{x^j}{(1-x^j)^2} \Sigma^{-1}\(\text{logit}(x)-\mu\)^j\).
\end{align*}

\section{Applications} \label{applications}

\subsection{Finite Dimensional Approximation} \label{fin_dim}
Next we turn our attention to establishing a method for finding the solution to \eqref{EC}.
The class of exponentially concave functions has nice properties that make the optimization problem susceptible to a finite dimensional approximation which can be numerically implemented. Indeed, the next two results establish a large class of extreme points for the set $\partial \mathcal{E}$ defined in \eqref{partialE}.
\begin{defn}
	Let $\mathcal{C}$ be a convex set. We say that $f \in \mathcal{C}$ is an extreme point if whenever we have $f = \alpha g + (1-\alpha)h$ for some $g,h \in \mathcal{C}$ and $\alpha \in (0,1)$ then we must have that $ g= h = f$.
\end{defn}
For the next lemma set
$$\R^d_{++} = \{a \in \R^d: a^i > 0, \ i = 1,\dots,d\}.$$
\begin{lem} \label{log_lemma}
	Let $a,v,w \in \R^d_{++}$, and $\alpha \in (0,1)$ be given such that 
	\[
	\log(a^\top x) = \alpha \log(v^\top x) + (1-\alpha)\log(w^\top x)
	\] for every $x \in D$, where $D$ is any nonempty relatively open set in $\Delta^{d-1}_+$. Then there exist constants $c_1,c_2 > 0$ such that $a = c_1v = c_2w$.
\end{lem}
\begin{proof} 
	To simplify the presentation of the proof we explicitly make the transformation $x^d \mapsto 1- \sum_{i=1}^{d-1}x^i$ and work on $\R^{d-1}$. Write $x'$ for $(x_1,\dots,x_{d-1})$. Define $\tilde a \in \R^{d-1}$ via $\tilde a^i := a^i - a^d$ for $i=1,\dots,d-1$. The vectors $\tilde v$ and $\tilde w$ are defined analogously. Then by assumption we have that 
	$$\log (a^d + \tilde a^\top x' ) = \alpha \log(v^d + \tilde v^\top x') + (1-\alpha)\log(w^d + \tilde w^\top x')$$ for all $x'$ in an open set $D' \subset \R^{d-1}$.
	
	Since this equality holds on an open set we can take first and second partial derivatives to obtain
	\begin{align}
	\frac{\tilde a^i}{a^d + \tilde a^\top x'} &= \alpha \frac{\tilde v^i}{v^d + \tilde v^\top x'} + (1-\alpha)\frac{\tilde w^i}{w^d + \tilde w^\top x'} \label{first_deriv} \\[1ex]
	\frac{(\tilde a^i)^2}{(a^d + \tilde a^\top x')^2} & = \alpha \frac{(\tilde v^i)^2}{(v^d + \tilde v^\top x')^2} + (1-\alpha) \frac{(\tilde w^i)^2}{(w^d + \tilde w^\top x')^2} \label{second_deriv}
	\end{align}
	for every $x' \in D'$ and $i = 1,\dots,d-1$.
By squaring \eqref{first_deriv} and equating with \eqref{second_deriv} we obtain 
$$ 0 = \alpha(1-\alpha)\(\frac{\tilde v^i}{v^d + \tilde v^\top x'} - \frac{\tilde w^i}{w^d + \tilde w^\top x'}\)^2.$$ Since $\alpha \in (0,1)$ it follows that $\tilde v_i = \frac{v^d + \tilde v^\top x'}{w^d + \tilde w^\top x} \tilde w_i $.
	Now plugging this back into \eqref{first_deriv} gives
	$$ \frac{\tilde a^i}{a^d + \tilde a^\top x'} = \frac{\tilde w^i}{w^d + \tilde w^\top x'}.$$ Putting this all together we see that 
	\begin{equation} \label{linear_relationship}
	\tilde a^i = \frac{a^d + \tilde a^\top x'}{w^d + \tilde w^\top x'} \tilde w^i = \frac{a^d + \tilde a^\top x'}{v^d +\tilde v^\top x'}v^i; \quad x' \in D', \ i=1,\dots,d-1.
	\end{equation}
Multiplying \eqref{linear_relationship} by $x^i$, summing over the coordinates and simplifying yields
\begin{align*}w^d\tilde a^\top x' = a^d \tilde w^\top x', \qquad   v^d \tilde a^\top x' = a^d \tilde v^\top x',&&  x' \in D'
\end{align*}
Again taking derivative with respect to $x^i$ and simplifying we obtain 
\begin{equation} \label{linear_relationship2}
\frac{w^i}{w^d} = \frac{a^i}{a^d} = \frac{v^i}{v^d} 
\end{equation} for every $i$.
It follows that $a,v$ and $w$ are constant multiples of each other and by summing over \eqref{linear_relationship2} we see that the constants $c_1,c_2$ from the statement of theorem can be taken to be $a^\top \boldsymbol{1}/v^\top\boldsymbol{1}$ and $a^\top \boldsymbol{1}/w^\top \boldsymbol{1}$ respectively.  
\end{proof}
We now define the set \begin{equation} \label{E_laff} \partial\mathcal{E}_{\text{laff}} := \left\{\nabla \(\bigwedge_{k=1}^n \log w_i^\top x\) \ \bigg| \ n \in \N, \{w_k\}_{k=1}^n \subset \R^d_{++} \right\}.
\end{equation}
The functions appearing in the definition of $\partial\mathcal{E}_{\text{laff}}$ are logarithms of minimums of hyperplanes over the simplex. The ``laff" in the subscript refers to the functions being \emph{log-affine}. It is well known that any concave function can be pointwise approximated by affine functions. However, for our application to robust growth-optimization it is imperative that we can additionally approximate the asymptotic growth rate of a portfolio generated by a concave function. In view of the representation \eqref{lambda_EC} for $\lambda_{\mathcal{E}}$, this requires showing that gradients of minimums of log-affine functions are dense with respect to the $\mathcal{H}^{c,p}$ norm. This is the content of the next proposition.

\begin{prop} \label{extreme_points}
	The set $\partial\mathcal{E}_{\text{laff}}$ defined in \eqref{E_laff}
	 is dense in $\partial \mathcal{E}$ under the $\mathcal{H}^{c,p}$ norm and every member of $\partial\mathcal{E}_{\text{laff}}$ is an extreme point in $\partial \mathcal{E}$.
\end{prop}
\begin{proof} 
	We fix $n \in \N$, and a collection $\{w_k\}_{k=1}^n$ as in the definition of the set $\partial\mathcal{E}_{\text{laff}}$. Set $g_k(x) := \log(w_k^\top x)$ and $g(x) := \bigwedge_{k=1}^n g_k(x)$. By monotonicity of the logarithm and the fact that the minimum of finitely many concave functions is concave, we see that $g$ is exponentially concave. Now suppose that we have $\nabla g = \alpha \nabla \psi + (1-\alpha)\nabla \phi$ for some $\nabla\psi,\nabla \phi \in \partial \mathcal{E}$ and some $\alpha \in (0,1)$. It follows that there exists a constant $C > 0$ such that
	$g = \alpha \phi+ (1-\alpha)\psi + C$ and we assume without loss of generality that $C = 0$. Define $$D_k = \left\{x \in \Delta_+^{d-1}: g_k(x)  < \min\limits_{j \ne k}g_j(x)\right\}$$ and assume without loss of generality that $D_k \ne \emptyset$ for every $k = 1,\dots,n$ (since otherwise we could exclude the corresponding $g_k$ from the minimum). If $n = 1$ we just set $D_1 = \Delta^{d-1}_+$. It is clear that each $D_k$ is convex and relatively open in $\Delta^{d-1}_+$. Additionally we have that $g(x) = g_k(x)$ on $\bar D_k$ and that $\Delta_+^{d-1} = \cup_{i=1}^n \bar D_k$ where $\bar \cdot$  denotes the closure in $\Delta^{d-1}_+$. Now fix an index $k$ and let $X$ be a random variable whose (essential) range is equal to $\bar D_k$. By convexity of $\bar D_k$ we have that $\E[X] \in \bar D_k$ so we see that
	\begin{align*}
	\alpha \psi(\E[X]) + (1-\alpha)\phi(\E[X])&  = g(\E[X]) & \\
	& = \log\E[w_k^\top X]\\
	& = \log\E[e^{\log(w_k^\top X)}]\\
	& = \log\E[e^{g(X)}]\\
	& = \log\E[e^{\alpha \psi(X)}e^{(1-\alpha)\phi(X)}] \\ 
	& \leq \alpha \log\E[e^{\psi(X)}] + (1-\alpha)\log\E[e^{\phi(X)}]\\
	& \leq \alpha \psi(\E[X]) + (1-\alpha)\phi(\E[X])  
	\end{align*}
	where we used Hölder's inequality with exponent $1/\alpha$ in the first inequality and concavity of $e^\psi$ and $e^\phi$ together with Jensen's inequality in the second one. Thus we see that we have equality all the way through and in particular
	$$ 0 = \alpha\(\log\E[e^{\psi(X)}] - \psi(\E[X])\) + (1-\alpha)\(\log\E[e^{\phi(X)}] -\phi(\E[X]) \).$$
	By Jensen's inequality we have that $\log(\E[e^{\psi(X)}]) = \E[\psi(X)]$. Since we have equality in Jensen's inequality if and only if $e^\psi$ is affine on the (essential) range of $X$ we conclude that $\psi(x) = \log(u_k^\top x)$ for all $x \in \bar D_k$ and some $u_k \in \R^d_{++}$. By a similar argument $\phi(x) = \log(v_k^\top x)$ for every $x \in \bar D_k$ and some $v_k\in \R^d_{++}$. It follows that
	$$\log(w_k^\top x) = \alpha \log(u_k^\top x) + (1-\alpha) \log(v_k^\top x); \quad x \in \bar D_k.$$
Lemma \ref{log_lemma} now implies that $w_k = c_1u_k = c_2 v_k$ for some constants $c_1, c_2 > 0$. As such, we conclude that $g(x) = g_k(x) = \psi(x) + \alpha \log c_1 = \phi(x) + (1-\alpha) \log c_2$ on $D_k$ for every $k$. It follows that $\nabla g(x) = \nabla \psi(x) = \nabla \phi(x)$ proving that $\nabla g$ is an extreme point.

To prove density we argue in a similar way as in the proof of Theorem \ref{partialE_thm}.  Fix $ \nabla \phi \in \partial \mathcal{E}$. Since $e^\phi$ is a positive concave function and $x \mapsto e^x$ is monotone we can find a sequence of vectors $\{v_k\}_{k \in \N} \subseteq \R^{d}_{++}$ such that $\phi(x) = \lim_{n \to \infty} h_n(x)$ for all $x \in \Delta^{d-1}_+$ where $h_n(x) = \bigwedge_{k=1}^n \log (v^\top_k x)$. Since $h_n$ is a decreasing sequence of concave functions converging to a concave function it follows by Theorem 25.7 in \cite{rockafellar1970convex} that $h_n \to \phi$ uniformly and $\nabla h_n \to \nabla \phi$ almost everywhere. As in the proof of Theorem \ref{partialE_thm}, since $\{h_n\}_{n \in \N}$ is a bounded sequence in $\mathcal{H}^{c,p}$ we can find a subsequence that converges weakly. The almost everywhere convergence $\nabla h_n \to \nabla \phi$ implies that this weak limit must be $\nabla \phi$. This, along with the fact that $h_n$ and $\nabla h_n$ are bounded functions for every $n$, implies that 
$$ \partial \mathcal{E} = \overline{\partial \mathcal{E}_{\text{laff}}}^w = \overline{\partial \mathcal{E}_b}^w$$
where $\overline{ \cdot}^w$ denotes the weak closure in $\mathcal{H}^{cp}$ and 
$$ \partial \mathcal{E}_b = \left\{\nabla \phi \in \partial \mathcal{E}: \sup_{x \in \Delta^{d-1}_+} |\phi(x)| < \infty \text{ and } \text{esssup} |\nabla \phi| < \infty\right\}.$$
Since $\partial \mathcal{E}_b$ is a convex set, its strong and weak closure coincide so we see from this that $\partial \mathcal{E}_b$ is dense in $\partial \mathcal{E}$. As such to show that $\partial \mathcal{E}_{\text{laff}}$ (which is not a convex set) is dense in $\partial \mathcal{E}$ it suffices to show that we can strongly approximate any member of $\partial \mathcal{E}_b$.

We now assume that $\nabla \phi \in \partial\mathcal{E}_b$ is given and approximate $\phi,\nabla \phi$ by $h_n$, $\nabla h_n$ as before. In particular $\nabla h_n \to \nabla \phi$ weakly in $\mathcal{H}^{c,p}$. Setting $A =\sup_{x \in \Delta^{d-1}_+} e^\phi(x)$, and $a = \inf_{x \in \Delta^{d-1}_+} e^\phi(x)$ we have by definition of $\partial \mathcal{E}_b$ that $\infty > A \geq a > 0$.  Moreover, we can assume without loss of generality that $h_n(x) \leq  A + 1$ (for example by setting $v_1 = (A+1)\boldsymbol{1}$) for every $n$. For a fixed $x$ we have by definition that $h_n(x) = \log (v_{k(x)}^\top x)$ for some $v_{k(x)}$ where the index $k$ depends on $x$. It follows by the monotone convergence of $e^{h_n}$ to $e^\phi$ that
$$ (A+1) \geq v_{k(x)}^\top x \geq e^\phi(x) \geq a.$$ In particular we see that $1/({v_{k(x)}^\top x}) \leq 1/a$ and that $v_{k(x)}^i \leq (A+1)/x^i$ for every $i$. 
Now fixing $x$ in the full measure set 
$$\{x \in \Delta^{d-1}_+: \phi \text{ and } h_n \text{ are differentiable at } x \text{ for every } n\}$$ 
we obtain $$\partial_ih_n(x) = \frac{v_{k(x)}^i}{v_{k(x)}^\top x} \leq \frac{A+1}{ax^i}.$$
Next we estimate that
\begin{align*} \nabla h_n(x)^\top c(x) \nabla h_n(x) p(x) & = \sum_{i,j=1}^d  c_{ij}(x) \partial_i h_n(x) \partial_j h_n(x)p(x) \\ & \leq  \frac{(A+1)^2}{a^2} \sum_{i,j=1}^d \frac{1}{x^ix^j}|c_{ij}(x)|p(x)
\end{align*}
for almost every $x$.
By Lemma \ref{dom_conv} the function on the right hand side is integrable so by dominated convergence we conclude that 
$\|\nabla \phi\|_{\mathcal{H}^{c,p}} = \lim_{n \to \infty} \|\nabla h_n\|_{\mathcal{H}^{c,p}}$. Since we know that $\nabla h_n \to \nabla \phi$ weakly, it follows that $\nabla h_n \to \nabla \phi$ strongly as required. This completes the proof.
\end{proof}

 \subsection{Numerical Results} \label{numerical}
 In this section we make a specific choice for the model inputs $(c,p)$ to simulate the behaviour of the optimal portfolios encountered in the previous sections in a rank-based model. We make the choice of the symmetric Dirichlet distribution $p(x) = \frac{\Gamma(ad)}{\Gamma(a)^d}\prod_{i=1}^d (x^i)^{a-1}$ as in the volatility-stabilized model of Section~\ref{a_motivating_example}. 
  Figure \ref{fig:curves} plots on a log-log scale the order statistics of  random samples drawn from a Dirichlet distribution with parameters $a=0.5,1,$ and $2$ for $d=500$ and $d=5000$. These samples represent the capital distribution curve for the ranked market weights in the model described above.
 
 \begin{figure}[h]
 	\centering
 	\textbf{Capital Distribution Curve Simulation}
 	\includegraphics[width=1\linewidth]{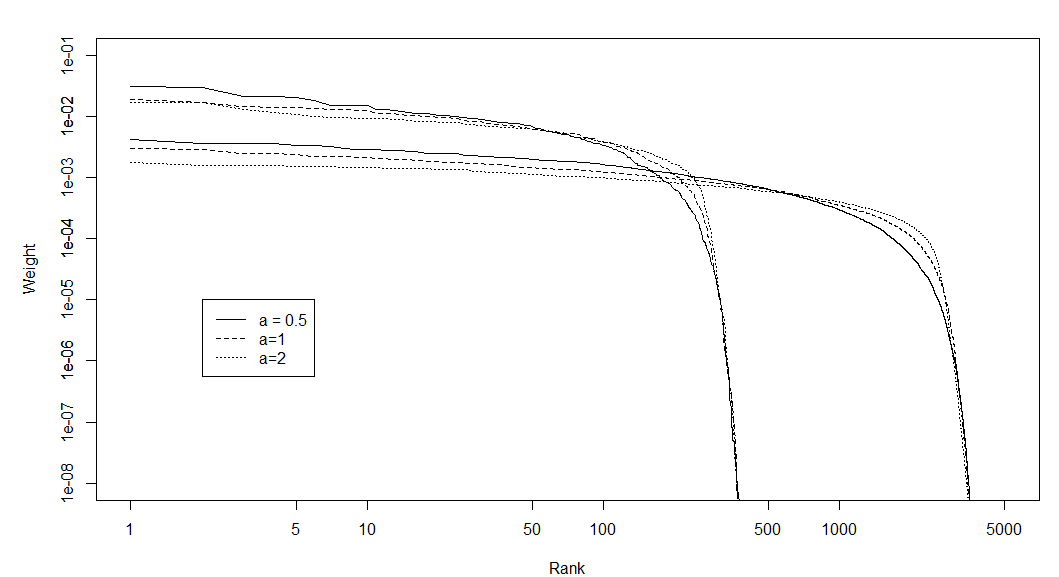}
 	\caption{Capital distribution curves from Dirichlet densities}
 	\label{fig:curves}
 \end{figure}

The shape of the curves are similar to those observed empirically (c.f. Figure 5.1 in \cite{fernholz2002stochastic}), however the largest market weight is not quite as large as seen in real financial markets. Nevertheless, since we are unaware of any better fitting distributions coming from financial models with an analytically available density, we choose the Dirichlet distribution as the invariant density for the ranked market weights in our simulations. We wish to remark that the best choice of density $p$ to take in order to accurately represent real world equity markets is non-trivial. In particular, although the setting adopted in both this paper and \cite{kardaras2018ergodic} consists of a \emph{closed market}, where the number of stocks is fixed, real equity markets do not have this property. This will impact the optimal robust asymptotic growth rate and carries difficulties for modelling the capital distribution curve as the number of assets in real world equity markets change over time. Though, in the recent paper \cite{karatzas2020open} the authors developed a framework to study \emph{open markets} -- where the investor is constrained to only invest in a subset of the market consisting of high capitalization stocks -- the extension of the robust growth problem studied in \cite{kardaras2018ergodic} and in this paper to open markets remains an open question for future research.

  For the $c$ matrix we use the specification \eqref{c_def} with the choices $f_i(x) = x$ and $f_{ij}(x) = \sigma^2$ for every $i,j \in \{1,\dots,d\}$ and some constant $\sigma^2  > 0$. This recovers the volatility structure from the volatility-stabilized model considered in Section~\ref{a_motivating_example} up to the constant volatility factor $\sigma^2$ (in that example we had $\sigma^2=1$). The unconstrained optimal strategy and the optimal portfolio in feedback form for $d =2$ are available analytically from Theorems \ref{KR_main} and \ref{opt_long_only} respectively. In Example \ref{dirichlet_example} we observed that the optimal concave portfolio is given by the optimal portfolio in feedback form, but for $d > 2$ the optimal concave portfolio needs to be estimated numerically. We now describe the numerical scheme used to compute this portfolio. We emphasize that this is for illustrative purposes; we have not attempted to optimize robustness and accuracy.

We fix $K,M \in \N$ and select vectors $a_{km} \in \R^d_{++}$ such that $a_{km}/d$ is uniformly distributed on the simplex for $k = 1\dots,K$ and $m = 1, \dots, M$. Then we set $$\phi_m(x) = \min_{k \in\{1,\dots,K\}} \log(a_{km}^\top x)$$ for $m = 1,\dots,M$ so that $\nabla \phi_m$ is an extreme point of the set $\partial \mathcal{E}$ by Proposition \ref{extreme_points}. Next we define $$\psi^\mu(x)= \sum_{m=1}^M \mu_m \phi_m(x)$$ for parameters $\mu_m$ such that $\mu_m \geq 0$ for every $m$ and $\sum_{m=1}^M \mu_m = 1$. Relying on the density and extremity of $\partial\mathcal{E}_{\text{laff}}$ in $\partial \mathcal{E}$ we replace the infinite dimensional problem \eqref{EC} by the finite dimensional problem 
$$\min\limits_{\mu} \|\nabla \psi^\mu - \ell\|^2_{\mathcal{H}^{c,p}}.$$
Notice that for a given convex combination $\mu$ we have that  $$\|\nabla \psi^\mu - \ell\|^2_{\mathcal{H}^{c,p}} = \E[ (\nabla \psi^\mu - \ell)^\top c (\nabla \psi^\mu - \ell)(Y)]$$ where $Y \sim p$. We then approximate this quantity by $\frac{1}{N}\sum_{n=1}^N    (\nabla \psi^\mu - \ell)^\top c (\nabla \psi^\mu - \ell)(Y_n)$ for a large value of $N$ and iid samples $Y_1,\dots,Y_N \sim p$.
A calculation shows that 
\begin{equation} \label{quad_prog}
\|\nabla \psi^\mu - \ell\|^2_{\mathcal{H}^{c,p}} \approx \frac{1}{2}\mu^\top \(\frac{1}{N}\sum\limits_{n=1}^N Q(Y_n)\)\mu - \mu^\top \(\frac{1}{N}\sum\limits_{n=1}^N r(Y_n)\) + C
\end{equation}
where $C$ is some constant independent of $\mu$, while
 $Q :\Delta^{d-1}_+ \to \R^{M \times M}$ and $r: \Delta^{d-1}_+ \to \R^M$ are given by
\begin{align*}
Q_{ij}(x) & = 2 \nabla \phi_i^\top c \nabla \phi_j(x) \\
r_{i}(x) & = 2 \nabla \phi_i^\top c \ell(x).
\end{align*}
As before $\nabla \phi_i$ refers to a version of the superdifferential of $\phi_i$.
It follows that $Q$ is positive-semidefinite with $v^\top Q v = (\sum_{m=1}^M v_m \nabla \phi_m)^\top c (\sum_{m=1}^M v_m \nabla \phi_m)$ for every $v \in \R^M$. Thus, minimizing \eqref{quad_prog} over $\mu$ is a quadratic programming problem that can be numerically solved. Using the optimal solution $\hat \mu_N$ we simulate the portfolio generated by $\exp \psi^{\hat \mu_N}$. Note that, $\hat \mu_N$ depends on the samples $Y_1,\dots,Y_N$ and so is a random concave function on the simplex -- such functions and their properties have recently been studied in \cite{baxendale2019random}.
 
 We now present simulations under the worst case measure $\tilde \P$, under which the coordinate process $X$ has dynamics given by \eqref{worst_case_dynamics}. Figure \ref{fig:d=2} shows the result in the $d=2$ case. We used the aforementioned algorithm to estimate the optimal concave portfolio, represented by the green curve, which as mentioned is the same as the explicitly known optimal portfolio in feedback form, represented by the red curve.
\begin{figure}[h]
 	\centering
 	\textbf{$\boldsymbol{d=2}$ Simulation}
 	\includegraphics[width=1\linewidth]{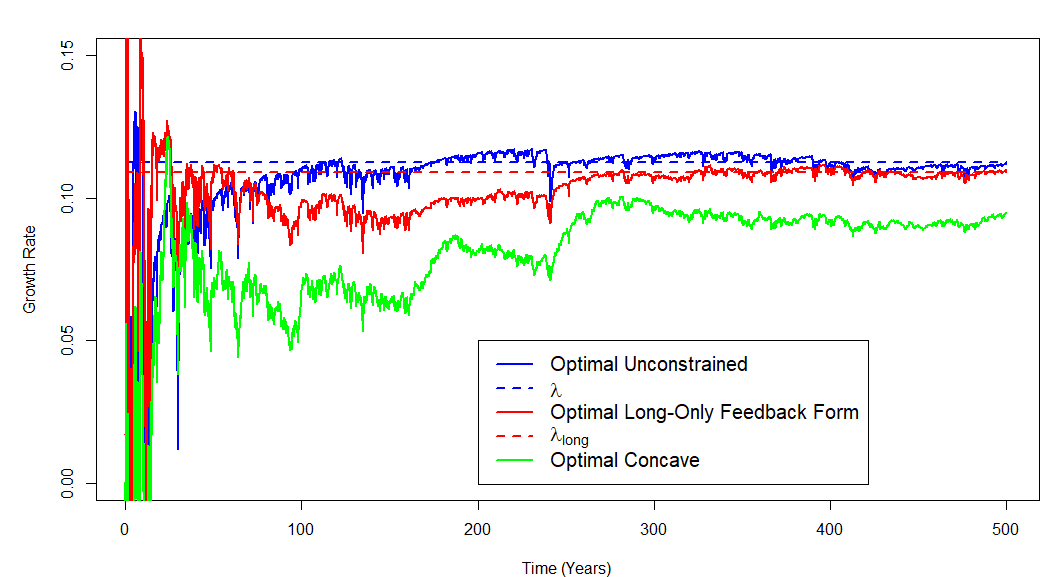}
 	
 	\caption{Parameters are $a=3,\sigma^2 = 0.1, M=25, K = 100, N=100$} 
 	\label{fig:d=2}
 \end{figure}
 As predicted by the theoretical calculations, the optimal unconstrained portfolio performs best, closely followed by the optimal long-only portfolio in feedback form ahead of our estimated optimal concave strategy. In this case we see that the attained growth rates are of the same order of magnitude though the relationship $\lambda > \lambda_{\text{long}} = \lambda_{\mathcal{E}}$ persists. This is in contrast to the results of Figure \ref{fig:d=20}. In this case we set $d=20$ and observe that $\lambda_{\mathcal{E}}$ is an order of magnitude lower than $\lambda$ (note the use of logarithmic scale for the $y$-axis). It is not surprising to see a larger disparity between the two growth rates in a higher dimension model. Indeed, as the dimension increases, the region of the simplex where the unconstrained strategy takes on a long position grows smaller relative to the size of the entire simplex. As such, the market weights do not spend as much (if any) time in this region when $d=20$ in comparison to when $d=2$. The interpretation is that a long-only investor is missing out more heavily on the possible gains from short-selling when the number of assets is large. This leads to a relatively smaller growth rate. It is unclear, however, how much of the observed loss in growth rate is purely due to the long-only restriction, how much is due to our further restriction to consider only those portfolios generated by  concave functions and how much is additionally lost from the aforementioned numerical implementation. It remains an open problem to identify tight upper bounds on the growth rates for these constrained-optimal portfolios.

 \begin{figure}[H]
	\centering
	\textbf{$\boldsymbol{d = 20}$ Simulation}
		\includegraphics[width=1\linewidth]{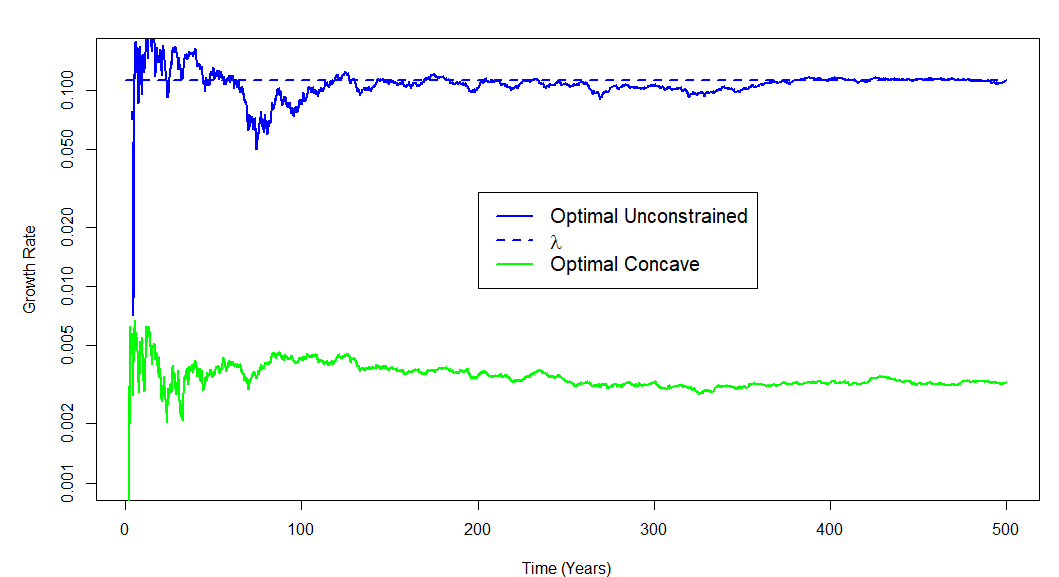}
	\caption{Parameters are $a=3,\sigma^2 = 5\times 10^{-4}, M=25, K = 100, N=100$.} 
	\label{fig:d=20}
\end{figure}

	\appendix
	
	\section{Generalized Martingale Problem on \texorpdfstring{$\Delta^{d-1}_+$}{the Simplex}}\label{mart_prob_app}
	The basic reference for the material in this section is \cite{pinsky1995positive}.
	Assumption \ref{finite_assumption}(iii) introduces a generalized martingale problem, for which having a nonexplosive solution is imperative for the analysis in the paper to hold. In this section we rigorously introduce the generalized martingale problem and using test function techniques prove Proposition \ref{sufficient}. 
	
	First we view $\Delta^{d-1}_+$ as an open subset of $\R^{d-1}$ and identify $\partial \Delta^{d-1}_+$ with a single absorbing state $\Theta$. That is, formally we define the domain of our generalized martingale problem to be $\tilde \Delta^{d-1}_+ := \Delta^{d-1}_+ \cup \{\Theta\}$; the one-point compactification of $\Delta^{d-1}_+$. Next, since the simplex has a non-smooth boundary we will need a localization procedure to define the generalized martingale problem. To this end, let $\{E_n\}_{n \in \N}$ be an increasing sequence of open connected domains such that (i) $\bar E_n \subseteq E_{n+1}$ for every $n \in \N$, (ii) $\partial E_n$ is of class $C^{2,\beta}$ for every $n \in \N$ and for some $\beta \in (0,1]$ and (iii) $\Delta^{d-1}_+ = \cup_nE_n$. For example one could take
	$$E_n = \left\{x \in \Delta^{d-1}_+: \prod_{i=1}^d x^i > 1/n\right\}.$$ Define the exit times $\tau_n :=\inf\{t \geq 0: X_t \not \in E_n\}$ and $\tau := \lim_{n \to \infty} \tau_n$. The corresponding sample space is
	$$ \tilde \Omega := \left\{\omega \in C([0,\infty);\tilde \Delta^{d-1}_+): \omega_{\tau + t} = \Theta \text{ for all } t \geq 0\right\}.$$
We denote by $\tilde X$ the coordinate process on $\tilde \Omega$, define $\tilde \F$ to be the Borel $\sigma$-algebra on $\tilde \Omega$ and let $\tilde {\mathbb{F}}$ be the right-continuous modification of the filtration generated by $\tilde X$. Now we are ready to define the generalized martingale problem.
	
	\begin{defn} \label{gen_mart}
	The generalized martingale problem on $ \Delta^{d-1}_+$ associated to an operator $\tilde L$ is for every $x \in \tilde \Delta^{d-1}_+$ to find a probability measure $\tilde \P_x$ on $(\tilde \Omega,\tilde \F)$ such that 
	\begin{enumerate}[label = ({\alph*})]
		\item $\tilde \P_x(\tilde X_0 = x) = 1$,
		\item $f(\tilde X_{T \land \tau_n}) - \int_0^{T \land \tau_n} \tilde Lf(\tilde X_t)dt$ is a martingale on $(\tilde \Omega, \tilde \F, \tilde {\mathbb{F}},  \tilde \P_x)$ for every $n \in \N$ and $f \in C^2(\Delta^{d-1}_+)$.
	\end{enumerate}
We say that the solution is nonexplosive if $\tilde \P_x(\tau < \infty) =  0$ for every $x \in \Delta^{d-1}_+$. 
	\end{defn}
We note that a nonexplosive solution to the generalized martingale problem does not charge the set of paths that reach $\Theta$ and as such a nonexplosive solution $\tilde \P_x$ can be viewed as a measure on $(\Omega,\F)$.

Now that the generalized martingale problem is introduced, we set our sights on establishing criteria for when a nonexplosive solution exists. As such, we assume that $c$ is of the form \eqref{c_def} and satisfies Assumption~\ref{function_assumption}. We define the function $R = p^{1/2}g^{1/2}\prod_{i=1}^d f_i^{1/2}$ and using the notation of Assumption~\ref{finite_assumption}(iii) define the operator $L^R$ via $L^Rh := Lh - \nabla \log R ^\top c \nabla h$ for every $h \in C^2(\Delta^{d-1}_+)$.
The following theorem from \cite{pinsky1995positive} establishes existence and uniqueness for the generalized martingale problem associated to this operator:

\begin{thm} [Theorem 1.13.1 in \cite{pinsky1995positive}] \label{nonexplosion} Under Assumption \ref{input_assumption}(i) on the instantaneous covariation matrix $c$ there exits a unique solution to the generalized martingale corresponding to the operator $L^R$.
\end{thm}
In general the unique solution guaranteed by this theorem may explode in finite time. Thus, to verify Assumption \ref{finite_assumption}(iii), one must ensure that the solution is nonexplosive. We will establish sufficient conditions on the inputs $g,p$ and $f_i$ for $i=1,\dots,d$ that guarantee this is the case. The main tool to establish this is the following test-function method:
\begin{prop}[Theorem 6.7.1(i) in \cite{pinsky1995positive}] \label{test_function} Let $\tilde L$ be an operator for which there exists a unique solution to the generalized martingale problem given in Definition \ref{gen_mart}. Assume there exists a constant $\lambda > 0$ and a function $u \in C^2(\Delta^{d-1}_+)$ such that $\tilde  Lu(x) \leq \lambda u(x)$ for every $x \in \Delta^{d-1}_+$ and $\lim_{x \to \partial \Delta^{d-1}_+} u(x) = \infty$. Then the solution to the generalized martingale problem corresponding to $\tilde L$ does not explode.
\end{prop}
With these preliminary results established we are ready to prove Proposition~\ref{sufficient}.
 
\begin{proof} [Proof of Proposition \ref{sufficient} ]
	We note that $\ell = \nabla \log R$ and a direct calculation shows that 
	$$ \frac{1}{2}\ell^\top c \ell = \frac{LR}{R} - L\log R$$ so that Assumption~\ref{finite_assumption}$(i)$ holds by the integrability assumptions on $\frac{LR}{R}$ and $L \log R$. 	Next we compute that 
	\begin{align*}\diver{pc\ell} &= \sum\limits_{i=1}^d \partial_i\(\sum\limits_{j=1}^d c_{ij} \ell_j p\) = \sum\limits_{i,j=1}^d \partial_i c_{ij}\ell_jp + c_{ij}\partial_i\ell_j p + c_{ij}\ell_j \partial_ip\\
	&= p \diver c^\top \ell + 2p L\log R + \nabla p c \ell.
	\end{align*}
Recalling that $\ell = \frac{1}{2}(c^{-1}\diver c + \nabla \log p)$ we obtain
	 $$\diver{pc\ell}= 2\ell^\top c \ell p + 2pL\log R.$$
	Thus we again see that Assumption~\ref{finite_assumption}(ii) holds by the integrability assumptions. 
	
	Next we note by Theorem \ref{nonexplosion} and Proposition \ref{test_function} it is enough to find a function $u \in C^2(\Delta^{d-1}_+)$ such that $\lim_{x \to \partial \Delta^{d-1}_+} u(x) = \infty$, $u > 0$ on $\Delta^{d-1}_+$  and $ L^{R}u \leq \lambda u$ for some $\lambda > 0$ to verify that Assumption~\ref{finite_assumption}(iii) holds. A computation yields that $L^Ru = \frac{L(uR)}{R} - \frac{LR}{R}u$ for every $u \in C^2(\Delta^{d-1}_+)$. Choosing $u = 1/R$ we see that
	$$L^R u \leq \lambda u \iff -\frac{LR}{R} \leq \lambda$$ and the right hand side holds for some $\lambda > 0$ by the assumption that $LR/R$ is bounded from below. Moreover, since $R^{-1} > 0$ on $\Delta^{d-1}_+$ and $R(x) \to 0$ as $x \to \partial\Delta^{d-1}_+$ we see that this test function possesses all the required properties.
\end{proof}

\section{Proof of Theorem~\ref{concave_invariance_thm}} \label{proof_appendix}
The purpose of this section is to prove Theorem~\ref{concave_invariance_thm}. We first establish a few preparatory lemmas.

\begin{lem} \label{uniform_lemma}
	Let $h_n: \Delta^{d-1}_+ \to \R^d$ be such that $\text{supp}(h_n) \subseteq K$ for some compact set $K \subseteq \Delta^{d-1}_+$. Further assume that $\sup_n\sup_{x \in \Delta^{d-1}_+} |h_n(x)| < \infty$ and that $h_n(x) \to 0$ for almost every $x \in \Delta^{d-1}_+$. Then
	\begin{enumerate}[label = ({\roman*})]
		\item for every $\P \in \Pi$ $$\lim\limits_{T_0 \to \infty} \lim\limits_{n \to \infty} \sup\limits_{T \in [T_0,\infty)} \E\left[\frac{1}{T}\int_{T_0}^T \left|h_n(X_t)\right|dt\right] = 0$$
		\item
		for every $\P \in \tilde \Pi$ $$\lim\limits_{T_0 \to \infty} \lim\limits_{n \to \infty} \sup\limits_{T \in [T_0,\infty)} \E\left[\left|\frac{1}{T}\int_{T_0}^T h_n(X_t)^\top dX_t\right|\right] = 0.$$
	\end{enumerate}
\end{lem}
\begin{proof}
	Fix a measure $\P \in \Pi$. Recall that we view $\Delta^{d-1}_+$ as an open subset of $\R^{d-1}$. Since $|h_n| \to 0$ almost everywhere, by Egorov's Theorem for any $\epsilon > 0$ we can find a set $U_\epsilon$ with $\mathcal{L}^{d-1}(U_\epsilon) < \epsilon$ such that $|h_n|\to 0$ uniformly on $U_\epsilon^c$ where $\mathcal{L}^{d-1}$ is the Lebesgue measure on $\R^{d-1}$. Moreover we can assume without loss of generality that $U_\epsilon \subseteq U_\delta$ if $\epsilon < \delta$.
	Then we see that
	\begin{align*}\frac{1}{T}\int_{T_0}^T  \E[|h_n(X_t)|]dt & \leq \frac{1}{T}\int_{T_0}^T  \E[|h_n(X_t)|\I_{U_\epsilon^c}(X_t)]dt + \frac{1}{T}\int_{T_0}^T  \E[|h_n(X_t)|\I_{U_\epsilon}(X_t)]dt\\
		& \leq \sup\limits_{x \in U_{\epsilon}^c} |h_n(x)| + C\frac{1}{T}\int_{T_0}^T\P(X_t \in U_\epsilon)dt.
	\end{align*}
	where $C = \sup_n \sup_{x \in \Delta^{d-1}_+} |h_n(x)|< \infty$.
	Taking $\sup$ over $T$ and sending $n \to \infty$ we obtain
	\begin{equation} \label{uniform_lemma_estimate}
		\lim_{n \to \infty}\sup\limits_{T \in [T_0,\infty)} \E\left[\frac{1}{T}\int_{T_0}^T |h_n(X_t)|dt\right] \leq C\sup\limits_{T \in [T_0,\infty)} \frac{1}{T}\int_{T_0}^T \P(X_t \in U_\epsilon)dt
	\end{equation}
	for every $\epsilon > 0$.
	Now define the functions $g_\epsilon:[T_0,\infty] \to \R$ via $$g_\epsilon(T) = \begin{cases}
		\frac{1}{T}\int_{T_0}^T \P(X_t \in U_\epsilon)dt & T < \infty \\
		\int_{\Delta^{d-1}_+} \I_{U_\epsilon}(x)p(x)dx & T = \infty.
	\end{cases}
	$$
	We see that $g_\epsilon$ is continuous for every $\epsilon > 0$; indeed this is immediate for $T < \infty$. To establish continuity at $T = \infty$ we use the ergodic property to obtain 
	$$\lim\limits_{T\to \infty}\frac{1}{T}\int_{T_0}^T \I_{U_\epsilon}(X_t)dt = \int_{\Delta^{d-1}_+} \I_{U_\epsilon}(x)p(x)dx; \quad \P\text{-a.s.}$$ Thus by bounded convergence and Fubini it follows that 
	$$ \lim\limits_{T \to \infty}\frac{1}{T} \int_{T_0}^T\P(X_t \in U_\epsilon)dt = \int_{\Delta^{d-1}_+} \I_{U_\epsilon}(x)p(x)dx. $$
	Moreover we see that $g_\epsilon(T)$ is decreasing in $\epsilon$ for every $T$ and 
	\begin{align*}
		g_0(T) :=\lim\limits_{ \epsilon \to 0} g_\epsilon (T) &= \begin{cases}
			\frac{1}{T}\int_{T_0}^T \P(X_t \in U_0)dt & T < \infty \\
			\int_{\Delta^{d-1}_+} \I_{U_0}(x)p(x)dx & T = \infty \end{cases} \\
		& = \label{g0}
		\begin{cases}
			\frac{1}{T}\int_{T_0}^T \P(X_t \in U_0)dt & T < \infty \\
			0 & T = \infty \end{cases}
	\end{align*}
	where $U_0 = \cap_\epsilon U_\epsilon$ is a Lebesgue null set. Again by the ergodic property we see that $g_0$ is continuous so it follows by Dini's Theorem, since $[T_0,\infty]$ is a compact set, that $g_\epsilon \to g_0$ uniformly in $T$. Hence we obtain
	$$\lim\limits_{\epsilon \to 0} \sup\limits_{T \in [T_0,\infty]} g_\epsilon (T) = \sup\limits_{T \in [T_0,\infty]}g_0(T).$$
	This together with \eqref{uniform_lemma_estimate} yields 
	\begin{align*}
		\lim\limits_{T_0 \to \infty} &\lim\limits_{n \to \infty} \sup\limits_{T \in [T_0,\infty)} \E\left[\frac{1}{T}\int_{T_0}^T \left|h_n(X_t)\right|dt \right] 
		\leq C\lim\limits_{T_0 \to \infty} \sup \limits_{T \in [T_0,\infty]} g_0(T)
		= C\limsup\limits_{T_0 \to \infty} g_0(T_0) = 0
	\end{align*}
	proving (i).
	
	To prove (ii) we let $\P \in \tilde \Pi$ be given and write $b^\P$ for the drift process of $X$ under $\P$ as in Definition~\ref{tilde_pi_def}.
	Now we estimate
	\begin{align*} &\E\left[\left|\frac{1}{T}\int_{T_0}^T  h_n(X_t)^\top dX_t\right|  \right] \\
		&\leq \E\left[\left| \frac{1}{T}\int_{T_0}^Th_n(X_t)^\top \sigma(X_t)dW_t \right|  \right]  + \E\left[\frac{1}{T}\int_{T_0}^T\sum_{i=1}^d\left| h_n^i(X_t) (b_t^\P)^i \right|dt  \right]
		\\
		&\leq \frac{1}{T} \( \int_{T_0}^T  \E\left[ h_n(X_t)^\top c(X_t)h_n(X_t)\right]dt\)  ^{1/2} \\ 
		& \hspace{1cm} + \(\E\left[\frac{1}{T}\int_{T_0}^T \sum_{i=1}^d|h_n^i(X_t) |^{q'}dt\right]\)^{1/q'}\(\frac{1}{T}\int_{T_0}^T \sum_{i=1}^d \E[|(b_t^\P)^i|^q\I_{K}(X_t)]dt\)^{1/q}. 
	\end{align*}
	We used Hölder's inequality in the last inequality. Here $q \in (1,\infty)$ is as in \eqref{drift_assumption} and $q' \in (1,\infty)$ is such that $1/q + 1/q' = 1$. Since $c$ is bounded on $K$ we can apply part (i) to both terms together with \eqref{drift_assumption} to the second term above proving part (ii).
\end{proof}
\begin{remark}
	The drift assumption \eqref{drift_assumption} is only used to prove (ii). This uniform ergodic convergence result is needed to establish Theorem~\ref{concave_invariance_thm} so any other sufficient condition on the drift $b_t^\P$ under which Lemma~\ref{uniform_lemma}~(ii) holds will establish Theorem~\ref{concave_invariance_thm} for the corresponding measure $\P$.
\end{remark}
\begin{lem} \label{cutoff_lemma} 
	Let Assumptions~\ref{input_assumption} and \ref{finite_assumption} hold. Then there exist functions $\{\eta_m\}_{m=1}^\infty \subset C_c^\infty(\Delta^{d-1}_+)$ which satisfy the following conditions:
	\begin{enumerate}[label = ({\roman*}),noitemsep]
		\item $0 \leq \eta_m  \leq 1$ for every $m \in \N$,
		\item $\eta_m \to 1$ pointwise as $m \to \infty$,
		\item $\nabla \eta_m \to 0$ in $\mathcal{H}^{c,p}$ as $m \to \infty$.
	\end{enumerate}
\end{lem} 
\begin{proof}
	Fix $m \in \N$ and define the function $u_m:\R \to \R$ via 
	\[u_m(x) = (mx-1)^+\land 1.\]
	Then $u_m$ is $m$-Lipschitz and satisfies $u_m(x) = 1$ for $x \geq 2/m$ and $u_m(x) = 0$ for $x \leq 1/m$. 
	Let $\psi_m \in C_c^\infty(\R)$ be given such that $\psi_m \geq 0$, $\int_{\R} \psi_m = 1$ and $\text{supp}(\psi_m) \subset (0,\frac{1}{2m})$. Define $v_m = u_m *\psi_m$ so that $v_m \in C_c^\infty(\R)$. By construction $v_m(x) = 1$ for $x \geq \frac{5}{2m}$ and $v_m(x) = 0$ for $x \leq \frac{1}{2m}$. Hence it follows that $v_m'(x) = 0$ for $x \not \in (\frac{1}{2m},\frac{5}{2m})$. Moreover since $u_m$ is $m$-Lipschitz it follows by properties of convolution that $v_m$ is too. We obtain for $\frac{1}{2m} \leq x \leq \frac{5}{2m}$ that 
	\begin{equation} \label{v_m_bound}
		|v_m'(x)| \leq m \leq \frac{5}{2x},
	\end{equation} so that $|v_m'(x)| \leq \frac{5}{2x}$ for every $x \in (0,1)$. Now we define the function $\eta_m: \Delta^{d-1}_+ \to \R$ via $\eta_m(x) = \prod_{i=1}^d v_m(x^i)$. By construction $\eta_m \in C_c^\infty(\Delta^{d-1}_+)$ and satisfies the first two items in the statement of the lemma. Thus it just remains to prove that $\nabla\eta_m \to 0$ in $\mathcal{H}^{c,p}$ as $m \to \infty$. Note that $\partial_i \eta_m(x) = v_m'(x^i)\prod_{j \ne i} v_m(x^j)$ so that $\partial_i\eta_m(x) \to 0$ as $m \to \infty$ for ever $x \in \Delta^{d-1}_+$ and $i=1,\dots,d$. Next, using \eqref{v_m_bound}, we estimate that
	\[\nabla \eta_m(x)^\top c(x) \nabla\eta_m(x)p(x) \leq \sum_{i,j=1}^d |\partial_i \eta_m(x)||\partial_j \eta_m(x)||c_{ij}(x)|p(x) \leq \frac{25}{4}\sum_{i,j=1}^d \frac{|c_{ij}(x)|}{x^ix^j}p(x).\]
	By Lemma~\ref{dom_conv} the expression on the right hand side is integrable over $\Delta^{d-1}_+$. Hence by dominated convergence it follows that $\nabla\eta_m \to 0$ in $\mathcal{H}^{c,p}$ completing the proof.
\end{proof}
Now we are ready to prove Theorem~\ref{concave_invariance_thm}.
\begin{proof}[Proof of Theorem~\ref{concave_invariance_thm}] {}
	Fix a positive concave function $G$ on the simplex and a measure $\P \in \tilde \Pi$. Recall that we view $\Delta^{d-1}_+$ as an open subset of $\R^{d-1}$. Denote $\log G$ by $\phi$ and extend $\phi$ to $\R^{d-1}$ by setting $\phi(x) = 0$ for $x \not \in \Delta^{d-1}_+$. Let $\psi \in C_c^\infty(\R^{d-1})$ be such that $\text{supp}(\psi) \subseteq \overline{ B(0,1)}$ and $\int_{\R^{d-1}} \psi = 1$. Set $\psi_n = n^{d-1} \psi(nx)$ and $\phi_n = \phi*\psi_n$. By properties of convolution $\phi_n \to \phi$ pointwise on $\Delta^{d-1}_+$ and $\phi_n$ converges to $\phi$ uniformly on every compact set $K \subset \Delta^{d-1}_+$. Also $\nabla \phi_n \to \nabla \phi$ almost everywhere. By It\^o's formula we have that 
	\begin{align*}
		\int_0^T \nabla \phi_n(X_t)^\top dX_t &= \phi_n(X_T) -\int_0^TL\phi_n(X_t)dt\\
		\intertext{
			and since $\phi$ is exponentially concave we have that} 
		\int_0^T \nabla \phi(X_t)^\top dX_t & = \phi(X_T) - A_T
	\end{align*} for some process $A_t$ of finite variation.
	To keep the notation clearer we will set $G_n = \exp(\phi_n)$ and write $V$ and $V^n$ for the wealth process generated by $G$ and $G_n$ respectively. Note that $G_n$ may not be a concave function. We have 
	\begin{align*}d\log(V^n_t) = \nabla \phi_n(X_t)^\top dX_t - \frac{1}{2}\nabla \phi_n^\top c \nabla \phi_n(X_t)dt &= d\phi_n(X_t) - \(\frac{1}{2}\nabla \phi_n^\top c \nabla \phi_n +L\phi_n\)(X_t)dt,\\
		d\log(V_t) = \nabla \phi(X_t)^\top dX_t - \frac{1}{2}\nabla \phi^\top c \nabla \phi(X_t)dt &= d\phi(X_t) - \frac{1}{2}\nabla \phi^\top c \nabla \phi(X_t) dt - dA_t.
	\end{align*}
	Next note that both $\phi_n(X_T)/T$ and $\phi(X_T)/T$ converge to $0$ in probability as $T \to \infty$ for every $n$. Indeed, this follows from the tightness of the laws $\{X_T\}_{T\geq 0}$ as in the proof of Lemma~\ref{C2_lem}. Hence to determine $g(V;\P)$ and $g(V^n;\P)$ it suffices to study the terms
	\begin{align*}\Gamma_T & := - \frac{1}{2}\int_0^T\nabla \phi^\top c \nabla \phi (X_t)dt - A_T\\
		\Gamma^n_T &:= -\int_0^T (\frac{1}{2}\nabla \phi_n^\top c \nabla \phi_n +L\phi_n)(X_t)dt  = \int_0^T \frac{-LG_n}{G_n}(X_t)dt.
	\end{align*}
	$\Gamma^n$ and $\Gamma$ are the drift processes (in the sense of Definition~\ref{def_func_gen}) for the portfolios $\pi_{G^n}$ and $\pi_{G}$ respectively. Since $G$ is a concave function it follows that $\Gamma$ is an increasing process (see e.g.\ Theorem~3.7 in \cite{karatzas2017trading}).
	
	Next for $m \in \N$ let $\eta_m$ be as in Lemma~\ref{cutoff_lemma}. For each $m$ and $n$ define the processes
	\begin{align*}
		\Gamma^{m,n}_T& := \int_0^T \eta_m(X_t) d\Gamma^n_t = -\int_0^T (\frac{1}{2}\nabla \phi_n^\top c \nabla \phi_n +L\phi_n)\eta_m(X_t)dt\\
		\Gamma^m_T & := \int_0^T \eta_m(X_t)d\Gamma_t = -\int_0^T (\frac{1}{2}\nabla \phi^\top c \nabla \phi\eta_m(X_t) dt + \eta_m(X_t)dA_t).
	\end{align*}
	
	We will show that
	\begin{equation} \label{uniform_convergence}
		\lim\limits_{T_0 \to \infty} \lim\limits_{n \to \infty}\sup\limits_{T \in [T_0,\infty)} \E\left[ \left|\frac{\Gamma_{T}^{m,n}}{T} - \frac{\Gamma_T^m}{T}\right|\right] = 0
	\end{equation} 
	for every $m \in \N$.
	Note that since $\nabla \phi^\top c \nabla \phi\eta_m$ and $\nabla \phi_n^\top c \nabla \phi_n \eta_m$ are uniformly bounded in $x$ and $n$ for each fixed $m$ we have by Lemma \ref{uniform_lemma}(i) that
	\begin{align}
		\label{uniform_first_term}
		\lim\limits_{T_0 \to \infty}\lim\limits_{n \to \infty} \sup\limits_{T \in [T_0,\infty]} \E\left[\left| \frac{1}{T}\int_{T_0}^T (\nabla \phi_n^\top c \nabla  \phi_n\eta_m(X_t) - \nabla \phi^\top c \nabla \phi\eta_m(X_t))dt\right|\right]= 0.
	\end{align}
	Next by It\^o's formula we have that 
	\begin{flalign*}
		& \left| \frac{1}{T}\int_0^T \eta_m(X_t)\(L\phi_n(X_t)dt - dA_t\)\right|&\\
		& = \left|\frac{1}{T}\int_0^T \eta_m(X_t)d(\phi_n-\phi)(X_t) - \frac{1}{T}\int_0^T \eta_m(X_t)\nabla (\phi_n-\phi)(X_t)^\top dX_t\right|&\\
		& \leq \left|\frac{1}{T}\int_0^T \eta_m(X_t)d(\phi_n-\phi)(X_t)\right| + \left|\frac{1}{T}\int_0^T \eta_m(X_t)\nabla (\phi_n-\phi)(X_t)^\top dX_t\right|.&
	\end{flalign*}
	From Lemma \ref{uniform_lemma}(ii) we see that
	\begin{equation} \label{uniform_second_term}\lim\limits_{T_0 \to \infty}\lim\limits_{n \to \infty}\sup\limits_{T \in [T_0,\infty)}  \E\left[\left|\frac{1}{T}\int_0^T \eta_m(X_t)\nabla (\phi_n-\phi)(X_t)^\top dX_t\right|\right] = 0.
	\end{equation}
	For the first term we use product rule to compute that 
	\begin{align*}
		&\frac{1}{T}\int_0^T \eta_m(X_t)d(\phi_n-\phi)(X_t) \\
		&  = \frac{1}{T}\eta_m(\phi_n -\phi)(X_T)- \frac{1}{T}\int_0^T (\phi_n - \phi)(X_t)d\eta_m(X_t) - \frac{1}{T}\bigg\langle (\phi_n-\phi)(X),\eta_m(X)\bigg\rangle_T\\
		& = \frac{1}{T}\eta_m(\phi_n -\phi)(X_T)- \frac{1}{T}\int_0^T (\phi_n - \phi)\nabla \eta_m(X_t)^\top dX_t - \frac{1}{T}\int_0^T (\phi_n - \phi)L \eta_m(X_t)dt \\
		& \hspace{8.1 cm}- \frac{1}{T}\int_0^T \nabla(\phi_n-\phi)^\top c\nabla \eta_m(X_t)dt.
	\end{align*} 
	Invoking Lemma \ref{uniform_lemma} for the stochastic and Lebesgue integral terms along with the fact that $\phi_n \to \phi$ uniformly on $\text{supp}(\eta_m)$ yields 
	\begin{equation} \label{uniform_third_term}
		\lim\limits_{T_0 \to \infty} \lim\limits_{n \to \infty} \sup\limits_{T \in [T_0,\infty)} \E\left[\left|\frac{1}{T}\int_0^T \eta_m(X_t)d(\phi_n-\phi)(X_t)\right|\right] = 0.
	\end{equation} 
	The estimates \eqref{uniform_first_term}, \eqref{uniform_second_term} and \eqref{uniform_third_term} together prove \eqref{uniform_convergence}.
	
	Now define
	\begin{align*}
		\gamma_{m,n} := &  \int_{\Delta^{d-1}_+} \frac{-LG_n}{G_n}\eta_mp\\
		= &  \frac{1}{2} \int_{\Delta^{d-1}_+}(\ell^\top c \ell - (\ell - \nabla \phi_n)^\top c (\ell -\nabla \phi_n))\eta_mp + \int_{\Delta^{d-1}_+} \nabla \eta_m^\top c \nabla \phi_np,
	\end{align*}
	where the second identity follows from integration by parts.
	We additionally define the quantities
	\begin{align*}\gamma_m & := \lim_{n \to \infty} \gamma_{m,n} = \frac{1}{2}\int_{\Delta^{d-1}_+}(\ell^\top c \ell - (\ell - \nabla \phi)^\top c (\ell -\nabla \phi))\eta_mp + \int_{\Delta^{d-1}_+} \nabla \eta_m^\top c \nabla \phi p, \\
		\gamma & := \lim_{m \to \infty} \gamma_m =\frac{1}{2} \int_{\Delta^{d-1}_+}(\ell^\top c \ell - (\ell - \nabla \phi)^\top c (\ell -\nabla \phi))p.
	\end{align*}
	The identity for $\gamma_m$ follows from the bounded convergence theorem while the identity for $\gamma$ follows from the fact that $\eta_m \to 1$ pointwise and $\nabla \eta_m \to 0$ in $\mathcal{H}^{c,p}$ as $m \to \infty$ by Lemma~\ref{cutoff_lemma}. Now we estimate that 
	\begin{align*}
		\E\left[\left|\gamma_m - \frac{\Gamma^{m}_T}{T}\right|\right] & \leq |\gamma_m - \gamma_{m,n}| + \E\left[\left|\gamma_{m,n} - \frac{\Gamma^{m,n}_T}{T}\right|\right] + \E\left[\left|\frac{\Gamma^{m,n}_T}{T} - \frac{\Gamma^m_T}{T}\right|\right]\\
		& \leq  |\gamma_m - \gamma_{m,n}| + \E\left[\left|\gamma_{m,n} - \frac{\Gamma^{m,n}_T}{T}\right|\right] + \sup\limits_{T' \in [T_0,\infty)}\E\left[\left|\frac{\Gamma^{m,n}_{T'}}{T'} - \frac{\Gamma^m_{T'}}{T'}\right| \right].
	\end{align*} 
	By the ergodic property and the bounded convergence theorem $\lim_{T \to \infty} \E[|\gamma_{m,n}-\Gamma_{T}^{m,n}/T|] = 0$. Thus first sending $T \to \infty$, then $n \to \infty$ and finally $T_0 \to \infty$ above yields $\lim_{T \to \infty} \E[| \gamma_{m} -\Gamma_{T}^m/T|] = 0$ by \eqref{uniform_convergence}. In particular we obtain that $\Gamma_T^m/T \to \gamma_m$ in probability as $T \to \infty$. 
	
	To prove the growth rate property fix $\theta < \gamma$ and choose $m$ large enough so that $\theta < \gamma_m$. By definition of $\Gamma^m$ and the fact that $\Gamma$ is an increasing process we have that $\Gamma_T \geq \Gamma^m_T$ for every $T$. Hence we obtain
	\[\lim_{T\to \infty}\P\(\frac{\Gamma_T}{T} \geq \theta\) \geq \lim_{T\to \infty} \P\(\frac{\Gamma^m_T}{T} \geq \theta\) = 1.\]
	It follows that $g(V^{\pi_G};\P) \geq \gamma$. But by Lemma~\ref{worst_case_growth} we see that $\gamma = g(V^{\pi_G};\tilde \P)$ which completes the proof.
\end{proof}
\begin{remark} \label{rank_appendix_remark}
	In the proof of Theorem~\ref{rank_optimal} we use the fact that Theorem~\ref{concave_invariance_thm} continues to hold for measures $\P \in \tilde\Pi_{\geq}$ as long as the concave generating function $G$ is permutation invariant. Indeed by inspecting the proof of Lemma~\ref{uniform_lemma} we see that the lemma continues to hold for $\P \in \tilde \Pi_{\geq}$ as long as the functions $h_n^\top c h_n$, $|h_n|$ and $\sum_{i=1}^d (h_n^i)^{q'}$ are permutation invariant since for measures $\P \in \tilde \Pi_{\geq}$ we have that $X^{()}$ satisfies the ergodic property. Since the functions $\eta_m$ constructed in Lemma~\ref{cutoff_lemma} are permutation invariant and the functions $G_n$ appearing in the proof of Theorem~\ref{concave_invariance_thm} above can be chosen to be permutation invariant when $G$ is, it follows that the instances where Lemma~\ref{uniform_lemma} was used in the proof of Theorem~\ref{concave_invariance_thm} will continue to hold. Hence the result for the rank-based case follows in exactly the same way as the proof of Theorem~\ref{concave_invariance_thm}.
\end{remark}

\bibliographystyle{plain}
\bibliography{references}

\begin{thebibliography}{10}

\bibitem{alirezaei2018exponentially}
Gholamreza Alirezaei and Rudolf Mathar.
\newblock On exponentially concave functions and their impact in information
  theory.
\newblock In {\em 2018 Information Theory and Applications Workshop (ITA)},
  pages 1--10. IEEE, 2018.

\bibitem{banner2008short}
Adrian~D Banner and Daniel Fernholz.
\newblock Short-term relative arbitrage in volatility-stabilized markets.
\newblock {\em Ann. Finance}, 4(4):445--454, 2008.

\bibitem{banner2005atlas}
Adrian~D. Banner, Robert Fernholz, and Ioannis Karatzas.
\newblock Atlas models of equity markets.
\newblock {\em Ann. Appl. Probab.}, 15(4):2296--2330, 2005.

\bibitem{baxendale2019random}
Peter Baxendale and Ting-Kam~Leonard Wong.
\newblock Random concave functions.
\newblock {\em arXiv preprint arXiv:1910.13668}, 2019.

\bibitem{bayraktar2013robust}
Erhan Bayraktar and Yu-Jui Huang.
\newblock Robust maximization of asymptotic growth under covariance
  uncertainty.
\newblock {\em Ann. Appl. Probab.}, 23(5):1817--1840, 2013.

\bibitem{cover2011universal}
Thomas~M. Cover.
\newblock Universal portfolios.
\newblock {\em Math. Finance}, 1(1):1--29, 1991.

\bibitem{cuchiero2019polynomial}
Christa Cuchiero.
\newblock Polynomial processes in stochastic portfolio theory.
\newblock {\em Stochastic Process. Appl.}, 129(5):1829--1872, 2019.

\bibitem{cuchiero2019cover}
Christa Cuchiero, Walter Schachermayer, and Ting-Kam~Leonard Wong.
\newblock Cover's universal portfolio, stochastic portfolio theory, and the
  num\'{e}raire portfolio.
\newblock {\em Math. Finance}, 29(3):773--803, 2019.

\bibitem{fernholz2011optimal}
Daniel Fernholz and Ioannis Karatzas.
\newblock Optimal arbitrage under model uncertainty.
\newblock {\em Ann. Appl. Probab.}, 21(6):2191--2225, 2011.

\bibitem{fernholz2002stochastic}
E.~Robert Fernholz.
\newblock {\em Stochastic portfolio theory}, volume~48 of {\em Applications of
  Mathematics (New York)}.
\newblock Springer-Verlag, New York, 2002.
\newblock Stochastic Modelling and Applied Probability.

\bibitem{fernholz1999diversity}
Robert Fernholz.
\newblock On the diversity of equity markets.
\newblock {\em J. Math. Econom.}, 31(3):393--417, 1999.

\bibitem{fernholz2005relative}
Robert Fernholz and Ioannis Karatzas.
\newblock Relative arbitrage in volatility-stabilized markets.
\newblock {\em Ann. Finance}, 1(2):149--177, 2005.

\bibitem{ichiba2011hybrid}
Tomoyuki Ichiba, Vassilios Papathanakos, Adrian Banner, Ioannis Karatzas, and
  Robert Fernholz.
\newblock Hybrid atlas models.
\newblock {\em Ann. Appl. Probab.}, 21(2):609--644, 2011.

\bibitem{karatzas2020open}
Ioannis Karatzas and Donghan Kim.
\newblock Open markets.
\newblock {\em Math. Finance}, pages {1--52}, 2020.

\bibitem{karatzas2017trading}
Ioannis Karatzas and Johannes Ruf.
\newblock Trading strategies generated by {L}yapunov functions.
\newblock {\em Finance Stoch.}, 21(3):753--787, 2017.

\bibitem{karatzas1998brownian}
Ioannis Karatzas and Steven~E. Shreve.
\newblock Brownian motion and stochastic calculus.
\newblock volume 113 of {\em Graduate Texts in Mathematics}. Springer-Verlag,
  New York, second edition, 1991.

\bibitem{kardaras2018ergodic}
Constantinos Kardaras and Scott Robertson.
\newblock Ergodic robust maximization of asymptotic growth.
\newblock {\em Ann. Appl. Probab.}, Forthcoming.

\bibitem{kardaras2012robust}
Constantinos Kardaras and Scott Robertson.
\newblock Robust maximization of asymptotic growth.
\newblock {\em Ann. Appl. Probab.}, 22(4):1576--1610, 2012.

\bibitem{pal2016geometry}
Soumik Pal and Ting-Kam~Leonard Wong.
\newblock The geometry of relative arbitrage.
\newblock {\em Math. Financ. Econ.}, 10(3):263--293, 2016.

\bibitem{pickova2014generalized}
Radka Pickov\'{a}.
\newblock Generalized volatility-stabilized processes.
\newblock {\em Ann. Finance}, 10(1):101--125, 2014.

\bibitem{pinsky1995positive}
Ross~G. Pinsky.
\newblock {\em Positive harmonic functions and diffusion}, volume~45 of {\em
  Cambridge Studies in Advanced Mathematics}.
\newblock Cambridge University Press, Cambridge, 1995.

\bibitem{rockafellar1970convex}
R.~Tyrrell Rockafellar.
\newblock {\em Convex analysis}.
\newblock Princeton Mathematical Series, No. 28. Princeton University Press,
  Princeton, N.J., 1970.

\bibitem{schaefer1971locally}
H.~H. Schaefer and M.~P. Wolff.
\newblock Topological vector spaces.
\newblock volume~3 of {\em Graduate Texts in Mathematics}. Springer-Verlag, New
  York, second edition, 1999.

\end{thebibliography}

\end{document}